\def\llncs{0}
\def\fullpage{1}
\def\anonymous{0}
\def\authnote{1}
\def\notxfont{0}
\def\submission{0}
\def\cameraready{0}

\def\arxiv{0}
\ifnum\submission=1
\def\anonymous{1}
\def\llncs{1}
\else
\fi

\ifnum\submission=1
\def\llncs{1}
\def\authnote{0}
\def\anonymous{1}
\else
\fi

\ifnum\anonymous=1
\def\llncs{1}
\def\authnote{0}
\else
\fi

\ifnum\llncs=1
	\documentclass[envcountsect,a4paper,runningheads,10pt]{llncs}
\else
	\documentclass[letterpaper,hmargin=1.05in,vmargin=1.05in]{article}
			\ifnum\fullpage=1
		\usepackage{fullpage}
		\else
		\fi
\fi

\usepackage[%
  colorlinks=true,
  citecolor=blue,
  pagebackref=true
]{hyperref}

\usepackage{amsmath, amsfonts, amssymb, mathtools,amscd}

\usepackage{amsthm}

\usepackage{lmodern}
\usepackage[T1]{fontenc}
\usepackage[utf8]{inputenc}

\usepackage{arydshln} 
\usepackage{url}
\usepackage{ifthen}
\usepackage{bm}
\usepackage{multirow}
\usepackage[dvips]{graphicx}
\usepackage[usenames]{color}
\usepackage{xcolor,colortbl} 
\usepackage{threeparttable}
\usepackage{comment}
\usepackage{paralist,verbatim}
\usepackage{cases}
\usepackage{booktabs}
\usepackage{cancel} 
\usepackage{ascmac} 
\usepackage{framed}
\usepackage{authblk}
\usepackage{pifont}
\usepackage{autonum}
\usepackage{physics}
\definecolor{darkblue}{rgb}{0,0,0.6}
\definecolor{darkgreen}{rgb}{0,0.5,0}
\definecolor{maroon}{rgb}{0.5,0.1,0.1}
\definecolor{dpurple}{rgb}{0.2,0,0.65}
\usepackage[normalem]{ulem}

\usepackage[capitalise,noabbrev]{cleveref}
\usepackage[absolute]{textpos}
\usepackage[final]{microtype}
\usepackage[absolute]{textpos}
\usepackage{everypage}

\newtheoremstyle{thicktheorem}%
{\topsep}
{\topsep}
{\itshape}{}%
{\bfseries}%
{.}
{ }%
{\thmname{#1}\thmnumber{ #2}%
		\thmnote{ (#3)}%
}

\newtheoremstyle{remark}
{\topsep}
{\topsep}
	{}
	{}
	{}
	{.}
	{ }
	{\textit{\thmname{#1}}\thmnumber{ #2}
			\thmnote{ (#3)}%
	}

\ifnum\llncs=0
	\theoremstyle{thicktheorem}
	\newtheorem{theorem}{Theorem}[section]
	\newtheorem{lemma}[theorem]{Lemma}
	
	\newtheorem{proposition}[theorem]{Proposition}
	\newtheorem{definition}[theorem]{Definition}
	\newtheorem{game}[theorem]{Game}

	\theoremstyle{remark}
	
	\newtheorem{remark}[theorem]{Remark}

\else
\fi

	\crefname{theorem}{Theorem}{Theorems}
	\crefname{assumption}{Assumption}{Assumptions}
	\crefname{construction}{Construction}{Constructions}
	\crefname{corollary}{Corollary}{Corollaries}
	\crefname{conjecture}{Conjecture}{Conjectures}
	\crefname{definition}{Definition}{Definitions}
	\crefname{exmaple}{Example}{Examples}
	\crefname{experiment}{Experiment}{Experiments}
	\crefname{counterexample}{Counterexample}{Counterexamples}
	\crefname{lemma}{Lemma}{Lemmata}
	\crefname{observation}{Observation}{Observations}
	\crefname{proposition}{Proposition}{Propositions}
	\crefname{remark}{Remark}{Remarks}
	\crefname{claim}{Claim}{Claims}
	\crefname{fact}{Fact}{Facts}
	\crefname{note}{Note}{Notes}

\ifnum\llncs=1
 \crefname{appendix}{App.}{Appendices}
 \crefname{section}{Sec.}{Sections}
\else
\fi

\ifnum\llncs=1
\renewcommand*{\backref}[1]{}
\else
	\renewcommand*{\backref}[1]{(Cited on page~#1.)}
	\ifnum\notxfont=1
	\else
		\usepackage{newtxtext}
	\fi
\fi

\ifnum\authnote=0  
\newcommand{\mor}[1]{}
\newcommand{\taiga}[1]{}
\newcommand{\ryo}[1]{}
\newcommand{\takashi}[1]{}

\else
\newcommand{\mor}[1]{$\ll$\textsf{\color{red} Tomoyuki: { #1}}$\gg$}
\newcommand{\taiga}[1]{$\ll$\textsf{\color{magenta} Taiga: { #1}}$\gg$}
\newcommand{\takashi}[1]{$\ll$\textsf{\color{orange} Takashi: { #1}}$\gg$}
\newcommand{\ryo}[1]{$\ll$\textsf{\color{darkgreen} Ryo: { #1}}$\gg$}

\fi

\newcommand{\ot}{\mathsf{ot}}

\newcommand{\Open}{\mathsf{Open}}
\newcommand{\ONE}{\mathsf{ONE}}
\newcommand{\one}{\mathsf{one}}
\newcommand{\Modify}{\mathsf{Modify}}
\newcommand{\nad}{\mathsf{nad}}
\newcommand{\NAD}{\mathsf{NAD}}
\newcommand{\state}{\mathsf{st}}

\newcommand{\Garble}{\algo{Grbl}}

\newcommand{\GC}{\mathsf{GC}}
\newcommand{\gc}{\mathsf{gc}}

\newcommand{\Delete}{\algo{Del}}
\newcommand{\cert}{\keys{cert}}

\newcommand{\NCE}{\mathsf{NCE}}
\newcommand{\nce}{\mathsf{nce}}
\newcommand{\Fake}{\algo{Fake}}
\newcommand{\Reveal}{\algo{Reveal}}

\newcommand{\Sim}{\algo{Sim}}

\newcommand{\Samp}{\algo{Samp}}

\newcommand{\lrun}{\leftarrow}

\newcommand{\la}{\leftarrow}
\newcommand{\ra}{\rightarrow}

\newcommand{\seteq}{\coloneqq}

\newcommand{\cA}{\mathcal{A}}
\newcommand{\cB}{\mathcal{B}}

\newcommand{\cD}{\mathcal{D}}

\newcommand{\cM}{\mathcal{M}}

\newcommand{\cQ}{\mathcal{Q}}

\def\makeuppercase#1{
\expandafter\newcommand\csname tl#1\endcsname{\widetilde{#1}}
}

\def\makelowercase#1{
\expandafter\newcommand\csname tl#1\endcsname{\widetilde{#1}}
}

\newcommand{\N}{\mathbb{N}}

\newcommand{\R}{\mathbb{R}}

\newcommand{\Ms}{\mathcal{M}}

\newcommand{\Cs}{\mathcal{C}}
\newcommand{\Ks}{\mathcal{K}}

\newcommand{\secp}{\lambda}

\newcommand{\crs}{\mathsf{crs}}

\newcommand{\aux}{\mathsf{aux}}

\newcommand{\adva}[2]{\mathsf{Adv}_{#1}^{\mathsf{#2}}}
\newcommand{\advb}[3]{\mathsf{Adv}_{#1}^{\mathsf{#2} \mbox{-} \mathsf{#3}}}
\newcommand{\advc}[4]{\mathsf{Adv}_{#1}^{\mathsf{#2} \mbox{-} \mathsf{#3} \mbox{-} \mathsf{#4}}}
\newcommand{\advd}[5]{\mathsf{Adv}_{#1}^{\mathsf{#2} \mbox{-} \mathsf{#3} \mbox{-} \mathsf{#4} \mbox{-} \mathsf{#5}}}

\newcommand{\Expt}[2]{\mathsf{Exp}_{#1}^{#2}}

\newcommand{\expa}[3]{\mathsf{Exp}_{#1}^{ \mathsf{#2} \mbox{-} \mathsf{#3}}}
\newcommand{\expb}[4]{\mathsf{Exp}_{#1}^{ \mathsf{#2} \mbox{-} \mathsf{#3} \mbox{-} \mathsf{#4}}}

\newcommand{\expc}[4]{\mathsf{Exp}_{#1}^{ \mathsf{#2} \mbox{-} \mathsf{#3} \mbox{-} \mathsf{#4}}}
\newcommand{\expd}[5]{\mathsf{Exp}_{#1}^{\mathsf{#2} \mbox{-} \mathsf{#3} \mbox{-} \mathsf{#4} \mbox{-} \mathsf{#5}}}

\newcommand{\sfhyb}[2]{\mathsf{Hyb}_{#1}^{#2}}

\newcommand*{\sk}{\keys{sk}}
\newcommand*{\pk}{\keys{pk}}

\newcommand{\ct}{\keys{CT}}

\newcommand*{\MPK}{\keys{MPK}}
\newcommand*{\MSK}{\keys{MSK}}

\newcommand*{\vk}{\keys{vk}}

\newcommand*{\td}{\keys{td}}

\newcommand{\CT}{\keys{CT}}

\newcommand*{\keys}[1]{\mathsf{#1}}

\newcommand*{\algo}[1]{\ensuremath{\mathsf{#1}}}

\newcommand{\compclass}[1]{\textbf{\textrm{#1}}}

\newenvironment{boxfig}[2]{\begin{figure}[#1]\fbox{\begin{minipage}{0.97\linewidth}
                        \vspace{0.2em}
                        \makebox[0.025\linewidth]{}
                        \begin{minipage}{0.95\linewidth}
            {{
                        #2 }}
                        \end{minipage}
                        \vspace{0.2em}
                        \end{minipage}}}{\end{figure}}

\newcommand{\pprotocol}[4]{
\begin{boxfig}{h!}{\footnotesize 
\centering{\textbf{#1}}
    #4
\vspace{0.2em} } \caption{\label{#3} #2}
\end{boxfig}
}

\newcommand{\protocol}[4]{
\pprotocol{#1}{#2}{#3}{#4} }

\newcommand{\bit}{\{0,1\}}

\newcommand{\Setup}{\algo{Setup}}
\newcommand{\setup}{\algo{Setup}}

\newcommand{\keygen}{\algo{KeyGen}}

\newcommand{\Enc}{\algo{Enc}}
\newcommand{\Dec}{\algo{Dec}}

\newcommand{\Vrfy}{\algo{Vrfy}}

\newcommand{\Eval}{\algo{Eval}}

\newcommand\SKE{\algo{SKE}}

\newcommand{\ske}{\mathsf{ske}}

\newcommand\PKE{\algo{PKE}}

\newcommand{\pke}{\mathsf{pke}}

\newcommand{\negl}{{\mathsf{negl}}}

\newcommand{\poly}{{\mathrm{poly}}}

\newcommand{\Ppoly}{\compclass{P}/\compclass{poly}}
\newcommand{\NCone}{\compclass{NC}^1}

\usepackage{fancyhdr}

\ifnum\llncs=1
\let\oldvec\vec
\let\vec\oldvec
\makeatletter
\renewcommand*\l@author[2]{}
\renewcommand*\l@title[2]{}
\makeatletter

\else
\fi    

\theoremstyle{remark}

\title{
\textbf{Certified Everlasting Functional Encryption}
}

\begin{document}

\ifnum\anonymous=1
\author{\empty}\institute{\empty}
\else
%
%
\ifnum\llncs=1
\author{
	Taiga Hiroka\inst{1} \and Tomoyuki Morimae\inst{1,2} \and Ryo Nishimaki\inst{3} \and Takashi Yamakawa\inst{3}
}
\institute{
	Yukawa Institute for Theoretical Physics, Kyoto University, Japan \and PRESTO, JST, Japan \and NTT Secure Platform Laboratories
}
\else
%
%
\author[1]{Taiga Hiroka}
\author[1]{\hskip 1em Tomoyuki Morimae}
\author[2]{\hskip 1em Ryo Nishimaki}
\author[1,2]{\hskip 1em Takashi Yamakawa}
\affil[1]{{\small Yukawa Institute for Theoretical Physics, Kyoto University, Kyoto, Japan}\authorcr{\small \{taiga.hiroka,tomoyuki.morimae\}@yukawa.kyoto-u.ac.jp}}
\affil[2]{{\small NTT Corporation, Tokyo, Japan}\authorcr{\small \{ryo.nishimaki.zk,takashi.yamakawa.ga\}@hco.ntt.co.jp}}
\renewcommand\Authands{, }
\fi 
\fi

\ifnum\llncs=1
\date{}
\else
\date{\today}
\fi
\maketitle

\ifnum\anonymous=1
\else
\thispagestyle{fancy}
\rhead{YITP-22-73}
\fi

\begin{abstract}
Computational security in cryptography has a risk that computational assumptions underlying the security are broken in the future.
One solution is to construct information-theoretically-secure protocols, but
many cryptographic primitives are known to be impossible (or unlikely) to have information-theoretical security even in the quantum world.
A nice compromise (intrinsic to quantum) is certified everlasting security, which roughly means the following.
A receiver with possession of quantum encrypted data can issue a certificate that shows that the receiver has deleted the encrypted data. 
If the certificate is valid, the security is guaranteed even if the receiver becomes computationally unbounded. 
Although several cryptographic primitives, such as commitments and zero-knowledge, have been made certified everlasting secure,
there are many other important primitives that are not known to be certified everlasting secure.

In this paper, we introduce certified everlasting FE.
In this primitive, the receiver with the ciphertext of a message $m$ and the functional decryption key of a function $f$ can obtain $f(m)$ and nothing else.
The security holds even if the adversary becomes computationally unbounded after issuing a valid certificate.
We, first, construct certified everlasting FE for $\Ppoly$ circuits where only a single key query is allowed for the adversary. 
We, then, extend it to $q$-bounded one for $\NCone$ circuits where $q$-bounded means that
$q$ key queries are allowed for the adversary with an a priori bounded polynomial $q$.  
For the construction of certified everlasting FE, we introduce and construct
certified everlasting versions of secret-key encryption, public-key encryption, receiver non-committing encryption, and a garbling scheme,
which are of independent interest.
\end{abstract}

\ifnum\cameraready=1
\else
\ifnum\llncs=1
\else
\newpage
  \setcounter{tocdepth}{2}      
  \setcounter{secnumdepth}{2}   
  \setcounter{page}{0}          
  \tableofcontents
  \thispagestyle{empty}
  \clearpage
\fi
\fi

\section{Introduction}
\subsection{Background}
Computational security in cryptography relies on assumptions that some problems are hard to solve.
It, however, has a risk that the assumptions could be broken in a future
when revolutionary novel algorithms are discovered or computing devices are drastically improved.
One solution to the problem is to construct information-theoretically-secure protocols~\cite{Shamir79,BB84},
but even in the quantum world, many cryptographic primitives are known to be impossible (or unlikely) 
to have information-theoretical security~\cite{LC97,Mayers97,MW18}.

Good compromises (intrinsic to quantum!) have been studied recently~\cite{JACM:Unruh15,TCC:BroIsl20,kundutan,Asia:HMNY21,EPRINT:HMNY21_2,Poremba22}. 
In particular, certified everlasting security, which was introduced in~\cite{EPRINT:HMNY21_2} based on
~\cite{JACM:Unruh15,TCC:BroIsl20}, achieves the following security: 
a receiver with possession of quantum encrypted date issues a certificate which shows that the receiver has deleted its quantum encrypted data.
If the certificate is valid, the security is guaranteed even if the receiver becomes computationally unbounded later (and even if some secret information like the secret key is leaked).
This security notion is weaker than the information-theoretical security. (For example, a malicious receiver may refuse to issue a valid certificate.)
It is, however, still a useful security notion, because, for example, a sender can penalize receivers who do not issue valid certificates.
Moreover, certified everlasting security is intrinsically quantum property, because it implies information-theoretical security in the classical world. \footnote{This is because a malicious receiver can copy the encrypted data freely, and thus the encrypted data must be secure against an unbounded malicious receiver at the point when the receiver obtains the encrypted data.
On the other hand, in the quantum world, the same discussion does not go through, because even a malicious receiver cannot copy the quantum encrypted data due to the quantum no-cloning theorem.}

Certified everlasting security can bypass the impossibility of information-theoretical security.
In fact, several cryptographic primitives have been shown to have certified everlasting security, such as commitments and zero-knowledge~\cite{EPRINT:HMNY21_2}. 
An important open problem in this direction is
\begin{center}
\textit{
Which cryptographic primitives can have certified everlasting security?
}
\end{center}

Functional encryption (FE) is one of the most advanced cryptographic primitives that achieves large flexibility in controlling encrypted data~\cite{BSW11}.
In FE, an owner of a master secret key $\MSK$ can generate a functional decryption key $\sk_f$ that hardwires a function $f$.
When a ciphertext $\ct(m)$ of a message $m$ is decrypted by $\sk_f$, we obtain $f(m)$. No information beyond $f(m)$ is obtained.
Information-theoretically-secure FE seems to be unlikely, and in fact
all known constructions are computationally secure ones~\cite{C:GorVaiWee12,FOCS:GGHRSW13,TCC:GGHZ16,EC:AnaSah17,C:LinTes17,C:AJLMS19,Eprint:AnaJaiSah18,EPRINT:Agrawal18,EPRINT:LinMat18}.
Hence we have the following open problem:
\begin{center}
\textit{
Is it possible to construct certified everlasting secure FE?
}
\end{center}
We remark that certified everlasting FE is particularly useful compared with certified everlasting public key encryption (PKE) (or more generally ``all-or-nothing encryption''~\cite{C:GarMahMoh17} such as identity-based encryption (IBE), attribute-based encryption (ABE), or witness encryption (WE)) because it ensures security even against an honest receiver who holds a decryption key. That is, we can ensure that a receiver who holds a decryption key $\sk_f$ with respect to a function $f$ cannot learn more than $f(m)$ even if the receiver can run unbounded-time computation after issuing a valid certificate. 
In contrast, certified everlasting PKE does not ensure any security against an honest receiver since the receiver can simply copy an encrypted message after honestly decrypting a ciphertext and then no security remains.

\subsection{Our Results}\label{sec:result}
We partially solve the above questions affirmatively.
Our contributions are as follows:
    \begin{enumerate}
     \item We formally define certified everlasting versions of secret-key encryption (SKE) (\cref{sec:def_ske}), 
     public-key encryption (PKE) (\cref{sec:def_pke}), receiver non-committing encryption (RNCE) (\cref{sec:def_rnce}), 
     a garbling scheme (\cref{sec:def_garbled}), 
     and FE (\cref{sec:def_ever_fe}), respectively.
     \item We present two constructions of certified everlasting SKE (resp. PKE).
     An advantage of the first construction is that the certificate is classical, but a disadvantage is that the security proof relies on the quantum random oracle model (QROM)~\cite{AC:BDFLSZ11}.
     On the other hand, in the second construction, the security holds without relying on the QROM, but the certificate is quantum. 
     \item We construct certified everlasting RNCE from certified everlasting PKE in a black-box way (\cref{sec:const_rnce_classic}).
     \item We construct a certified everlasting garbling scheme for all $\Ppoly$ circuits from certified everlasting SKE in a black-box way (\cref{sec:const_garbling}).
     \item We construct 1-bounded certified everlasting FE with adaptive security for all $\Ppoly$ circuits. The adaptive security means that the adversary can call key queries before and after seeing the challenge ciphertext. 
     The 1-bounded means that only a single key query is allowed for the adversary.
     The construction is done in the following two steps.
     First, we construct 1-bounded certified everlasting FE with {\it non-adaptive} security for all $\Ppoly$ circuits from a certified everlasting garbling scheme and certified everlasting PKE in a black-box way
     (\cref{sec:const_func_non_adapt}).
     Second, we change it to the {\it adaptively-secure} one by using certified everlasting RNCE in a black-box way (\cref{sec:const_fe_adapt}).
     \label{step:one_FE}
     \item We construct $q$-bounded certified everlasting FE with adaptive security for all $\NCone$ circuits, where $q$-bounded means that
     the total number of key queries is bounded by an a priori fixed polynomial $q$.
     This is constructed from
     the 1-bounded one constructed in Step.~\ref{step:one_FE} by using multi-party computation in a black-box way (\cref{sec:const_multi_fe}).
     \end{enumerate}

\subsection{Related Works}

Unruh~\cite{JACM:Unruh15} introduced the concept of revocable quantum time-released encryption.
In this primitive, a receiver with possession of quantum encrypted data can obtain its plaintext after predetermined time $T$.
The sender can revoke the quantum encrypted data before time $T$. 
If the revocation succeeds, the receiver cannot obtain the information of the plaintext even if its computing power becomes unbounded.

Broadbent and Islam~\cite{TCC:BroIsl20} constructed one-time SKE with certified deletion.
This is ordinary one-time SKE, but once the receiver issues a valid classical certificate, the receiver cannot obtain the information of plaintext even if it later obtains the secret key of the ciphertext. 
(See also \cite{kundutan}).

Hiroka, Morimae, Nishimaki, and Yamakawa~\cite{Asia:HMNY21} constructed reusable SKE, PKE, and attribute-based encryption (ABE) with certified deletion.
These reusable SKE, PKE, and ABE with certified deletion are ordinary reusable SKE, PKE, and ABE, respectively. However, once the receiver issues a valid classical certificate, the receiver cannot obtain the information of plaintext even if it obtains some secret information (e.g. the master secret key of ABE). Note that, in these primitives, the security holds against computationally bounded adversaries unlike the present paper.

Hiroka, Morimae, Nishimaki, and Yamakawa~\cite{EPRINT:HMNY21_2} constructed commitments with statistical binding and certified everlasting hiding. From it, they also constructed certified everlasting zero-knowledge proof for 
QMA based on the zero-knowledge protocol of~\cite{FOCS:BroGri20}.

Poremba~\cite{Poremba22} constructed fully homomorphic encryption (FHE) with certified deletion 
where the security holds against only semi-honest adversaries that behaves maliciously only after outputting a certificate.

\subsection{Concurrent and Independent Work}
There is a concurrent and independent work.
Recently, Bartusek and Khurana have uploaded their paper on arXiv~\cite{BK22}
where a generic compiler is introduced.
The generic compiler can change many cryptographic primitives into ones with certified deletion, 
such as PKE, ABE, FHE, witness encryption, and timed-release encryption.

Their constructions via the generic compiler achieve classical certificates without QROM, which is an advantage of their results.
On the other hand, we construct certified everlasting garbling schemes and FE, which is not done in their work.
In fact, it is not clear how to construct certified everlasting garbling schemes and certified everlasting FE via their generic compiler.
For example, a natural construction of FE via their generic compiler would be as follows.
The ciphertext consists of a classical part and a quantum part. The classical part is the ciphertext of ordinary FE whose plaintext is $m\oplus r$, and the quantum part is random BB84 states whose computational basis states encode $r$.
The decryption key of the function $f$ consists of a functional decryption key $\sk_f$ and the basis of the BB84 states.
However, in this construction, a receiver with the ciphertext and the decryption key cannot obtain $f(m)$.
This is because the receiver can obtain only $f(m\oplus r)$ and $r$, which cannot recover $f(m)$.

\def\techoverview{0}
\ifnum\techoverview=1
\subsection{A Technical Overview}
We now provide an overview of main ideas and steps to obtain our results.
To explain our idea, we introduce the definition of certified everlasting FE.
\paragraph{Definition of Certified Everlasting FE:}Certified Everlasting FE consists of the following algorithms.
\begin{description}
\item[$\setup(1^\secp)\ra(\MPK,\MSK)$:]
This is a setup algorithm that generates a key pair of master public and secret keys.
\item[$\keygen(\MSK,f)\ra \sk_f$:] This is a keygeneration algorithm that generates functional secret key. 
\item[$\Enc(\MPK,m)\ra (\vk,\ct)$:]This is an encryption algorithm that generates a quantum ciphertext and a verification key for the ciphertext.
\item[$\Dec(\sk_f,\ct)\ra y$]
This is a decryption algorithm that outputs $y$. As correctness, we require that it outputs $f(m)$ if it takes as input the functional decryption key $\sk_f$ of $f$ and the ciphertext $\ct$ of $m$.
\item[$\Delete(\ct)\ra\cert$:] This is a deletion algorithm that outputs a certificate, which guarantees that the receiver deletes the quantum ciphertext.
\item[$\Vrfy(\vk,\cert)\ra \top\,\,{\bf or}\,\,\bot$:]
This is a verification algorithm that verifies the validity of $\cert$ using a verification key. As correctness, we require that it outputs $\top$ if $\cert$ was honestly generated by $\cert\la\Delete(\ct)$, and $(\vk,\ct)$ was honestly generated by running $\Enc$. 
\end{description}
We require that the scheme satisfies certified everlasting non-adaptive or adaptive security in addition to ordinary non-adaptive security or adaptive security.
Roughly speaking, certified everlasting adaptive security requires that no quantum polynomial time (QPT) adversaries with $\MPK$, $\ct$, and the ability to call the limited number of times of key queries cannot obtain any information of plaintext even if its computing power becomes unbounded and receives $\MPK$ after it outputs a valid certificate.
If the QPT adversary is allowed to call key queries only before it receives the challenger ciphertext, we call the security certified everlasting non-adaptive security.

\paragraph{Construction of $q$-bounded Certified Everlasting FE:}
In a nutshell, our main idea is to construct certified everlasting FE based on the construction of ordinary $q$-bounded computationally secure FE.

It is known that we can construct ordinary FE with bounded-collusion for all $\NCone$ circuits from ordinary PKE, ordinary RNCE, and an ordinary garbling scheme.
More concretely, this is achieved in the following three steps.
First, we can construct single-key FE with non-adaptive security by running an encryption algorithm of PKE on the labels of a garbling scheme, where the non-adaptive security means that the adversary can call key queries only before it receives the challenge ciphertext.
Then, we can transform it to FE with adaptive security by running the RNCE encryption algorithm on the ciphertext of FE with non-adaptive security.
Finally, by combining it with multi-party computation technique, we can construct $q$-bounded FE with adaptive security for all $\NCone$ circuits presented as in~\cite{C:GorVaiWee12}.

Therefore, naively thinking, we can construct $q$-bounded certified everlasting FE with adaptive security for $\NCone$ circuits from certified everlasting PKE, certified everlasting RNCE, and a certified everlasting garbling scheme as follows.
First, it seems possible to construct single-key certified everlasting FE with non-adaptive security by running an encryption algorithm of certified everlasting PKE on labels of a certified everlasting garbling scheme.
Then, it seems possible to transform it to single-key certified everlasting FE with adaptive security by running an encryption algorithm of certified everlasting RNCE on the ciphertext of single-key certified everlasting FE with non-adaptive security.
Finally, it seems possible to transform it to $q$-bounded certified everlasting FE with adaptive security for all $\NCone$ circuits by using the technique of multi-party computation.
In fact, this naive idea works to some extent although there remains two problems to resolve.

The first obvious problem is that it is not known how to construct even certified everlasting PKE, certified everlasting RNCE, and a certified everlasting garbling scheme, which seems non-trivial in itself to construct these primitives.
(We explain the details of our constructions in the other paragraphs.)

The second problem is that we cannot construct single-key certified everlasting FE with adaptive security by running the encryption algorithm of certified everlasting RNCE on the ciphertext of certified everlasting FE with non-adaptive security because of the following two reasons.
First, if we encrypt the ciphertext of certified everlasting FE in some way, the sender cannot check the validity of certificate.
This is because, in certified everlasting FE, the sender can check the validity of certificate $\cert$ only if it was issued by running $\cert\la\Delete(\ct)$.
If we encrypt the ciphertext of FE in some way, it disturbs the quantum ciphertext $\ct$, and thus the certificate $\cert$ is also disturbed.
Therefore, we have to find the way where the
sender can check the validity of $\cert$ even if the certificate is issued by using the encrypted ciphertext.
Second, in certified everlasting FE, the ciphertext must become quantum state although the encryption algorithm of our certified everlasting RNCE is for a classical message not for a quantum message as that of ordinary RNCE is. So we have to find the way where we don't need to run an encryption algorithm for classical message on the ciphertext of certified everlasting FE.

In our work, we present a construction which resolves these problems.
In our construction, we use single-key certified everlasting FE with non-adaptive security $\NAD.(\Setup,\keygen,\Enc,\Dec,\Delete,\Vrfy)$ and certified everlasting RNCE for classical message $\NCE.(\Setup,\keygen,\Enc,\Dec,\Delete,\Vrfy)$.
Here, we require that the FE scheme satisfies additional correctness, which we call modification correctness.
In our construction, this correctness is necessary for the sender to check the validity of certificate.
This correctness guarantees that there exists a quantum algorithm $\Modify$ that takes as input classical bit strings $(a,c)$ and $\cert^*$, and outputs a valid certificate where $\cert^*$ was issued by $\cert^*\la \NAD.\Delete(Z^cX^a\nad.\ct X^aZ^c)$ and $\nad.\ct$ is generated by running $\NAD.\Enc$ honestly.

Our single-key certified everlasting FE with adaptive security is described as follows.
\begin{description}
\item[$\Setup(1^\secp)$:]
This generates $(\nce.\pk,\nce.\MSK)\la\NCE.\Setup(1^\secp)$ and $(\nad.\MPK,\nad.\MSK)\la\mathsf{NAD}.\Setup(1^\secp)$, and outputs $\MPK\seteq(\nce.\pk,\nad.\MPK)$ and $\MSK\seteq (\nce.\MSK,\nad.\MSK)$.
\item[$\keygen(\MSK,f)$:]
This generates $\nce.\sk\la\NCE.\keygen(\nce.\MSK)$ and $\nad.\sk_f\la\NAD.\keygen(\nad.\MSK,f)$, and outputs $\sk_f\seteq (\nce.\sk,\nad.\sk_f)$.
\item[$\Enc(\MPK,m)$:]
This generates classical random bit strings $(a,c)$, $(\nad.\vk,\nad.\ct)\la\NAD.\Enc(\nad.\MPK,m)$ $(\nce.\vk,\nce.\ct)\la\NCE.\Enc(\nce.\pk,(a,c))$ and $\Psi\seteq X^aZ^c\nad.\ct Z^cX^a$, and outputs $\vk\seteq ((a,c),\nce.\vk,\nad.\vk)$ and $\ct\seteq(\nce.\ct,\Psi)$.
\item[$\Dec(\sk_f,\ct)$:]
This computes $(a,c)\la\NCE.\Dec(\nce.\sk,\nce.\ct)$, $\nad.\ct^*\seteq Z^cX^a\Psi X^aZ^c$, and $y\la\NAD.\Dec(\nad.\sk_f,\nad.\ct^*)$, and outputs $y$.
\item[$\Delete(\ct)$:]
This generates $\nce.\cert\la\NCE.\Delete(\nce.\ct)$ and $\nad.\cert\la\NAD.\Delete(\Psi)$, and outputs $\cert\seteq (\nce.\cert,\nad.\cert)$.
\item[$\Vrfy(\vk,\cert)$:]
First, this runs $\nad.\cert^*\la\NAD.\Modify((a,c),\nad.\cert)$.
This outputs $\top$ if $\top\la\NCE.\Vrfy(\nce.\vk,\nce.\cert)$ and $\top\la\NAD.\Vrfy(\nad.\vk,\nad.\cert^*)$.
\end{description}

In our construction, the sender can check the validity of the certificate since it can recover the correct certificate by running $\NAD.\Modify((a,c),\nad.\cert)$, and $\NCE.\Enc$, which is an encryption algorithm for a classical message, does not run on $\NAD.\ct$. Therefore, our construction avoids the problems above.

Furthermore, we can prove certified everlasting adaptive security by reducing it to certified everlasting non-adaptive security, RNC security, and using quantum teleportation in the sequence of hybrid experiments.
To prove certified everlasting adaptive security, we have to simulate the challenge ciphertext and the functional secret key.
More formally, when the challenge query of the message $m$ is called, we have to simulate the ciphertext of $m$ without using $\NAD.\ct$ in an indistinguishable way.
In addition, when a key query of function $f$ is called, we have to simulate the secret key using $\NAD.\ct$ in an indistinguishable way.

\paragraph{Certified Everlasting PKE:}

\paragraph{Certified Everlasting RNCE:}

\paragraph{Certified Everlasting Garbling Scheme:}

\paragraph{Certified Everlasting FE:}

\color{black}
\else\fi


\section{Preliminaries}\label{sec:preliminaries}

\subsection{Notations}\label{sec:notations}
Here we introduce basic notations we will use
in this paper. $x\leftarrow X$ denotes selecting an element $x$ from a finite set $X$ uniformly at random, 
and $y\leftarrow A(x)$ denotes assigning to $y$ the output of a quantum or probabilistic or deterministic algorithm $A$ on an input $x$.
When we explicitly show that $A$ uses randomness $r$, we write $y\leftarrow A(x;r)$.
When $D$ is a distribution, $x\leftarrow D$ denotes sampling an element $x$ from $D$.
$y\seteq z$ denotes that $y$ is set, defined, or substituted by $z$.
Let $[n]\seteq \{1,\dots,n\}$. Let $\lambda$ be a security parameter.
By $[N]_p$ we denote the set of all size-$p$ subsets of $\{1,2\cdots, N\}$.
For classical strings $x$ and $y$, $x||y$ denotes the concatenation of $x$ and $y$.
For a bit string $s\in\{0,1\}^n$, $s_i$ and $s[i]$ denotes the $i$-th bit of $s$.
QPT stands for quantum polynomial time.
PPT stands for (classical) probabilistic polynomial time.
A function $f: \N \ra \R$ is a negligible function if for any constant $c$, there exists $\secp_0 \in \N$ such that for any $\secp>\secp_0$, $f(\secp) < \secp^{-c}$. We write $f(\secp) \leq \negl(\secp)$ to denote $f(\secp)$ being a negligible function.

\subsection{Quantum Computations}\label{sec:quantum_computation}
We assume familiarity with the basics of quantum computation and use standard notations. 
Let $\cQ$ be the state space of a single qubit.
$I$ is the two-dimensional identity operator.
$X$ and $Z$ are the Pauli $X$ and $Z$ operators, respectively.
For an operator $A$ acting on a single qubit and a bit string $x\in\bit^n$, we write $A^x$ as $A^{x_1}\otimes A^{x_2}\otimes \cdots A^{x_n}$.
The trace distance between two states $\rho$ and $\sigma$ is given by 
$\frac{1}{2}\norm{\rho-\sigma}_{\tr}$, where $\norm{A}_{\tr}\seteq \tr \sqrt{{\it A}^{\dagger}{\it A}}$ is the trace norm. 
If $\frac{1}{2}\norm{\rho-\sigma}_{\tr}\leq \epsilon$, we say that $\rho$ and $\sigma$ are $\epsilon$-close.
If $\epsilon=\negl(\lambda)$, then we say that $\rho$ and $\sigma$ are statistically indistinguishable.
\paragraph{Quantum Random Oracle.}
We use the quantum random oracle model (QROM)~\cite{AC:BDFLSZ11} to construct certified everlasting SKE and certified everlasting PKE
in \cref{sec:const_ske_rom,sec:const_pke_rom}, respectively.
In the QROM, a uniformly random function with a certain domain and range is chosen at the beginning, and quantum access to this function is given to all parties including an adversary.
Zhandry showed that quantum access to random functions can be efficiently simulatable by using so-called compressed random oracle technique~\cite{C:Zhandry19}.

We review the one-way to hiding lemma~\cite{JACM:Unruh15,C:AmbHamUnr19}, which is useful when analyzing schemes in the QROM. The following form of the lemma is based on~\cite{C:AmbHamUnr19}.

\begin{lemma}[One-Way to Hiding Lemma \cite{C:AmbHamUnr19}]\label{lem:one-way_to_hiding}
Let $S\subseteq \mathcal{X}$ be a random subset of $\mathcal{X}$. Let $G,H:\mathcal{X}\rightarrow\mathcal{Y}$ be random functions satisfying $\forall x\notin S$ $[G(x)=H(x)]$. Let $z$ be a random classical bit string. 
($S,G,H,z$ may have an arbitrary joint distribution.)
Let $\cA$ be an oracle-aided quantum algorithm that makes at most $q$ quantum queries.
Let $\cB$ be an algorithm that on input $z$ chooses $i\leftarrow[q]$, runs $\cA^{H}(z)$, measures $\cA$'s $i$-th query, and outputs the measurement outcome.
Then we have
$
    \abs{\Pr[\cA^G(z)=1]-\Pr[\cA^H(z)=1]}\leq2q\sqrt{\Pr[\cB^H(z)\in S]}.
$
\end{lemma}

\paragraph{Quantum Teleportation.}
We use quantum teleportation to prove that our construction of the FE scheme in \cref{sec:const_fe_adapt} satisfies adaptive security.
\begin{lemma}[Quantum Teleportation]\label{lem:quantum_teleportation}
Suppose that we have $N$ Bell pairs between registers $A$ and $B$, i.e.,$\frac{1}{\sqrt{2^N}}\sum_{s\in\bit^N}\allowbreak \ket{s}_A \otimes \ket{s}_B$, and let $\rho$ be an arbitrary $N$-qubit quantum state in register $C$. Suppose that we measure $j$-th qubits of $C$ and $A$ in the Bell basis and let $(x_j,z_j)\in\bit\times\bit$ be the measurement outcome for all $j\in[N]$.
Let $x\seteq x_1||x_2||\cdots ||x_N$ and $z\seteq z_1||z_2||\cdots ||z_N$.
Then $(x,z)$ is uniformly distributed  over $\bit^N\times \bit^N$.
Moreover, conditioned on the measurement outcome $(x,z)$, the resulting state in $B$ is $X^xZ^z\rho Z^zX^x$.
\end{lemma}

\paragraph{CSS code.}
We explain basics of CSS codes. CSS codes are used only in the constructions of certified everlasting SKE and PKE (\cref{sec:const_ske_wo_rom} and \cref{sec:const_pke_wo_rom}),
and therefore readers who are not interested in these constructions can skip this paragraph.
A CSS code with parameters $q,k_1,k_2, t$ consists of two classical linear binary codes. One is a $[q,k_1]$ code $C_1$
\footnote{A $[q,k]$ code is a code consisting of $2^k$ codewords, each of length $q$.
That is, a $k$-dimensional subspace of $\bit^q={\rm GF}(2)^q$.
}
and the other is a $[q,k_2]$ code. 
Both $C_1$ and $C_2^\bot$ can correct up to $t$ errors, 
and they satisfy $C_2\subseteq C_1$.
We require that the parity check matrices of $C_1,C_2$ are computable in polynomial time, and that error correction can be performed in polynomial time. 
Given two binary codes $C\subseteq D$, let $D/C\seteq\{x \mbox{ mod } C:x\in D\}$.
Here, mod $C$ is a linear polynomial-time operation on $\bit^q$ with the following three properties.
First, $x$ mod $C=x'$ mod $C$ if and only if $x-x'\in C$ for any $x,x'\in\bit^q$.
Second, for any binary code $D$ such that $C\subseteq D$,
$x$ mod $C\in D$ for any $x\in D$.
Third,  ($x$ mod $C$) mod $C$= $x$ mod $C$ for any $x\in\bit^q$.

\subsection{Cryptographic Tools}\label{sec:crypt_tool}
In this section, we review the cryptographic tools used in this paper.

\begin{lemma}[Difference Lemma~\cite{EPRINT:Shoup04}]\label{lem:defference}
Let $A,B,F$ be events defined in some probability distribution, and suppose 
$\Pr[A\wedge \overline{F}]=\Pr[B\wedge \overline{F}]$.
Then $\abs{\Pr[A]-\Pr[B]}\leq \Pr[F]$.
\end{lemma}

\paragraph{Encryption with Certified Deletion.}
Broadbent and Islam~\cite{TCC:BroIsl20} introduced the notion of encryption with certified deletion.

\begin{definition}[One-Time SKE with Certified Deletion (Syntax)~\cite{TCC:BroIsl20,Asia:HMNY21}]\label{def:sk_cert_del}
Let $\lambda$ be a security parameter and let $p$, $q$ and $r$ be some polynomials.
A one-time secret key encryption scheme with certified deletion is a tuple of algorithms $\Sigma=(\keygen,\Enc,\Dec,\Delete,\Vrfy)$ with 
plaintext space $\Ms:=\{0,1\}^n$, ciphertext space $\Cs:= \cQ^{\otimes p(\lambda)}$, key space $\Ks:=\{0,1\}^{q(\lambda)}$ and deletion certificate space $\mathcal{D}:= \{0,1\}^{r(\lambda)}$.
\begin{description}
    \item[$\keygen (1^\secp) \ra \sk$:] The key generation algorithm takes as input the security parameter $1^\secp$, and outputs a secret key $\sk \in \Ks$.
    \item[$\Enc(\sk,m) \ra \ct$:] The encryption algorithm takes as input $\sk$ and a plaintext $m\in\Ms$, and outputs a ciphertext $\ct\in \Cs$.
    \item[$\Dec(\sk,\ct) \ra m^\prime~or~\bot$:] The decryption algorithm takes as input $\sk$ and $\ct$, and outputs a plaintext $m^\prime \in \Ms$ or $\bot$.
    \item[$\Delete(\ct) \ra \cert$:] The deletion algorithm takes as input $\ct$, and outputs a certification $\cert\in\mathcal{D}$.
    \item[$\Vrfy(\sk,\cert)\ra \top~or~\bot$:] The verification algorithm takes $\sk$ and $\cert$ as input, and outputs $\top$ or $\bot$.
\end{description}
\end{definition}

We require that a one-time SKE scheme with certified deletion satisfies correctness defined below.
\begin{definition}[Correctness for One-Time SKE with Certified Deletion]\label{def:correctness_sk_cert_del}
There are three types of correctness, namely,
decryption correctness,
verification correctness,
and modification correctness.
\paragraph{Decryption Correctness:} There exists a negligible function $\negl$ such that for any $\secp\in \N$ and $m\in\Ms$, 
\begin{align}
\Pr\left[
m'\neq m
\ \middle |
\begin{array}{ll}
\sk\lrun \keygen(1^\secp)\\
\ct \lrun \Enc(\sk,m)\\
 m'\la\Dec(\sk,\ct)
\end{array}
\right] 
\leq\negl(\secp).
\end{align}

\paragraph{Verification Correctness:} There exists a negligible function $\negl$ such that for any $\secp\in \N$ and $m\in\Ms$, 
\begin{align}
\Pr\left[
\Vrfy(\sk,\cert)=\bot
\ \middle |
\begin{array}{ll}
\sk\lrun \keygen(1^\secp)\\
\ct \lrun \Enc(\sk,m)\\
\cert \lrun \Delete(\ct)
\end{array}
\right] 
\leq
\negl(\secp).
\end{align}

\paragraph{Modification Correctness:}
There exists a negligible function $\negl$ and a QPT algorithm $\Modify$ such that for any $\secp\in \N$ and $m\in\Ms$, 
\begin{align}
\Pr\left[
\Vrfy(\sk,\cert^*)=\bot
\ \middle |
\begin{array}{ll}
\sk\lrun \keygen(1^\secp)\\
\ct \lrun \Enc(\sk,m)\\
a,b\la\bit^{p(\lambda)}\\
\cert \lrun \Delete(Z^bX^a\ct X^aZ^b)\\
\cert^*\lrun \Modify(a,b,\cert)  
\end{array}
\right] 
\leq
\negl(\secp).
\end{align}

\end{definition}

\begin{remark}
The original definition~\cite{TCC:BroIsl20,Asia:HMNY21} only considers decryption correctness and verification correctness.
In this paper, we additionally require modification correctness.
This is because we need modification correctness for the construction of FE in \cref{sec:const_fe_adapt}.
In fact, the construction of ~\cite{TCC:BroIsl20} satisfies modification correctness as well.
\end{remark}

We require that a one-time SKE with certified deletion satisfies certified deletion security defined below.
\begin{definition}[Certified Deletion Security for One-Time SKE with Certified Deletion]\label{def:security_sk_cert_del}
Let $\Sigma=(\keygen,\Enc,\Dec,\allowbreak \Delete,\Vrfy)$ be a one-time SKE scheme with certified deletion.
We consider the following security experiment $\expb{\Sigma,\cA}{otsk}{cert}{del}(\secp,b)$ against an unbounded adversary $\cA$.

\begin{enumerate}
    \item The challenger computes $\sk \la \keygen(1^\secp)$.
    \item $\cA$ sends $(m_0,m_1)\in\cM^2$ to the challenger.
    \item The challenger computes $\ct \la \Enc(\sk,m_b)$ and sends $\ct$ to $\cA$.
    \item $\cA$ sends $\cert$ to the challenger.
    \item The challenger computes $\Vrfy(\sk,\cert)$. If the output is $\bot$, the challenger sends $\bot$ to $\cA$.
    If the output is $\top$, the challenger sends $\sk$ to $\cA$. 
    \item $\cA$ outputs $b'\in \bit$. This is the output of the experiment.
\end{enumerate}
We say that $\Sigma$ is OT-CD secure if, for any unbounded $\cA$, it holds that
\begin{align}
\advc{\Sigma,\cA}{otsk}{cert}{del}(\secp)
\seteq \abs{\Pr[ \expb{\Sigma,\cA}{otsk}{cert}{del}(\secp, 0)=1] - \Pr[ \expb{\Sigma,\cA}{otsk}{cert}{del}(\secp, 1)=1] }\leq \negl(\secp).
\end{align}
\end{definition}

Broadbent and Islam~\cite{TCC:BroIsl20} showed that a one-time SKE scheme with certified deletion that satisfies the above correctness and security exists unconditionally.

\paragraph{Secret Key Encryption (SKE).}
\begin{definition}[Secret Key Encryption (Syntax)]\label{def:ske}
Let $\lambda$ be a security parameter and let $p$, $q$, $r$ and $s$ be some polynomials.
A secret key encryption scheme is a tuple of algorithms
$\Sigma=(\keygen,\Enc,\Dec)$ with plaintext space $\Ms:=\{0,1\}^n$, ciphertext space $\Cs:= \bit^{p(\lambda)}$,
and secret key space $\mathcal{SK}:=\{0,1\}^{q(\lambda)}$.
\begin{description}
    \item[$\keygen (1^\secp) \ra \sk$:]
    The key generation algorithm takes the security parameter $1^\secp$ as input and outputs a secret key $\sk \in \mathcal{SK}$.
    \item[$\Enc(\sk,m) \ra \ct$:]
    The encryption algorithm takes $\sk$ and a plaintext $m\in\Ms$ as input, and outputs a ciphertext $\ct\in \Cs$.
    \item[$\Dec(\sk,\ct) \ra m^\prime~or~\bot$:]
    The decryption algorithm takes $\sk$ and $\ct$ as input, and outputs a plaintext $m^\prime \in \Ms$ or $\bot$.
\end{description}
\end{definition}

We require that a SKE scheme satisfies correctness defined below.
\begin{definition}[Correctness for SKE]\label{def:correctness_ske}
There are two types of correctness, namely,
decryption correctness and special correctness.
\paragraph{Decryption Correctness:}
There exists a negligible function $\negl$ such that for any $\secp\in \N$ and $m\in\Ms$, 
\begin{align}
\Pr\left[
\Dec(\sk,\ct)\neq m
\ \middle |
\begin{array}{ll}
\sk\lrun \keygen(1^\secp)\\
\ct \lrun \Enc(\sk,m)
\end{array}
\right] 
\leq\negl(\secp).
\end{align}

\paragraph{Special Correctness:}
There exists a negligible function $\negl$ such that for any $\secp\in \N$ and $m\in\Ms$, 
\begin{align}
\Pr\left[
\Dec(\sk_2,\ct)\neq \bot
\ \middle |
\begin{array}{ll}
\sk_2,\sk_1\lrun \keygen(1^\secp)\\
\ct \lrun \Enc(\sk_1,m)
\end{array}
\right] 
\leq\negl(\secp).
\end{align}
\end{definition}
\begin{remark}
In the original definition of SKE schemes, only decryption correctness is required.
In this paper, however, we additionally require special correctness.
This is because we need special correctness for the construction of FE in \cref{sec:const_garbling}.
In fact, special correctness can be easily satisfied as well.
\end{remark}

As security of SKE schemes, we consider OW-CPA security or IND-CPA security defined below.
\begin{definition}[OW-CPA Security for SKE]\label{def:OW-CPA_security_ske}
Let $\ell$ be a polynomial of the security parameter $\secp$.
Let $\Sigma=(\keygen,\Enc,\Dec)$ be a SKE scheme.
We consider the following security experiment $\expa{\Sigma,\cA}{ow}{cpa}(\secp)$ against a QPT adversary $\cA$.
\begin{enumerate}
    \item The challenger computes $\sk \la \keygen(1^\secp)$.
    \item $\cA$ sends an encryption query $m$ to the challenger.
    The challenger computes $\ct\la\Enc(\sk,m)$ and returns $\ct$ to $\cA$.
    $\cA$ can repeat this process polynomially many times.
    \item The challenger samples $(m^1,\cdots, m^\ell)\la \Ms^\ell$, computes $\ct^i \la \Enc(\sk,m^i)$ for all $i\in[\ell]$ and sends $\{\ct^i\}_{i\in[\ell]}$ to $\cA$.
    \item $\cA$ sends an encryption query $m$ to the challenger.
    The challenger computes $\ct\la\Enc(\sk,m)$ and returns $\ct$ to $\cA$.
    $\cA$ can repeat this process polynomially many times.
    \item $\cA$ outputs $m'$.
    \item The output of the experiment is $1$ if $m'=m^i$ for some $i\in[\ell]$. Otherwise, the output of the experiment is $0$.
\end{enumerate}
We say that the $\Sigma$ is OW-CPA secure if, for any QPT $\cA$, it holds that
\begin{align}
\advb{\Sigma,\cA}{ow}{cpa}(\secp)
\seteq \Pr[ \expa{\Sigma,\cA}{ow}{cpa}(\secp)=1]\leq \negl(\secp).
\end{align}
Note that we assume $1/|\Ms|$ is negligible.
\end{definition}

\begin{definition}[IND-CPA Security for SKE]\label{def:IND-CPA_security_ske}
Let $\Sigma=(\keygen,\Enc,\Dec)$ be a SKE scheme.
We consider the following security experiment $\expa{\Sigma,\cA}{ind}{cpa}(\secp,b)$ against a QPT adversary $\cA$.
\begin{enumerate}
    \item The challenger computes $\sk \la \keygen(1^\secp)$.
    \item $\cA$ sends an encryption query $m$ to the challenger.
    The challenger computes $\ct\la\Enc(\sk,m)$ and returns $\ct$ to $\cA$.
    $\cA$ can repeat this process polynomially many times.
    \item $\cA$ sends $(m_0,m_1)\in\cM^2$ to the challenger.
    \item The challenger computes $\ct \la \Enc(\sk,m_b)$ and sends $\ct$ to $\cA$.
    \item $\cA$ sends an encryption query $m$ to the challenger.
    The challenger computes $\ct\la\Enc(\sk,m)$ and returns $\ct$ to $\cA$.
    $\cA$ can repeat this process polynomially many times.
    \item $\cA$ outputs $b'\in \bit$. This is the output of the experiment.
\end{enumerate}
We say that $\Sigma$ is IND-CPA secure if, for any QPT $\cA$, it holds that
\begin{align}
\advb{\Sigma,\cA}{ind}{cpa}(\secp)
\seteq \abs{\Pr[ \expa{\Sigma,\cA}{ind}{cpa}(\secp, 0)=1] - \Pr[ \expa{\Sigma,\cA}{ind}{cpa}(\secp, 1)=1] }\leq \negl(\secp).
\end{align}
\end{definition}
It is well-known that IND-CPA security implies OW-CPA security. 
A SKE scheme exists if there exists a pseudorandom function.

\def\Garbled{1}
\ifnum\Garbled=0
\paragraph{Garbling Scheme.}
\begin{definition}[Garbling Scheme (Syntax)]\label{def:garble}
Let $\lambda$ be a security parameter and let $p$ and $q$ be some polynomials.
Let $\Cs_n$ be a family of circuits that take $n$-bit inputs.
A garbling scheme is a tuple of algorithms $\Sigma=(\Samp,\Garble,\Eval)$
with label space $\mathcal{L}\seteq \bit^{p(\lambda)}$ and
garbled circuit space $\widetilde{\mathcal{C}}\seteq \bit^{q(\lambda)}$.
\begin{description}
\item[$\Samp(1^\secp)\ra \{L_{i,\alpha}\}_{i\in[n],\alpha\in\bit}$:]
The sampling algorithm takes a security parameter $1^\lambda$ as input and outputs $2n$ labels $\{L_{i,\alpha}\}_{i\in[n],\alpha\in\bit}$ where $L_{i,\alpha}\in\mathcal{L}$ for each $i\in[n]$ and $\alpha\in\bit$.
\item[$\Garble(1^{\secp},C,\{L_{i,\alpha}\}_{i\in[n],\alpha\in\{0,1\}})\ra \widetilde{C}$:] 
The garbling algorithm takes $1^\lambda$, a circuit $C\in\Cs_n$ and $2n$ labels $\{L_{i,\alpha}\}_{i\in[n],\alpha\in\bit}$ as input and outputs a garbled circuit $\widetilde{C}\in\widetilde{\mathcal{C}}$.

\item[$\Eval(\widetilde{C},\{L_{i,x_i}\}_{i\in[n]})\ra y$:]
The evaluation algorithm takes $\widetilde{C}$ and $n$ labels $\{L_{i,x_i}\}_{i\in[n]}$ where $x_i\in\{0,1\}$ as input and outputs $y$.
\end{description}
\end{definition}

We require that a garbling scheme satisfies evaluation correctness defined below.
\begin{definition}[Evaluation Correctness for Garbling Scheme]\label{def:correctness_garble}
There exists a negligible function $\negl$ such that for any $\secp\in\N$, $n\in \poly(\lambda)$ , $x\in\bit^n$, and $C\in\Cs_n$,
\begin{align}
\Pr\left[
\Eval(\widetilde{C},\{L_{i,x_i}\}_{i\in[n]})\neq C(x)
\ \middle |
\begin{array}{ll}
\{L_{i,\alpha}\}_{i\in[n],\alpha\in\bit}\la\Samp(1^\secp),
\widetilde{C}\lrun \Garble(1^\secp,C,\{L_{i,\alpha}\}_{i\in[n],\alpha\in\{0,1\}})
\end{array}
\right] 
\leq\negl(\secp).
\end{align}
\end{definition}

We require that a garbling scheme satisfies selective security defined below.
\begin{definition}[Selective Security for Garbling Scheme]\label{def:security_garble}
Let $\Sigma=(\Samp,\Garble,\Eval)$ be a garbling scheme.
Let $\GC.\Sim$ be a QPT algorithm.
We consider the following experiment $\Expt{\Sigma,\cA}{\mathsf{selct}}(1^\secp,b)$.
\begin{enumerate}
    \item $\cA$ sends a circuit $C\in\Cs_n$ and an input $x\in\bit^n$ to the challenger.
    \item The challenger computes $\{L_{i,\alpha}\}_{i\in[n],\alpha\in\bit}\la\Samp(1^\secp)$.
    \item If $b=0$, the challenger computes $\widetilde{C}\la\Garble(1^\secp,C,\{L_{i,\alpha}\}_{i\in[n],\alpha\in\bit})$ and returns $(\widetilde{C},\{L_{i,x_i}\}_{i\in[n]})$ to $\cA$.
    If $b=1$, the challenger computes $\widetilde{C}\la\GC.\Sim(1^\secp,1^{|C|},C(x),\{L_{i,x_i}\}_{i\in[n]})$ and returns
    $(\widetilde{C},\{L_{i,x_i}\}_{i\in[n]})$ to $\cA$.
    \item $\cA$ outputs $b'\in\bit$. This is the output of the experiment.
\end{enumerate}
We say that the $\Sigma$ is selective secure if there exists a QPT simulator $\GC.\Sim$ such that for any QPT $\cA$,
\begin{align}
    \adva{\Sigma,\cA}{selct}(\secp)\seteq \abs{\Pr[\Expt{\Sigma,\cA}{\mathsf{selct}}(1^\secp,0)=1]-\Pr[\Expt{\Sigma,\cA}{\mathsf{selct}}(1^\secp,1)=1]}\leq\negl(\secp).
\end{align}
\end{definition}
\else
\fi

\paragraph{Public Key Encryption (PKE).}
\begin{definition}[Public Key Encryption (Syntax)]\label{def:pke}
Let $\lambda$ be a security parameter and let $p$, $q$ and $r$ be some polynomials.
A public key encryption scheme is a tuple of algorithms $\Sigma=(\keygen,\Enc,\Dec)$ with 
plaintext space $\Ms:=\{0,1\}^n$, ciphertext space $\Cs:= \bit^{p(\lambda)}$, public key space $\mathcal{PK}:=\{0,1\}^{q(\lambda)}$ 
and secret key space $\mathcal{SK}\seteq \bit^{r(\lambda)}$.
\begin{description}
    \item[$\keygen(1^\secp)\ra (\pk,\sk)$:] The key generation algorithm takes as input the security parameter $1^\secp$ and outputs a public key $\pk\in\mathcal{PK}$ and a secret key $\sk\in\mathcal{SK}$.
    \item[$\Enc(\pk,m) \ra \ct$:] The encryption algorithm takes as input $\pk$ and a plaintext $m \in \Ms$, and outputs a ciphertext $\ct\in\Cs$.
    \item[$\Dec(\sk,\ct) \ra m^\prime \mbox{ or } \bot$:] The decryption algorithm takes as input $\sk$ and $\ct$, and outputs a plaintext $m^\prime$ or $\bot$.
\end{description}
\end{definition}

We require that a PKE scheme satisfies decryption correctness defined below.
\begin{definition}[Decryption Correctness for PKE]\label{def:correctness_pke}
There exists a negligible function $\negl$ such that for any $\secp\in \N$, $m\in\Ms$,
\begin{align}
\Pr\left[
\Dec(\sk,\ct)\ne m
\ \middle |
\begin{array}{ll}
(\pk,\sk)\lrun \keygen(1^\secp)\\
\ct \lrun \Enc(\pk,m)
\end{array}
\right] 
\le\negl(\secp).
\end{align}
\end{definition}

As security, we consider OW-CPA security or IND-CPA security defined below.
\begin{definition}[OW-CPA Security for PKE]\label{def:OW-CPA_security_pke}
Let $\ell$ be a polynomial of the security parameter $\secp$.
Let $\Sigma=(\keygen,\Enc,\Dec)$ be a PKE scheme.
We consider the following security experiment $\expa{\Sigma,\cA}{ow}{cpa}(\secp)$ against a QPT adversary $\cA$.
\begin{enumerate}
    \item The challenger computes $(\pk,\sk) \la \keygen(1^\secp)$.
    \item The challenger samples $(m^1,\cdots, m^\ell)\la \Ms^\ell$, computes $\ct^i \la \Enc(\pk,m^i)$ for all $i\in[\ell]$ and sends $\{\ct^i\}_{i\in[\ell]}$ to $\cA$.
    \item $\cA$ outputs $m'$.
    \item The output of the experiment is $1$ if $m'=m^i$ for some $i\in[\ell]$. Otherwise, the output of the experiment is $0$.
\end{enumerate}
We say that $\Sigma$ is OW-CPA secure if, for any QPT $\cA$, it holds that
\begin{align}
\advb{\Sigma,\cA}{ow}{cpa}(\secp)
\seteq \Pr[ \expa{\Sigma,\cA}{ow}{cpa}(\secp)=1]\leq \negl(\secp).
\end{align}
Note that we assume $1/|\Ms|$ is negligible.
\end{definition}

\begin{definition}[IND-CPA Security for PKE]\label{def:IND-CPA_pke}
Let $\Sigma=(\keygen,\Enc,\Dec)$ be a PKE scheme.
We consider the following security experiment $\expa{\Sigma,\cA}{ind}{cpa}(\secp,b)$ against a QPT adversary $\cA$.

\begin{enumerate}
    \item The challenger generates $(\pk,\sk)\lrun \keygen(1^{\secp})$, and sends $\pk$ to $\cA$.
    \item $\cA$ sends $(m_0,m_1)\in\cM^2$ to the challenger.
    \item The challenger computes $\ct \lrun \Enc(\pk,m_b)$, and sends $\ct$ to $\cA$.
    \item $\cA$ outputs $b'\in\bit$. This is the output of the experiment.
\end{enumerate}
We say that $\Sigma$ is IND-CPA secure if, for any QPT $\cA$, it holds that
\begin{align}
\advb{\Sigma,\cA}{ind}{cpa}(\secp) \seteq \abs{\Pr[\expa{\Sigma,\cA}{ind}{cpa}(\secp,0)=1]  - \Pr[\expa{\Sigma,\cA}{ind}{cpa}(\secp,1)=1]} \leq \negl(\secp).
\end{align}
\end{definition}

It is well known that IND-CPA security implies OW-CPA security.
There are many IND-CPA secure PKE schemes against QPT adversaries under standard cryptographic assumptions.
A famous one is Regev PKE scheme, which is IND-CPA secure if the learning with errors (LWE) assumption holds against QPT adversaries~\cite{JACM:Regev09}.
See ~\cite{JACM:Regev09,GPV08} for the LWE assumption and constructions of post-quantum PKE.

\def\ot{0}
\ifnum\ot=1
{\color{red}
\paragraph{$k$-out-of-$n$ Oblivious Transfer}
We use dual mode $k$-out-of-$n$ oblivious transfer introduced in~\cite{TomYam21}.
In this paper, we use oblivious transfer in the binding mode.

\begin{definition}[$k$-out-of-$n$ Oblivious Transfer(Syntax)]\label{def:ot}
A $k$-out-of-$n$ oblivious transfer scheme with message $\Ms$ is a tuple of algorithm $\Sigma=(\mathsf{CRSGen},\mathsf{Receiver},\mathsf{Sender},\mathsf{Derive})$.
\begin{description}
    \item[$\mathsf{CRSGen}(1^\secp)\ra \crs$:] 
    This is an algorithm supposed to run by a trusted third party that takes as input the security parameter $1^\secp$, and outputs $\crs$.
    \item[$\mathsf{Receiver}(\crs,J)\ra (\mathsf{ot}_1,\state)$:] 
    This is an algorithm supposed to run by a receiver that takes as input $\crs$ and $J\in[n]_k$, and outputs $\mathsf{ot}_1$ and an internal state $\state$.
    \item[$\mathsf{Sender}(\crs,\mathsf{ot}_1,\bm{\mu}) \ra \mathsf{ot}_2$:]
    This is an algorithm supposed to run by a sender that takes as input that takes as input $\crs$, $\mathsf{ot}_1$ and $\bm{\mu}\in\Ms^n$, and outputs $\mathsf{ot}_2$.
    \item[$\mathsf{Derive}(\state,\mathsf{ot}_2)\ra \bm{\mu}'$:]
    This is an algorithm supposed to run by a receiver that takes as input $\state$ and $\mathsf{ot}_2$, and outputs $\bm{\mu}'\in\Ms^k$.
\end{description}
\end{definition}
We require that a oblivious transfer satisfies decryption correctness defined below.

\begin{definition}[Decryption Correctness]
There exists a negligible function such that for any $\secp\in \N$, $J=(j_1,\cdots ,j_k)\in[n]_k$ and $\bm{\mu} \in\Ms^n$, 
\begin{align}
\Pr\left[
\bm{\mu}'\neq (\mu_{j_1},\cdots, \mu_{j_k}) 
\ \middle |
\begin{array}{ll}
\crs\la\mathsf{CRSGen}(1^\secp)\\
(\mathsf{ot}_1,\state)\la\mathsf{Receiver}(\crs,J) \\
\mathsf{ot}_2\la\mathsf{Sender}(\crs,\mathsf{ot}_1,\bm{\mu})\\
\bm{\mu}'\la\mathsf{Derive}(\state,\mathsf{ot}_2)
\end{array}
\right] 
\leq\negl(\secp).
\end{align}
\end{definition}

We require that oblivious transfer satisfies the securities defined below.
\begin{definition}[Security for Oblivious Transfer]
There are two types of security.
The first is statistical receiver's security, and the second is computational sender's security.

\paragraph{Statistical Receiver's Security:}
There exists a QPT simulator $\Sim$ such that for any $\secp\in\N$,
and any unbounded distinguisher $\cD$,

\begin{align}
\left|
\Pr\left[
1\la\cD(\crs,\mathsf{ot}_1)
\ \middle |
\begin{array}{ll}
\crs\la\mathsf{CRSGen}(1^\secp)\\
(\mathsf{ot}_1,\state)\la\mathsf{Receiver}(\crs,J) 
\end{array}
\right]
-
\Pr\left[
1\la\cD(\crs,\mathsf{ot}_1)
\ \middle |
\begin{array}{ll}
\crs\la\mathsf{CRSGen}(1^\secp)\\
\mathsf{ot}_1\la\Sim(\crs)
\end{array}
\right] 
\right|
\leq\negl(\secp).
\end{align}

\paragraph{Computational Sender's Security}
There are a PPT algorithm $\Sim_{\mathsf{CRS}}$ and a QPT algorithm $\Sim_{\mathsf{sen}}$ and a QPT algorithm $\Open_{\mathsf{rec}}$ such that the following two properties are satisfied.

\begin{itemize}
    \item For any QPT distinguisher $\cD$, we have 
    \begin{align}
        \abs{\Pr[1\la\cD(\crs)|\crs\la\mathsf{CRSGen}(1^\secp)]-\Pr[1\la\cD(\crs)|(\crs,\td)\la\Sim_{\mathsf{CRS}}(1^\secp)]}\leq \negl(\secp).
    \end{align}
    \item For any QPT adversary $\cA=(\cA_0,\cA_1)$ and $\bm{\mu}=(\mu_1,\cdots ,\mu_n)$, we have
    \begin{align}
        \left|
        \Pr\left[
        1\la\cA_1(\state_\cA, \mathsf{ot}_2)
        \ \middle |
        \begin{array}{ll}
        (\crs,\td)\la\Sim_{\mathsf{CRS}}(1^\secp)\\
        (\mathsf{ot}_1,\state_\cA)\la\cA_0(\crs,\td)\\
        \mathsf{ot}_2\la\mathsf{Sender}(\crs,\mathsf{ot}_1,\bm{\mu})
        \end{array}
        \right]
        -
        \Pr\left[
        1\la\cA_1(\state_\cA,\mathsf{ot}_2)
        \ \middle |
        \begin{array}{ll}
        (\crs,\td)\la\Sim_{\mathsf{CRS}}(1^\secp)\\
        (\mathsf{ot}_1,\state_\cA)\la\cA_0(\crs,\td)\\
        J\seteq \Open_{\mathsf{rec}}(\td,\mathsf{ot}_1)\\
        \mathsf{ot}_2\la\Sim_{\mathsf{Sen}}(\crs,\mathsf{ot}_1,J,\bm{\mu})
        \end{array}
        \right] 
        \right|
        \leq\negl(\secp).
    \end{align}
\end{itemize}

\end{definition}

}
\else
\fi

\section{Certified Everlasting Secret Key Encryption}\label{sec:SKE}
In \cref{sec:def_ske}, we define certified everlasting SKE.
In \cref{sec:const_ske_rom} and \cref{sec:const_ske_wo_rom}, 
we construct a certified everlasting SKE scheme with and without QROM, respectively.

\subsection{Definition}\label{sec:def_ske}

\begin{definition}[Certified Everlasting SKE (Syntax)]\label{def:cert_ever_ske}
Let $\lambda$ be a security parameter and let $p$, $q$, $r$ and $s$ be some polynomials.
A certified everlasting SKE scheme is a tuple of algorithms $\Sigma=(\keygen,\Enc,\Dec,\Delete,\Vrfy)$ with plaintext space $\Ms:=\{0,1\}^n$, ciphertext space $\Cs:= \cQ^{\otimes p(\lambda)}$,
secret key space $\mathcal{SK}:=\{0,1\}^{q(\lambda)}$,
verification key space $\mathcal{VK}\seteq\bit^{r(\secp)}$,
and deletion certificate space $\mathcal{D}:= \cQ^{\otimes s(\lambda)}$.
\begin{description}
    \item[$\keygen (1^\secp) \ra \sk$:]
    The key generation algorithm takes the security parameter $1^\secp$ as input and outputs a secret key $\sk \in \mathcal{SK}$.
    \item[$\Enc(\sk,m) \ra (\vk,\ct)$:]
    The encryption algorithm takes $\sk$ and a plaintext $m\in\Ms$ as input, and outputs a verification key $\vk\in\mathcal{VK}$ and a ciphertext $\ct\in \Cs$.
    \item[$\Dec(\sk,\ct) \ra m^\prime~or~\bot$:]
    The decryption algorithm takes $\sk$ and $\ct$ as input, and outputs a plaintext $m^\prime \in \Ms$ or $\bot$.
    \item[$\Delete(\ct) \ra \cert$:] The deletion algorithm takes $\ct$ as input, and outputs a certification $\cert\in\mathcal{D}$.
    \item[$\Vrfy(\vk,\cert)\ra \top~\mbox{\bf or}~\bot$:] The verification algorithm takes $\vk$ and $\cert$ as input, and outputs $\top$ or $\bot$.
\end{description}
\end{definition}

We require that a certified everlasting SKE scheme satisfies correctness defined below.
\begin{definition}[Correctness for Certified Everlasting SKE]\label{def:correctness_cert_ever_ske}
There are four types of correctness, namely,
decryption correctness,
verification correctness,
special correctness,
and modification correctness.

\paragraph{Decryption Correctness:} There exists a negligible function $\negl$ such that for any $\secp\in \N$ and $m\in\Ms$, 
\begin{align}
\Pr\left[
m'\neq m
\ \middle |
\begin{array}{ll}
\sk\lrun \keygen(1^\secp)\\
(\vk,\ct) \lrun \Enc(\sk,m)\\
m'\la\Dec(\sk,\ct)
\end{array}
\right] 
\leq\negl(\secp).
\end{align}

\paragraph{Verification Correctness:} There exists a negligible function $\negl$ such that for any $\secp\in \N$ and $m\in\Ms$, 
\begin{align}
\Pr\left[
\Vrfy(\vk,\cert)=\bot
\ \middle |
\begin{array}{ll}
\sk\lrun \keygen(1^\secp)\\
(\vk,\ct) \lrun \Enc(\sk,m)\\
\cert \lrun \Delete(\ct)
\end{array}
\right] 
\leq
\negl(\secp).
\end{align}

\paragraph{Special Correctness:}
There exists a negligible function $\negl$ such that for any $\secp\in \N$ and $m\in\Ms$, 
\begin{align}
\Pr\left[
\Dec(\sk_2,\ct)\neq \bot
\ \middle |
\begin{array}{ll}
\sk_2,\sk_1\lrun \keygen(1^\secp)\\
(\vk,\ct) \lrun \Enc(\sk_1,m)
\end{array}
\right] 
\leq\negl(\secp).
\end{align}

\paragraph{Modification Correctness:}
There exists a negligible function $\negl$ and a QPT algorithm $\Modify$ such that for any $\secp\in \N$ and $m\in\Ms$, 
\begin{align}
\Pr\left[
\Vrfy(\vk,\cert^*)=\bot
\ \middle |
\begin{array}{ll}
\sk\lrun \keygen(1^\secp)\\
(\vk,\ct) \lrun \Enc(\sk,m)\\
a,b\la\bit^{p(\lambda)}\\
\cert \lrun \Delete(Z^bX^a\ct X^aZ^b)\\
\cert^*\lrun \Modify(a,b,\cert)  
\end{array}
\right] 
\leq
\negl(\secp).
\end{align}

\end{definition}

\begin{remark}
Minimum requirements for correctness are decryption correctness and verification correctness.
In this paper, however, we also require special correctness and modification correctness,
because we need special correctness for the construction of the garbling scheme in \cref{sec:const_garbling},
and modification correctness for the construction of functional encryption in \cref{sec:const_fe_adapt}.
\end{remark}


As security, we consider two definitions, 
\cref{def:IND-CPA_security_cert_ever_ske} 
and
\cref{def:cert_ever_security_cert_ever_ske} given below.
The former is just the ordinal IND-CPA security 
and the latter is the certified everlasting security
that we newly define in this paper.
Roughly, the everlasting security guarantees that any QPT adversary cannot obtain plaintext information even if it becomes computationally
unbounded and obtains the secret key after it issues a valid certificate.

\begin{definition}[IND-CPA Security for Certified Everlasting SKE]\label{def:IND-CPA_security_cert_ever_ske}
Let $\Sigma=(\keygen,\Enc,\Dec,\Delete,\Vrfy)$ be a certified everlasting SKE scheme.
We consider the following security experiment $\expa{\Sigma,\cA}{ind}{cpa}(\secp,b)$ against a QPT adversary $\cA$.
\begin{enumerate}
    \item The challenger computes $\sk \la \keygen(1^\secp)$.
    \item $\cA$ sends an encryption query $m$ to the challenger.
    The challenger computes $(\vk,\ct)\la\Enc(\sk,m)$, and returns $(\vk,\ct)$ to $\cA$.
    $\cA$ can repeat this process polynomially many times.
    \item $\cA$ sends $(m_0,m_1)\in\cM^2$ to the challenger.
    \item The challenger computes $(\vk,\ct) \la \Enc(\sk,m_b)$, and sends $\ct$ to $\cA$.
    \item $\cA$ sends an encryption query $m$ to the challenger.
    The challenger computes $(\vk,\ct)\la\Enc(\sk,m)$, and returns $(\vk,\ct)$ to $\cA$.
    $\cA$ can repeat this process polynomially many times.
    \item $\cA$ outputs $b'\in \bit$. This is the output of the experiment.
\end{enumerate}
We say that $\Sigma$ is IND-CPA secure if, for any QPT $\cA$, it holds that
\begin{align}
\advb{\Sigma,\cA}{ind}{cpa}(\secp)
\seteq \abs{\Pr[ \expa{\Sigma,\cA}{ind}{cpa}(\secp, 0)=1] - \Pr[ \expa{\Sigma,\cA}{ind}{cpa}(\secp, 1)=1] }\leq \negl(\secp).
\end{align}
\end{definition}

\begin{definition}[Certified Everlasting IND-CPA Security for Certified Everlasting SKE]\label{def:cert_ever_security_cert_ever_ske}
Let $\Sigma=(\keygen,\Enc,\Dec,\allowbreak \Delete,\Vrfy)$ be a certified everlasting SKE scheme.
We consider the following security experiment $\expd{\Sigma,\cA}{cert}{ever}{ind}{cpa}(\secp,b)$ against a QPT adversary $\cA_1$ and an unbounded adversary $\cA_2$.
\begin{enumerate}
    \item The challenger computes $\sk \la \keygen(1^\secp)$.
    \item $\cA_1$ sends an encryption query $m$ to the challenger.
    The challenger computes $(\vk,\ct)\la\Enc(\sk,m)$, and returns $(\vk,\ct)$ to $\cA_1$.
    $\cA_1$ can repeat this process polynomially many times.
    \item $\cA_1$ sends $(m_0,m_1)\in\cM^2$ to the challenger.
    \item The challenger computes $(\vk,\ct) \la \Enc(\sk,m_b)$, and sends $\ct$ to $\cA_1$.
    \item $\cA_1$ sends an encryption query $m$ to the challenger.
    The challenger computes $(\vk,\ct)\la\Enc(\sk,m)$, and returns $(\vk,\ct)$ to $\cA_1$.
    $\cA_1$ can repeat this process polynomially many times.
    \item At some point, $\cA_1$ sends $\cert$ to the challenger and sends the internal state to $\cA_2$.
    \item The challenger computes $\Vrfy(\vk,\cert)$.
    If the output is $\bot$, the challenger outputs $\bot$, and sends $\bot$ to $\cA_2$.
    Otherwise, the challenger outputs $\top$, and sends $\sk$ to $\cA_2$.
    \item $\cA_2$ outputs $b'\in\{0,1\}$.
    \item If the challenger outputs $\top$, then the output of the experiment is $b'$. Otherwise, the output of the experiment is $\bot$.
\end{enumerate}
We say that $\Sigma$ is certified everlasting IND-CPA secure if, for any QPT $\cA_1$ and any unbounded $\cA_2$, it holds that
\begin{align}
\advd{\Sigma,\cA}{cert}{ever}{ind}{cpa}(\secp)
\seteq \abs{\Pr[ \expd{\Sigma,\cA}{cert}{ever}{ind}{cpa}(\secp, 0)=1] - \Pr[ \expd{\Sigma,\cA}{cert}{ever}{ind}{cpa}(\secp, 1)=1] }\leq \negl(\secp).
\end{align}
\end{definition}

\subsection{Construction with QROM}\label{sec:const_ske_rom}
In this section, we construct a certified everlasting SKE scheme with QROM.
Our construction is similar to that of the certified everlasting commitment scheme in \cite{EPRINT:HMNY21_2}.
The difference is that we use SKE instead of commitment.

\paragraph{Our certified everlasting SKE scheme.}
We construct a certified everlasting SKE scheme $\Sigma_{\mathsf{cesk}}=(\keygen,\Enc,\Dec,\allowbreak \Delete,\Vrfy)$ from the following primitives.
\begin{itemize}
    \item A one-time SKE with certified deletion scheme~(\cref{def:sk_cert_del}) $\Sigma_{\mathsf{skcd}}=\mathsf{CD}.(\keygen,\Enc,\Dec,\Delete,\Vrfy)$. 
    \item A SKE scheme~(\cref{def:ske}) $\Sigma_{\mathsf{sk}}=\mathsf{SKE}.(\keygen,\Enc,\Dec)$ with plaintext space $\bit^\secp$.
    \item A hash function $H$ 
    modeled as a quantum random oracle.
\end{itemize}

\begin{description}
    \item[$\keygen(1^{\secp})$:] $ $
    \begin{itemize}
        \item Generate $\ske.\sk\la\SKE.\keygen(1^\secp)$.
        \item Output $\sk\seteq\ske.\sk$.
    \end{itemize}
    \item[$\Enc(\sk,m)$:] $ $
    \begin{itemize}
        \item Parse $\sk=\ske.\sk$.
        \item Generate $\mathsf{cd}.\sk\la\mathsf{CD}.\keygen(1^\secp)$ and $R\la\bit^\lambda$.
        \item Compute $\ske.\ct\la \SKE.\Enc(\ske.\sk,R)$.
        \item Compute $h\seteq H(R)\oplus \mathsf{cd}.\sk$ and $\mathsf{cd}.\ct\la \mathsf{CD}.\Enc(\mathsf{cd}.\sk,m)$.
        \item Output $\ct\seteq (h,\ske.\ct,\mathsf{cd}.\ct)$ and $\vk\seteq \mathsf{cd}.\sk$.
    \end{itemize}
    \item[$\Dec(\sk,\ct)$:] $ $
    \begin{itemize}
        \item Parse $\sk=\ske.\sk$ and $\ct= (h,\ske.\ct,\mathsf{cd}.\ct)$.
        \item Compute $R' {\bf\,\, or \,\,}\bot\la\SKE.\Dec(\ske.\sk,\ske.\ct)$.
        If it outputs $\bot$, $\Dec(\sk,\ct)$ outputs $\bot$.
        \item Compute $\mathsf{cd}.\sk'\seteq H(R')\oplus h$.
        \item Compute $m' \la\mathsf{CD}.\Dec(\mathsf{cd}.\sk',\mathsf{cd}.\ct)$.
        \item Output $m'$.
    \end{itemize}
    \item[$\Delete(\ct)$:] $ $
    \begin{itemize}
        \item Parse $\ct=(h,\ske.\ct,\mathsf{cd}.\ct)$.
        \item Compute $\mathsf{cd}.\cert\la\mathsf{CD}.\Delete(\mathsf{cd}.\ct)$.
        \item Output $\cert\seteq\mathsf{cd}.\cert$.
    \end{itemize}
    \item[$\Vrfy(\vk,\cert)$:] $ $
    \begin{itemize}
        \item Parse $\vk=\mathsf{cd}.\sk$ and $\cert=\mathsf{cd}.\cert$.
        \item Compute $b\la\SKE.\Vrfy(\mathsf{cd}.\sk,\mathsf{cd}.\cert)$.
        \item Output $b$.
    \end{itemize}
\end{description}

\paragraph{Correctness:}
It is easy to see that correctness of $\Sigma_{\mathsf{cesk}}$ comes from those of $\Sigma_{\mathsf{sk}}$ and $\Sigma_{\mathsf{skcd}}$.

\paragraph{Security:}
The following two theorems hold.
\begin{theorem}\label{thm:sk_comp_security}
If $\Sigma_{\mathsf{sk}}$ satisfies the OW-CPA security~(\cref{def:OW-CPA_security_ske}) and $\Sigma_{\mathsf{skcd}}$ satisfies the OT-CD security~(\cref{def:sk_cert_del}), 
$\Sigma_{\mathsf{cesk}}$ satisfies the IND-CPA security~(\cref{def:IND-CPA_security_cert_ever_ske}).
\end{theorem}
Its proof is similar to that of \cref{thm:sk_ever_security}, and therefore we omit it.

\begin{theorem}\label{thm:sk_ever_security}
If $\Sigma_{\mathsf{sk}}$ satisfies the OW-CPA security~(\cref{def:OW-CPA_security_ske})
and $\Sigma_{\mathsf{skcd}}$ satisfies the OT-CD security~(\cref{def:sk_cert_del}),
$\Sigma_{\mathsf{cesk}}$ satisfies the certified everlasting IND-CPA security~(\cref{def:cert_ever_security_cert_ever_ske}).
\end{theorem}
Its proof is similar to that of \cite[Theorem~5.8]{EPRINT:HMNY21_2}.

\subsection{Construction without QROM}\label{sec:const_ske_wo_rom}
In this section, we construct a certified everlasting SKE scheme without QROM.
Note that unlike the construction with QROM (\cref{sec:const_ske_rom}), in this construction
the plaintext space is of constant size.
However, the size can be extended to the polynomial size via the standard hybrid argument.
Our construction is similar to that of revocable quantum timed-release encryption in \cite{JACM:Unruh15}.
The difference is that we use SKE instead of timed-release encryption.

\paragraph{Our certified everlasting SKE scheme without QROM.}
Let $k_1$ and $k_2$ be constants such that $k_1>k_2$.
Let $p$, $q$, $r$, $s$ and $t$ be polynomials.
Let $(C_1,C_2)$ be a CSS code with parameters $q,k_1,k_2,t$.
We construct a certified everlasting SKE scheme
$\Sigma_{\mathsf{cesk}}=(\keygen,\Enc,\Dec,\Delete,\Vrfy)$ with plaintext space $\Ms=C_1/C_2$ (isomorphic to $\bit^{k_1-k_2}$), 
ciphertext space $\Cs=\cQ^{\otimes \left(p(\lambda)+q(\lambda)\right)}\times \{0,1\}^{r(\lambda)}\times \bit^{q(\lambda)}/C_1\times C_1/C_2$,
secret key space $\mathcal{SK}= \bit^{s(\lambda)}$,
verification key space $\mathcal{VK}=\{0,1\}^{p(\lambda)}\times [p(\lambda)+q(\lambda)]_{p(\lambda)}\times \bit^{p(\lambda)}$ 
and deletion certificate space $\mathcal{D}=\mathcal{Q}^{\otimes \left(p(\lambda)+q(\lambda)\right)}$
from the following primitive.
\begin{itemize}
    \item A SKE scheme~(\cref{def:ske}) $\Sigma_{\mathsf{sk}}=\SKE.(\keygen,\Enc,\Dec)$ with plaintext space $\Ms=\{0,1\}^{p(\lambda)}\times[p(\lambda)+q(\lambda)]_{p(\lambda)}\times \bit^{p(\lambda)} \times C_1/C_2$, secret key space $\mathcal{SK}=\{0,1\}^{s(\secp)}$ and ciphertext space $\Cs=\bit^{r(\lambda)}$. 
\end{itemize}
The construction is as follows. (We will omit the security parameter below.)
\begin{description}
\item[$\keygen(1^\secp)$:]$ $
\begin{itemize}
    \item Generate $\ske.\sk\la\SKE.\keygen(1^\secp)$.
    \item Output $\sk\seteq\ske.\sk$.
\end{itemize}
\item[$\Enc(\sk,m)$:] $ $
\begin{itemize}
    \item Parse $\sk=\ske.\sk$.
    \item Generate $B\la\bit^{p}$, $Q\la[p+q]_{p}$, $y\la C_1/C_2$, $u\la\bit^q/C_1$, $r\la\bit^p$,
    $x\la C_1/C_2$, $w\la C_2$.
    \item Compute $\ske.\ct\la\SKE.\Enc\left(\ske.\sk,(B,Q,r,y)\right)$.
    \item Let $U_Q$ be the unitary that permutes the qubits in $Q$ into the first half of the system.
    (I.e., $U_Q\ket{x_1x_2\cdots x_{p+q}}=\ket{x_{a_1}x_{a_2}\cdots x_{a_p}x_{b_1}x_{b_2}\cdots x_{b_q}}$ with $Q\seteq \{a_1,a_2,\cdots, a_p\}$ and $\{1,2,\cdots, p+q\}\backslash Q\seteq \{b_1,b_2,\cdots,b_q\}$.)
    \item Construct a quantum state $\ket{\Psi} \seteq U_{Q}^{\dagger}(H^B\otimes I^{\otimes q})(\ket{r}\otimes \ket{x\oplus w\oplus u})$.
    \item Compute $h\seteq m\oplus x\oplus y$.
    \item Output $\ct\seteq (\ket{\Psi},\ske.\ct,u,h)$ and $\vk\seteq (B,Q,r)$.
\end{itemize}
\item[$\Dec(\sk,\ct)$:] $ $
\begin{itemize}
    \item Parse $\sk=\ske.\sk$, $\ct=(\ket{\Psi},\ske.\ct,u,h)$.
    \item Compute $(B,Q,r,y)/\bot\la\SKE.\Dec(\ske.\sk,\ske.\ct)$.
    If $\bot\la\SKE.\Dec(\ske.\sk,\ske.\ct)$, $\Dec(\sk,\ct)$ outputs $\bot$ and aborts.
    \item Apply $U_Q$ to $\ket{\Psi}$, measure the last $q$-qubits in the computational basis and obtain the measurement outcome $\gamma\in\bit^q$. 
    \item Compute $x\seteq \gamma\oplus u$ mod $C_2$.
    \item Output $m'\seteq h\oplus x\oplus y$.
\end{itemize}
\item[$\Delete(\ct)$:] $ $
\begin{itemize}
    \item Parse $\ct=(\ket{\Psi},\ske.\ct,u,h)$.
    \item Output $\cert\seteq\ket{\Psi}$.
\end{itemize}
\item[$\Vrfy(\vk,\cert)$:] $ $
\begin{itemize}
    \item Parse $\vk=(B,Q,r)$ and $\cert=\ket{\Psi}$.
    \item Apply $(H^B\otimes I^{\otimes q})U_Q$ to $\ket{\Psi}$, measure the first $p$-qubits in the computational basis and obtain the measurement outcome $r'\in\bit^p$.
    \item Output $\top$ if $r=r'$ and output $\bot$ otherwise.
\end{itemize}
\end{description}

\paragraph{Correctness.}
Correctness easily follows from that of
$\Sigma_{\sk}$.

\paragraph{Security.}
The following two theorems hold.

\begin{theorem}\label{thm:sk_comp_security_wo_rom}
If $\Sigma_{\mathsf{sk}}$ is IND-CPA secure~(\cref{def:IND-CPA_security_ske}), then $\Sigma_{\mathsf{cesk}}$ is IND-CPA secure~(\cref{def:IND-CPA_security_cert_ever_ske}).
\end{theorem}

Its proof is straightforward, so we omit it.

\begin{theorem}\label{thm:sk_ever_security_wo_rom}
If $\Sigma_{\mathsf{sk}}$ is IND-CPA secure~(\cref{def:IND-CPA_security_ske}) and $tp/(p+q)-4(k_1-k_2){\rm ln}2$ is superlogarithmic, then $\Sigma_{\mathsf{cesk}}$ is certified everlasting IND-CPA secure~(\cref{def:cert_ever_security_cert_ever_ske}).
\end{theorem}

Its proof is similar to that of \cite[Theorem~3]{JACM:Unruh15}.

Note that the plaintext space is of constant size in our construction. However,
via the standard hybrid argument 
, we can extend it to the polynomial size.



\section{Certified Everlasting Public Key Encryption}\label{sec:PKE}
In \cref{sec:def_pke}, we define certified everlasting PKE.
In \cref{sec:const_pke_rom} and \cref{sec:const_pke_wo_rom}, we construct a certified everlasting PKE scheme with and without QROM, respectively.
\subsection{Definition}\label{sec:def_pke}

\begin{definition}[Certified Everlasting PKE]\label{def:cert_ever_pke}
Let $\secp$ be a security parameter and let $p$, $q$, $r$, $s$ and $t$ be polynomials. 
A certified everlasting PKE scheme is a tuple of algorithms $\Sigma=(\keygen,\Enc,\Dec,\Delete,\Vrfy)$
with plaintext space $\Ms\seteq\bit^n$,
ciphertext space $\Cs:= \cQ^{\otimes p(\lambda)}$,
public key space $\mathcal{PK}:=\{0,1\}^{q(\lambda)}$, secret key space $\mathcal{SK}\seteq \bit^{r(\secp)}$,
verification key space $\mathcal{VK}\seteq \bit^{s(\secp)}$
and deletion certificate space $\mathcal{D}:= \cQ^{\otimes t(\lambda)}$.
\begin{description}
    \item[$\keygen (1^\secp) \ra (\pk,\sk)$:]
    The key generation algorithm takes the security parameter $1^\secp$ as input and outputs a public key $\pk\in\mathcal{PK}$ and a secret key $\sk \in \mathcal{SK}$.
    \item[$\Enc(\pk,m) \ra (\vk,\ct)$:]
    The encryption algorithm takes $\pk$ and a plaintext $m\in\Ms$ as input,
    and outputs a verification key $\vk\in\mathcal{VK}$ and a ciphertext $\ct\in \Cs$.
    \item[$\Dec(\sk,\ct) \ra m^\prime~or~\bot$:]
    The decryption algorithm takes $\sk$ and $\ct$ as input, and outputs a plaintext $m^\prime \in \Ms$ or $\bot$.
    \item[$\Delete(\ct) \ra \cert$:]
    The deletion algorithm takes $\ct$ as input and outputs a certification $\cert\in\mathcal{D}$.
    \item[$\Vrfy(\vk,\cert)\ra \top~\mbox{\bf or}~\bot$:]
    The verification algorithm takes $\vk$ and $\cert$ as input, and outputs $\top$ or $\bot$.
    \end{description}
\end{definition}

We require that a certified everlasting PKE scheme satisfies correctness defined below.
\begin{definition}[Correctness for Certified Everlasting PKE]\label{def:correctness_cert_ever_pke}
There are three types of correctness, namely, decryption correctness, verification correctness, and modification correctness.

\paragraph{Decryption Correctness:}
There exists a negligible function $\negl$ such that for any $\secp\in \N$ and $m\in\Ms$,
\begin{align}
\Pr\left[
m'\neq m
\ \middle |
\begin{array}{ll}
(\pk,\sk)\lrun \keygen(1^\secp)\\
(\vk,\ct) \lrun \Enc(\pk,m)\\
m'\la\Dec(\sk,\ct)
\end{array}
\right] 
\le\negl(\secp).
\end{align}

\paragraph{Verification Correctness:}
There exists a negligible function $\negl$ such that for any $\secp\in \N$ and $m\in\Ms$,
\begin{align}
\Pr\left[
\Vrfy(\vk,\cert)=\bot
\ \middle |
\begin{array}{ll}
(\pk,\sk)\lrun \keygen(1^\secp)\\
(\vk,\ct) \lrun \Enc(\pk,m)\\
\cert \lrun \Delete(\ct)
\end{array}
\right] 
\leq
\negl(\secp).
\end{align}

\paragraph{Modification Correctness:}
There exists a negligible function $\negl$ and a QPT algorithm $\Modify$ such that for any $\secp\in \N$ and $m\in\Ms$, 
\begin{align}
\Pr\left[
\Vrfy(\vk,\cert^*)=\bot
\ \middle |
\begin{array}{ll}
(\pk,\sk)\lrun \keygen(1^\secp)\\
(\vk,\ct) \lrun \Enc(\pk,m)\\
a,b\la\bit^{p(\secp)}\\
\cert \lrun \Delete(Z^bX^a\ct X^aZ^b)\\
\cert^*\lrun \Modify(a,b,\cert)  
\end{array}
\right] 
\leq
\negl(\secp).
\end{align}
\end{definition}
\begin{remark}
Minimum requirements for correctness are decryption correctness and verification correctness.
In this paper, however, we also require modification correctness,
because we need modification correctness for the construction of functional encryption in \cref{sec:const_fe_adapt}.
\end{remark}


As security, we consider two definitions, 
\cref{def:IND-CPA_security_cert_ever_pke} 
and
\cref{def:cert_ever_security_cert_ever_pke} given below.
The former is just the ordinal IND-CPA security 
and the latter is the certified everlasting security
that we newly define in this paper.
Roughly, the everlasting security guarantees that any QPT adversary cannot obtain plaintext information even if it becomes computationally unbounded and obtains the secret key after it issues a valid certificate.

\begin{definition}[IND-CPA Security for Certified Everlasting PKE]\label{def:IND-CPA_security_cert_ever_pke}
Let $\Sigma=(\keygen,\Enc,\Dec,\Delete,\Vrfy)$ be a certified everlasting PKE scheme.
We consider the following security experiment $\expa{\Sigma,\cA}{ind}{cpa}(\secp,b)$ against a QPT adversary $\cA$.
\begin{enumerate}
    \item The challenger generates $(\pk,\sk)\lrun \keygen(1^{\secp})$, and sends $\pk$ to $\cA$.
    \item $\cA$ sends $(m_0,m_1)\in\cM^2$ to the challenger.
    \item The challenger computes $(\vk,\ct) \lrun \Enc(\pk,m_b)$, and sends $\ct$ to $\cA$.
    \item $\cA$ outputs $b'\in\bit$. This is the output of the experiment.
\end{enumerate}
We say that the $\Sigma$ is IND-CPA secure if, for any QPT $\cA$, it holds that
\begin{align}
\advb{\Sigma,\cA}{ind}{cpa}(\secp) \seteq \abs{\Pr[\expa{\Sigma,\cA}{ind}{cpa}(\secp,0)=1]  - \Pr[\expa{\Sigma,\cA}{ind}{cpa}(\secp,1)=1]} \leq \negl(\secp).
\end{align}
\end{definition}

\begin{definition}[Certified Everlasting IND-CPA Security for Certified Everlasting PKE]\label{def:cert_ever_security_cert_ever_pke}
Let $\Sigma=(\keygen,\Enc,\Dec,\allowbreak \Delete,\Vrfy)$ be a certified everlasting PKE scheme.
We consider the following security experiment $\expd{\Sigma,\cA}{cert}{ever}{ind}{cpa}(\secp,b)$ against a QPT adversary $\cA_1$ and an unbounded adversary $\cA_2$.
\begin{enumerate}
    \item The challenger computes $(\pk,\sk) \la \keygen(1^\secp)$, and sends $\pk$ to $\cA_1$.
    \item $\cA_1$ sends $(m_0,m_1)\in\cM^2$ to the challenger.
    \item The challenger computes $(\vk,\ct) \la \Enc(\pk,m_b)$, and sends $\ct$ to $\cA_1$.
    \item At some point, $\cA_1$ sends $\cert$ to the challenger, and sends the internal state to $\cA_2$.
    \item The challenger computes $\Vrfy(\vk,\cert)$.
    If the output is $\bot$, the challenger outputs $\bot$, and sends $\bot$ to $\cA_2$.
    Otherwise, the challenger outputs $\top$, and sends $\sk$ to $\cA_2$.
    \item $\cA_2$ outputs $b'\in\{0,1\}$.
    \item If the challenger outputs $\top$, then the output of the experiment is $b'$.
    Otherwise, the output of the experiment is $\bot$.
\end{enumerate}
We say that the $\Sigma$ is certified everlasting IND-CPA secure if for any QPT $\cA_1$ and any unbounded $\cA_2$, it holds that
\begin{align}
\advd{\Sigma,\cA}{cert}{ever}{ind}{cpa}(\secp)
\seteq \abs{\Pr[ \expd{\Sigma,\cA}{cert}{ever}{ind}{cpa}(\secp, 0)=1] - \Pr[ \expd{\Sigma,\cA}{cert}{ever}{ind}{cpa}(\secp, 1)=1] }\leq \negl(\secp).
\end{align}
\end{definition}

\subsection{Construction with QROM}\label{sec:const_pke_rom}
In this section, we construct a certified everlasting PKE scheme with QROM.
Our construction is similar to that of the certified everlasting commitment scheme in \cite{EPRINT:HMNY21_2}.
The difference is that we use PKE instead of commitment.

\paragraph{Our certified everlasting PKE scheme.}
We construct a certified everlasting PKE scheme $\Sigma_{\mathsf{cepk}}=(\keygen,\Enc,\Dec,\allowbreak\Delete,\Vrfy)$
from a one-time SKE with certified deletion scheme $\Sigma_{\mathsf{skcd}}=\SKE.(\keygen,\Enc,\Dec,\Delete,\Vrfy)$~(\cref{def:sk_cert_del}),
a PKE scheme $\Sigma_{\mathsf{pk}}=\PKE.(\keygen,\Enc,\Dec)$ with plaintext space $\bit^\secp$~(\cref{def:pke})
and a hash function $H$ modeled as quantum random oracle.

\begin{description}
    \item[$\keygen(1^{\secp})$:] $ $
    \begin{itemize}
        \item Generate $(\pke.\pk,\pke.\sk)\la\keygen(1^\secp)$.
        \item Output $\pk\seteq \pke.\pk$ and $\sk\seteq\pke.\sk$.
    \end{itemize}
    \item[$\Enc(\pk,m)$:] $ $
    \begin{itemize}
        \item Parse $\pk=\pke.\pk$.
        \item Generate $\ske.\sk\la\SKE.\keygen(1^\secp)$.
        \item Randomly generate $R\la\{0,1\}^\secp$.
        \item Compute $\pke.\ct\la\PKE.\Enc(\pke.\pk,R)$.
        \item Compute $h\seteq H(R)\oplus \ske.\sk$ and $\ske.\ct\la\SKE.\Enc(\ske.\sk,m)$.
        \item Output $\ct\seteq (h,\ske.\ct,\pke.\ct)$ and $\vk\seteq \ske.\sk$.
    \end{itemize}
    \item[$\Dec(\sk,\ct)$:] $ $
    \begin{itemize}
        \item Parse $\sk=\pke.\sk$ and $\ct= (h,\ske.\ct,\pke.\ct)$.
        \item Compute $R'\la\PKE.\Dec(\pke.\sk,\pke.\ct)$.
        \item Compute $\ske.\sk'\seteq h\oplus H(R')$.
        \item Compute $m'\la\SKE.\Dec(\ske.\sk',\ske.\ct)$.
        \item Output $m'$.
    \end{itemize}
    \item[$\Delete(\ct)$:] $ $
    \begin{itemize}
        \item Parse $\ct=(h,\ske.\ct,\pke.\ct)$.
        \item Compute $\ske.\cert\la\SKE.\Delete(\ske.\ct)$.
        \item Output $\cert\seteq\ske.\cert$.
    \end{itemize}
    \item[$\Vrfy(\vk,\cert)$:] $ $
    \begin{itemize}
        \item Parse $\vk=\ske.\sk$ and $\cert=\ske.\cert$.
        \item Compute $b\la\SKE.\Vrfy(\ske.\sk,\ske.\cert)$.
        \item Output $b$.
    \end{itemize}
\end{description}

\paragraph{Correctness:}
Correctness easily follows from those of
$\Sigma_{\mathsf{pk}}$ and $\Sigma_{\mathsf{skcd}}$.

\paragraph{Security:}
The following two theorems hold.
Their proofs are similar to those of \cref{thm:sk_comp_security,thm:sk_ever_security}, and therefore we omit them.
\begin{theorem}\label{thm:pk_comp_seucirty}
If $\Sigma_{\mathsf{pk}}$ satisfies the OW-CPA security~(\cref{def:OW-CPA_security_pke}) and $\Sigma_{\mathsf{skcd}}$ satisfies the OT-CD security~(\cref{def:security_sk_cert_del}), 
$\Sigma_{\mathsf{cepk}}$ is IND-CPA secure~(\cref{def:IND-CPA_security_cert_ever_pke}).
\end{theorem}

\begin{theorem}\label{thm:pk_ever_security}
If $\Sigma_{\mathsf{pk}}$ satisfies the OW-CPA security~(\cref{def:OW-CPA_security_pke}) and $\Sigma_{\mathsf{skcd}}$ satisfies the OT-CD security~(\cref{def:security_sk_cert_del}),
$\Sigma_{\mathsf{cepk}}$ is certified everlasting IND-CPA secure~(\cref{def:cert_ever_security_cert_ever_pke}).
\end{theorem}

\subsection{Construction without QROM}\label{sec:const_pke_wo_rom}
In this section, we construct a certified everlasting PKE scheme 
without QROM.
Our construction is similar to that of quantum timed-release encryption presented in \cite{JACM:Unruh15}.
The difference is that we use PKE instead of timed-release encryption.

\paragraph{Our certified everlasting PKE scheme without QROM.}
Let $k_1$ and $k_2$ be some constant such that $k_1>k_2$.
Let $p$, $q$, $r$, $s$, $t$ and $u$ be some polynomials.
Let $(C_1,C_2)$ be a CSS code with parameters $q,k_1,k_2,t$.
We construct a certified everlasting PKE scheme
$\Sigma_{\mathsf{cepk}}=(\keygen,\Enc,\Dec,\Delete,\Vrfy)$, with plaintext space $\Ms=C_1/C_2$ (isomorphic $\bit^{(k_1-k_2)}$), 
ciphertext space $\Cs=\cQ^{\otimes \left(p(\lambda)+q(\lambda)\right)}\times \{0,1\}^{r(\lambda)}\times \bit^{q(\lambda)}/C_1\times C_1/C_2$,
public key space $\mathcal{PK}=\bit^{u(\lambda)}$,
secret key space $\mathcal{SK}= \bit^{s(\lambda)}$,
verification key space $\mathcal{VK}=\{0,1\}^{p(\lambda)}\times [p(\lambda)+q(\lambda)]_{p(\lambda)}\times \bit^{p(\lambda)}$ 
and deletion certificate space $\mathcal{D}=\mathcal{Q}^{\otimes \left(p(\lambda)+q(\lambda)\right)}$
from the following primitive.
\begin{itemize}
    \item A PKE scheme~(\cref{def:pke})$\Sigma_{\mathsf{pk}}=\PKE.(\keygen,\Enc,\Dec)$ with
    plaintext space $\Ms=\{0,1\}^{p(\lambda)}\times[p(\lambda)+q(\lambda)]_{p(\lambda)}\times \bit^{p(\lambda)} \times C_1/C_2$,
    public key space $\mathcal{PK}=\bit^{u(\lambda)}$,
    secret key space $\mathcal{SK}=\{0,1\}^{s(\secp)}$ and ciphertext space $\Cs=\bit^{r(\lambda)}$.
\end{itemize}
The construction is as follows. (We will omit the security parameter below.)
\begin{description}
\item[$\keygen(1^\secp)$:]$ $
\begin{itemize}
    \item Generate $(\pke.\pk,\pke.\sk)\la\PKE.\keygen(1^\secp)$.
    \item Output $\pk\seteq\pke.\pk$ and $\sk\seteq\pke.\sk$.
\end{itemize}
\item[$\Enc(\pk,m)$:] $ $
\begin{itemize}
    \item Parse $\pk=\pke.\pk$.
    \item Generate $B\la\bit^{p}$, $Q\la[p+q]_{p}$, $y\la C_1/C_2$, $u\la\bit^q/C_1$, $r\la\bit^p$,
    $x\la C_1/C_2$, $w\la C_2$.
    \item Compute $\pke.\ct\la\PKE.\Enc\left(\pke.\pk,(B,Q,r,y)\right)$.
    \item Let $U_Q$ be the unitary that permutes the qubits in $Q$ into the first half of the system.
    (I.e., $U_Q\ket{x_1x_2\cdots x_{p+q}}=\ket{x_{a_1}x_{a_2}\cdots x_{a_p}x_{b_1}x_{b_2}\cdots x_{b_q}}$ with $Q\seteq \{a_1,a_2,\cdots, a_p\}$ and $\{1,2,\cdots, p+q\}\setminus Q\seteq \{b_1,b_2,\cdots,b_q\}$.)
    \item Generate a quantum state $\ket{\Psi} \seteq U_{Q}^{\dagger}(H^B\otimes I^{\otimes q})(\ket{r}\otimes \ket{x\oplus w\oplus u})$.
    \item Compute $h\seteq m\oplus x\oplus y$.
    \item Output $\ct\seteq (\ket{\Psi},\pke.\ct,u,h)$ and $\vk\seteq (B,Q,r)$.
\end{itemize}
\item[$\Dec(\sk,\ct)$:] $ $
\begin{itemize}
    \item Parse $\sk=\pke.\sk$ and $\ct=(\ket{\Psi},\pke.\ct,u,h)$.
    \item Compute $(B,Q,r,y)\la\PKE.\Dec(\pke.\sk,\pke.\ct)$.
    \item Apply $U_Q$ to $\ket{\Psi}$, measure the last $q$-qubits in the computational basis and obtain the measurement outcome $\gamma$. 
    \item Compute $x\seteq \gamma\oplus u$ mod $C_2$.
    \item Output $m'\seteq h\oplus x\oplus y$.
\end{itemize}
\item[$\Delete(\ct)$:] $ $
\begin{itemize}
    \item Parse $\ct=(\ket{\Psi},\pke.\ct,u,h)$.
    \item Output $\cert\seteq\ket{\Psi}$.
\end{itemize}
\item[$\Vrfy(\vk,\cert)$:] $ $
\begin{itemize}
    \item Parse $\vk=(B,Q,r)$ and $\cert=\ket{\Psi}$.
    \item Apply $(H^B\otimes I^{\otimes q})U_Q$ to $\ket{\Psi}$, measure the first $p$-qubits in the computational basis and obtain the measurement outcome $r'$.
    \item Output $\top$ if $r=r'$ and output $\bot$ otherwise.
\end{itemize}
\end{description}

\paragraph{Correctness.}
Correctness easily follows from that of $\Sigma_{\pk}$.

\paragraph{Security.}
The following two theorems hold.

\begin{theorem}\label{thm:pk_comp_security_wo_rom}
If $\Sigma_{\mathsf{pk}}$ is IND-CPA secure~(\cref{def:IND-CPA_pke}), then $\Sigma_{\mathsf{cepk}}$ is IND-CPA secure~(\cref{def:IND-CPA_security_cert_ever_pke}).
\end{theorem}
Its proof is straightforward, and therefore we omit it.

\begin{theorem}\label{thm:pk_ever_security_wo_rom}
If $\Sigma_{\mathsf{pk}}$ is IND-CPA secure~(\cref{def:IND-CPA_pke}) and $tp/(p+q)-4(k_1-k_2){\rm ln}2$ is superlogarithmic,
then $\Sigma_{\mathsf{cepk}}$ is certified everlasting IND-CPA secure~(\cref{def:cert_ever_security_cert_ever_pke}).
\end{theorem}
Its proof is similar to that of \cite[Theorem~3]{JACM:Unruh15}.

Note that the plaintext space is of constant size in our construction. However,
via the standard hybrid argument, we can extend it to the polynomial size.


\section{Certified Everlasting Receiver Non-Committing Encryption}\label{sec:RNCE}
In this section, we define and construct certified everlasting receiver non-committing encryption.
In \cref{sec:def_rnce}, we define certified everlasting RNCE.
In \cref{sec:const_rnce_classic}, we construct a certified everlasting RNCE scheme from certified everlasting PKE (\cref{sec:PKE}).

\subsection{Definition}\label{sec:def_rnce}
\begin{definition}[Certified Everlasting RNCE (Syntax)]\label{def:cert_ever_rnce_classic}
Let $\secp$ be the security parameter and let $p$, $q$, $r$, $s$, $t$, $u$, and $v$ be polynomials. 
A certified everlasting RNCE scheme is a tuple of algorithms $\Sigma=(\Setup,\keygen,\Enc,\Dec,\Fake,\Reveal,\allowbreak \Delete,\Vrfy)$
with plaintext space $\Ms\seteq\bit^{n}$,
ciphertext space $\Cs:= \cQ^{\otimes p(\lambda)}$,
public key space $\mathcal{PK}:=\{0,1\}^{q(\lambda)}$,
master secret key space $\mathcal{MSK}\seteq \bit^{r(\lambda)}$,
secret key space $\mathcal{SK}\seteq \bit^{s(\secp)}$,
verification key space $\mathcal{VK}\seteq \bit^{t(\secp)}$,
deletion certificate space $\mathcal{D}:= \cQ^{u(\lambda)}$,
and auxiliary state space $\mathcal{AUX}\seteq\bit^{v(\lambda)}$.
\begin{description}
    \item[$\Setup(1^\secp)\ra(\pk,\MSK)$:]The setup algorithm takes the security parameter $1^\secp$ as input, and outputs a public key $\pk\in\mathcal{PK}$ and a master secret key $\MSK\in\mathcal{MSK}$.
    \item[$\keygen (\MSK) \ra \sk$:]
    The key generation algorithm takes the master secret key $\MSK$ as input,
    and outputs a secret key $\sk\in\mathcal{SK}$.
    \item[$\Enc(\pk,m) \ra (\vk,\ct)$:]
    The encryption algorithm takes $\pk$ and a plaintext $m\in\Ms$ as input,
    and outputs a verification key $\vk\in\mathcal{VK}$ and a ciphertext $\ct\in \Cs$.
    \item[$\Dec(\sk,\ct) \ra m^\prime~or~\bot$:]
    The decryption algorithm takes $\sk$ and $\ct$ as input, and outputs a plaintext $m^\prime \in \Ms$ or $\bot$. 
    \item[$\Fake(\pk)\ra (\vk,\widetilde{\CT},\aux)$:]
    The fake encryption algorithm takes $\pk$ as input, and outputs a verification key $\vk\in\mathcal{VK}$, a fake ciphertext $\widetilde{\CT}\in \Cs$ and an auxiliary state $\aux\in\mathcal{AUX}$.
    \item[$\Reveal(\pk,\MSK,\aux,m)\ra \widetilde{\sk}$:]
    The reveal algorithm takes $\pk,\MSK,\aux$ and $m$ as input, and outputs a fake secret key $\widetilde{\sk}\in\mathcal{SK}$.
    \item[$\Delete(\ct) \ra \cert$:]
    The deletion algorithm takes $\ct$ as input and outputs a certification $\cert\in\mathcal{D}$.
    \item[$\Vrfy(\vk,\cert)\ra \top~\mbox{\bf or}~\bot$:]
    The verification algorithm takes $\vk$ and $\cert$ as input, and outputs $\top$ or $\bot$.
    \end{description}
\end{definition}

We require that a certified everlasting RNCE scheme satisfies correctness defined below.
\begin{definition}[Correctness for Certified Everlasting RNCE]\label{def:correctness_cert_ever_rnce_classic}
There are two types of correctness, namely, decryption correctness and verification correctness.

\paragraph{Decryption Correctness:}
There exists a negligible function $\negl$ such that for any $\secp\in \N$ and $m\in\Ms$,
\begin{align}
\Pr\left[
m'\neq m
\ \middle |
\begin{array}{ll}
(\pk,\MSK)\la\Setup(1^\secp)\\
(\vk,\ct) \lrun \Enc(\pk,m)\\
\sk\lrun\keygen(\MSK)\\
m'\la\Dec(\sk,\ct)
\end{array}
\right] 
\leq \negl(\lambda).
\end{align}

\paragraph{Verification Correctness:}
There exists a negligible function $\negl$ such that for any $\secp\in \N$ and $m\in\Ms$,
\begin{align}
\Pr\left[
\Vrfy(\vk,\cert)=\bot
\ \middle |
\begin{array}{ll}
(\pk,\MSK)\la\Setup(1^\secp)\\
(\vk,\ct) \lrun \Enc(\pk,m)\\
\cert \lrun \Delete(\ct)
\end{array}
\right] 
\leq
\negl(\secp).
\end{align}


\end{definition}

As security, we consider two definitions, 
\cref{def:rec_nc_security_classic} 
and
\cref{def:cert_ever_rec_nc_security_classic} given below.
The former is just the ordinal receiver non-committing security 
and the latter is the certified everlasting security that we newly define in this paper.
Roughly, the everlasting security guarantees that any QPT adversary cannot distinguish whether the ciphertext and the secret key are properly generated or not even if it becomes computationally unbounded and obtains the master secret key after it issues a valid certificate.

\begin{definition}[Receiver Non-Committing (RNC) Security for Certified Everlasting RNCE]\label{def:rec_nc_security_classic}
Let $\Sigma=(\Setup,\keygen,\allowbreak \Enc,\Dec,\Fake,\Reveal,\Delete,\Vrfy)$ be a certified everlasting RNCE scheme.
We consider the following security experiment $\expa{\Sigma,\cA}{rec}{nc}(\secp,b)$ against a QPT adversary $\cA$.
\begin{enumerate}
    \item The challenger runs $(\pk,\MSK)\la\Setup(1^\secp)$ and sends $\pk$ to $\cA$.
    \item $\cA$ sends $m\in\Ms$ to the challenger.
    \item The challenger does the following:
    \begin{itemize}
        \item If $b=0$, the challenger generates $(\vk,\ct)\la\Enc(\pk,m)$ and $\sk\la\keygen(\MSK)$, and sends $(\ct,\sk)$ to $\cA$.
        \item If $b=1$, the challenger generates $(\vk,\widetilde{\ct},\aux)\la \Fake(\pk)$ and $\widetilde{\sk}\la \Reveal(\pk,\MSK,\aux,m)$, and sends $(\widetilde{\CT},\widetilde{\sk})$ to $\cA$.
    \end{itemize}
    \item $\cA$ outputs $b'\in\bit$.
\end{enumerate}
We say that $\Sigma$ is RNC secure if, for any QPT $\cA$, it holds that
\begin{align}
    \advb{\Sigma,\cA}{rec}{nc}(\secp) \seteq \abs{\Pr[\expa{\Sigma,\cA}{rec}{nc}(\secp,0)=1]  - \Pr[\expa{\Sigma,\cA}{rec}{nc}(\secp,1)=1]} \leq \negl(\secp).
\end{align}
\end{definition}

\begin{definition}[Certified Everlasting RNC Security for Certified Everlasting RNCE]\label{def:cert_ever_rec_nc_security_classic}
Let $\Sigma=(\Setup,\keygen,\Enc,\allowbreak \Dec,\Fake, \Reveal,\Delete,\Vrfy)$ be a certified everlasting RNCE scheme.
We consider the following security experiment $\expd{\Sigma,\cA}{cert}{ever}{rec}{nc}(\secp,b)$ against a QPT adversary $\cA_1$ and an unbounded adversary $\cA_2$.
\begin{enumerate}
    \item The challenger runs $(\pk,\MSK)\la\Setup(1^\lambda)$ and sends $\pk$ to $\cA_1$.
    \item $\cA_1$ sends $m\in\Ms$ to the challenger.
    \item The challenger does the following:
    \begin{itemize}
        \item If $b=0$, the challenger generates $(\vk,\ct)\la\Enc(\pk,m)$ and $\sk\la\keygen(\MSK)$, and sends $(\ct,\sk)$ to $\cA_1$.
        \item If $b=1$, the challenger generates $(\vk,\widetilde{\ct},\aux)\la \Fake(\pk)$ and $\widetilde{\sk}\la\Reveal(\pk,\MSK,\aux,m)$,
        and sends $(\widetilde{\CT},\widetilde{\sk})$ to $\cA_1$.
    \end{itemize}
    \item At some point, $\cA_1$ sends $\cert$ to the challenger and its internal state to $\cA_2$.
    \item The challenger computes $\Vrfy(\vk,\cert)$.
    If the output is $\top$, the challenger outputs $\top$ and sends $\MSK$ to $\cA_2$. 
    If the output is $\bot$, the challenger outputs $\bot$ and sends $\bot$ to $\cA_2$.
    \item $\cA_2$ outputs $b'\in\bit$.
    \item If the challenger outputs $\top$, then the output of the experiment is $b'$.
    Otherwise, the output of the experiment is $\bot$.
\end{enumerate}
We say that $\Sigma$ is certified everlasting RNC secure if for any QPT $\cA_1$ and any unbounded $\cA_2$,
it holds that
\begin{align}
    \advd{\Sigma,\cA}{cert}{ever}{rec}{nc}(\secp) \seteq \abs{\Pr[\expd{\Sigma,\cA}{cert}{ever}{rec}{nc}(\secp,0)=1]  - \Pr[\expd{\Sigma,\cA}{cert}{ever}{rec}{nc}(\secp,1)=1]} \leq \negl(\secp).
\end{align}
\end{definition}


\subsection{Construction}\label{sec:const_rnce_classic}
In this section, we construct a certified everlasting RNCE scheme from a certified everlasting PKE scheme~(\cref{def:cert_ever_pke}).
Our construction is similar to that of the secret-key RNCE scheme presented in \cite{C:KNTY19}.
The difference is that we use a certified everlasting PKE scheme instead of an ordinary SKE scheme.

\paragraph{Our certified everlasting RNCE scheme.} 
We construct a certified everlasting RNCE scheme
$\Sigma_{\mathsf{cence}}=(\Setup,\keygen,\allowbreak \Enc,\Dec, \Fake,\Reveal,\Delete,\Vrfy)$
from a certified everlasting PKE scheme $\Sigma_{\mathsf{cepk}}=\PKE.(\keygen,\Enc,\Dec,\Delete,\Vrfy)$, which was introduced in~\cref{def:cert_ever_pke}.
    
\begin{description}
    \item[$\Setup(1^\secp)$:]$ $
    \begin{itemize}
    \item Generate $(\pke.\pk_{i,\alpha},\pke.\sk_{i,\alpha})\la\PKE.\keygen(1^\secp)$ for all $i\in[n]$ and $\alpha\in\bit$.
    \item Output $\pk\seteq\{\pke.\pk_{i,\alpha}\}_{i\in[n],\alpha\in\bit}$ and $\MSK\seteq\{\pke.\sk_{i,\alpha}\}_{i\in[n],\alpha\in\bit}$.
    \end{itemize}
    \item[$\keygen(\MSK)$:]$ $
    \begin{itemize}
        \item Parse $\MSK=\{\pke.\sk_{i,\alpha}\}_{i\in[n],\alpha\in\bit}$.
        \item Generate $x\la \bit^n$.
        \item Output $\sk\seteq(x,\{\pke.\sk_{i,x[i]}\}_{i\in[n]})$.
    \end{itemize}
    \item[$\Enc(\pk,m)$:]$ $
    \begin{itemize}
        \item Parse $\pk= \{\pke.\pk_{i,\alpha}\}_{i\in[n],\alpha\in\bit}$.
        \item Compute $(\pke.\vk_{i,\alpha},\pke.\ct_{i,\alpha})\la\PKE.\Enc(\pke.\pk_{i,\alpha},m[i])$ for all $i\in[n]$ and $\alpha\in\bit$.
        \item Output $\vk\seteq \{\pke.\vk_{i,\alpha}\}_{i\in[n],\alpha\in\bit}$ and $\ct\seteq\{\pke.\ct_{i,\alpha}\}_{i\in[n],\alpha\in\bit}$.
    \end{itemize}
    \item[$\Dec(\sk,\ct)$:]$ $
    \begin{itemize}
        \item Parse $\sk= (x,\{\pke.\sk_{i}\}_{i\in[n]})$ and $\ct=\{\pke.\ct_{i,\alpha}\}_{i\in[n],\alpha\in\bit}$.
        \item Compute $m[i]\la\PKE.\Dec(\pke.\sk_{i},\pke.\ct_{i,x[i]})$ for all $i\in[n]$.
        \item Output $m\seteq m[1]||m[2]||\cdots|| m[n]$.
    \end{itemize}
    \item[$\Fake(\pk)$:]$ $
    \begin{itemize}
        \item Parse $\pk= \{\pke.\pk_{i,\alpha}\}_{i\in[n],\alpha\in\bit}$.
        \item Generate $x^*\la\bit^n$.
        \item Compute $(\pke.\vk_{i,x^*[i]},\pke.\ct_{i,x^*[i]})\la \PKE.\Enc(\pke.\pk_{i,x^*[i]},0)$ and
        $(\pke.\vk_{i,x^*[i]\oplus 1},\pke.\ct_{i,x^*[i]\oplus 1})\la \PKE.\Enc(\pke.\pk_{i,x^*[i]\oplus 1},1)$ for all $i\in[n]$.
        \item Output $\vk\seteq \{\pke.\vk_{i,\alpha}\}_{i\in[n],\alpha\in\bit}$,
        $\widetilde{\ct}\seteq \{\pke.\ct_{i,\alpha}\}_{i\in[n],\alpha\in\bit}$ and $\aux=x^*$.
    \end{itemize}
    \item[$\Reveal(\pk,\MSK,\aux,m)$:]$ $
    \begin{itemize}
        \item Parse $\pk= \{\pke.\pk_{i,\alpha}\}_{i\in[n],\alpha\in\bit}$, $\MSK= \{\pke.\sk_{i,\alpha}\}_{i\in[n],\alpha\in\bit}$ and $\aux=x^*$.
        \item Output $\widetilde{\sk}\seteq \left(x^*\oplus m,\{\pke.\sk_{i,x^*[i]\oplus m[i]}\}_{i\in[n]}\right)$. 
    \end{itemize}
    \item[$\Delete(\ct)$:]$ $
    \begin{itemize}
        \item Parse $\ct=\{\pke.\ct_{i,\alpha}\}_{i\in[n],\alpha\in\bit}$.
        \item Compute $\pke.\cert_{i,\alpha}\la\PKE.\Delete(\pke.\ct_{i,\alpha})$ for all $i\in[n]$ and $\alpha\in\bit$.
        \item Output $\cert\seteq\{\pke.\cert_{i,\alpha}\}_{i\in[n],\alpha\in\bit}$.
    \end{itemize}
    \item[$\Vrfy(\vk,\cert)$:]$ $
    \begin{itemize}
        \item Parse $\vk=\{\pke.\vk_{i,\alpha}\}_{i\in[n],\alpha\in\bit}$ and $\cert=\{\pke.\cert_{i,\alpha}\}_{i\in[n],\alpha\in\bit}$.
        \item Compute $\top/\bot\la\PKE.\Vrfy(\pke.\vk_{i,\alpha},\pke.\cert_{i,\alpha})$ for all $i\in[n]$ and $\alpha\in\bit$.
        If all results are $\top$, $\Vrfy(\vk,\cert)$ outputs $\top$. Otherwise, it outputs $\bot$. 
    \end{itemize}
\end{description}

\paragraph{Correctness:}
Correctness easily follows from that of $\Sigma_{\mathsf{cepk}}$.

\paragraph{Security:}
The following two theorems hold.
\begin{theorem}\label{thm:comp_security_rnce_classic}
If $\Sigma_{\mathsf{cepk}}$ is IND-CPA secure~(\cref{def:IND-CPA_security_cert_ever_pke}), $\Sigma_{\mathsf{cence}}$ is RNC secure~(\cref{def:rec_nc_security_classic}).
\end{theorem}
Its proof is similar to that of \cref{thm:ever_security_rnce_classic}, and therefore we omit it.

\begin{theorem}\label{thm:ever_security_rnce_classic}
If $\Sigma_{\mathsf{cepk}}$ is certified everlasting IND-CPA secure~(\cref{def:cert_ever_security_cert_ever_pke}), $\Sigma_{\mathsf{cence}}$ is certified everlasting RNC secure~(\cref{def:cert_ever_rec_nc_security_classic}).
\end{theorem}
Its proof is given in \cref{sec:proof_rnce}.


\section{Certified Everlasting Garbling Scheme}\label{sec:garbling_scheme}
In \cref{sec:def_garbled}, we define certified everlasting garbling scheme.
In \cref{sec:const_garbling}, we construct a certified everlasting garbling scheme from a certified everlasting SKE scheme.

\subsection{Definition}\label{sec:def_garbled}
We define certified everlasting garbling schemes below. An important difference from ordinal classical garbling schemes is that the garbled circuit $\tilde{\mathcal{C}}$ (i.e., an output of $\Garble$) 
is a quantum state.
\begin{definition}[Certified Everlasting Garbling Scheme (Syntax)]\label{def:cert_ever_garbled}
Let $\lambda$ be a security parameter and $p,q,r$ and $s$ be polynomials. 
Let $\Cs_n$ be a family of circuits that take $n$-bit inputs.
A certified everlasting garbling scheme is a tuple of algorithms $\Sigma=(\Samp,\Garble,\Eval,\Delete,\Vrfy)$
with label space $\mathcal{L}\seteq\bit^{p(\lambda)}$, garbled circuit space $\mathcal{C}\seteq \cQ^{\otimes q(\lambda)}$,
verification key space $\mathcal{VK}\seteq \bit^{r(\lambda)}$ and deletion certificate space $\mathcal{D}\seteq\cQ^{\otimes s(\lambda)}$.
\begin{description}
\item[$\Samp(1^\secp)\ra\{L_{i,\alpha}\}_{i\in[n],\alpha\in\bit}$:]
The sampling algorithm takes a security parameter $1^\lambda$ as input, and outputs $2n$ labels $\{L_{i,\alpha}\}_{i\in[n],\alpha\in\bit}$ with $L_{i,\alpha}\in\mathcal{L}$ for each $i\in[n]$ and $\alpha\in\bit$.
\item[$\Garble(1^{\secp},C,\{L_{i,\alpha}\}_{i\in[n],\alpha\in\{0,1\}})\ra (\widetilde{C},\vk)$:] 
The garbling algorithm takes $1^\lambda$, a circuit $C\in\Cs_n$ and $2n$ labels $\{L_{i,\alpha}\}_{i\in[n],\alpha\in\bit}$ as input, and outputs a garbled circuit $\widetilde{C}\in\mathcal{C}$ and a verification key $\vk\in\mathcal{VK}$.
\item[$\Eval(\widetilde{C},\{L_{i,x_i}\}_{i\in[n]})\ra y$:]
The evaluation algorithm takes $\widetilde{C}$ and $n$ labels $\{L_{i,x_i}\}_{i\in[n]}$ where $x_i\in\{0,1\}$ as input,
and outputs $y$.
\item[$\Delete(\widetilde{C})\ra \cert$:]
The deletion algorithm takes $\widetilde{C}$ as input, and outputs a certificate $\cert\in\mathcal{D}$.
\item[$\Vrfy(\vk,\cert)\ra \top$ $\mbox{\bf  or }\bot$:]
The verification algorithm takes $\vk$ and $\cert$ as input, and outputs $\top$ or $\bot$.
\end{description}
\end{definition}

We require that a certified everlasting garbling scheme satisfies correctness defined below.
\begin{definition}[Correctness for Certified Everlasting Garbling Scheme]\label{def:correctness_cert_ever_garbled}
There are three types of correctness, namely, evaluation correctness, verification correctness, and modification correctness.

\paragraph{Evaluation Correctness:}
There exists a negligible function $\negl$ such that for any $\secp\in\N$, $C\in\Cs_n$ and $x\in \bit^n$,
\begin{align}
\Pr\left[
y\neq C(x)
\ \middle |
\begin{array}{ll}
\{L_{i,\alpha}\}_{i\in[n],\alpha\in\bit}\la\Samp(1^\secp)\\
(\widetilde{C},\vk)\lrun \Garble(1^\secp,C,\{L_{i,\alpha}\}_{i\in[n],\alpha\in\{0,1\}})\\
y\la\Eval(\widetilde{C},\{L_{i,x_i}\}_{i\in[n]})
\end{array}
\right] 
\leq\negl(\secp).
\end{align}

\paragraph{Verification Correctness:} There exists a negligible function $\negl$ such that for any $\secp\in\N$,
\begin{align}
\Pr\left[
\Vrfy(\vk,\cert)=\bot
\ \middle |
\begin{array}{ll}
\{L_{i,\alpha}\}_{i\in[n],\alpha\in\bit}\la\Samp(1^\secp)\\
(\widetilde{C},\vk)\la\Garble(1^{\secp},C,\{L_{i,\alpha}\}_{i\in[n],\alpha\in\{0,1\}}) \\
\cert\la\Delete(\widetilde{C})
\end{array}
\right] 
\leq
\negl(\secp).
\end{align}

\paragraph{Modification Correctness:}
There exists a negligible function $\negl$ and a QPT algorithm $\Modify$ such that for any $\secp\in \N$, 
\begin{align}
\Pr\left[
\Vrfy(\vk,\cert^*)=\bot
\ \middle |
\begin{array}{ll}
\{L_{i,\alpha}\}_{i\in[n],\alpha\in\bit}\la\Samp(1^\secp)\\
(\widetilde{C},\vk)\la\Garble(1^{\secp},C,\{L_{i,\alpha}\}_{i\in[n],\alpha\in\{0,1\}}) \\
a,b\la\bit^{q(\lambda)}\\
\cert \lrun \Delete(Z^bX^a\widetilde{C} X^aZ^b)\\
\cert^*\lrun \Modify(a,b,\cert)  
\end{array}
\right] 
\leq
\negl(\secp).
\end{align}
\end{definition}
\begin{remark}
Minimum requirements for correctness are evaluation correctness and verification correctness.
In this paper, however, we also require modification correctness,
because we need modification correctness for the construction of functional encryption in \cref{sec:const_fe_adapt}.
\end{remark}


As security, we consider two definitions, 
\cref{def:sel_sec_ever_garb} 
and
\cref{def:ever_sec_ever_garb} given below.
The former is just the ordinal selective security 
and the latter is the certified everlasting security
that we newly define in this paper.
Roughly, the everlasting security guarantees that any QPT adversary with the garbled circuit $\widetilde{C}$ and the labels $\{L_{i,x[i]}\}_{i\in[n]}$ cannot obtain any information
beyond $C(x)$ even if it becomes  computationally unbounded after it issues a valid certificate.

\begin{definition}[Selective Security for Certified Everlasting Garbling Scheme]\label{def:sel_sec_ever_garb}
Let $\Sigma=(\Samp,\Garble,\Eval,\Delete,\Vrfy)$ be a certified everlasting garbling scheme.
We consider the following security experiment $\Expt{\Sigma,\cA}{\mathsf{selct}}(1^\secp,b)$ against a QPT adversary $\cA$. Let $\Sim$ be a QPT algorithm.
\begin{enumerate}
    \item $\cA$ sends a circuit $C\in\Cs_n$ and an input $x\in\bit^n$ to the challenger.
    \item The challenger computes $\{L_{i,\alpha}\}_{i\in[n],\alpha\in\bit}\la\Samp(1^\secp)$.
    \item If $b=0$, the challenger computes $(\widetilde{C},\vk)\la\Garble(1^\secp,C,\{L_{i,\alpha}\}_{i\in[n],\alpha\in\bit})$, and returns $(\widetilde{C},\{L_{i,x_i}\}_{i\in[n]})$ to $\cA$.
    If $b=1$, the challenger computes $\widetilde{C}\la\Sim(1^\secp,1^{|C|},C(x),\{L_{i,x_i}\}_{i\in[n]})$,
    and returns $(\widetilde{C},\{L_{i,x_i}\}_{i\in[n]})$ to $\cA$.
    \item $\cA$ outputs $b'\in\bit$. The experiment outputs $b'$.
\end{enumerate}
We say that $\Sigma$ is selective secure if there exists a QPT simulator $\Sim$ such that for any QPT adversary $\cA$ it holds that
\begin{align}
    \adva{\Sigma,\cA}{selct}(\secp)\seteq \abs{\Pr[\Expt{\Sigma,\cA}{\mathsf{selct}}(1^\secp,0)=1]-\Pr[\Expt{\Sigma,\cA}{\mathsf{selct}}(1^\secp,1)=1]}\leq\negl(\secp).
\end{align}
\end{definition}

\begin{definition}[Certified Everlasting Selective Security for Certified Everlasting Garbling Scheme]\label{def:ever_sec_ever_garb}
Let $\Sigma=(\Samp,\Garble,\Eval,\Delete,\Vrfy)$ be a certified everlasting garbling scheme.
We consider the following security experiment $\expb{\cA,\Sigma}{cert}{ever}{selct}(1^\secp,b)$ against a QPT adversary $\cA_1$ and an unbounded adversary $\cA_2$. Let $\Sim$ be a QPT algorithm.
\begin{enumerate}
    \item $\cA_1$ sends a circuit $C\in\Cs_n$ and an input $x\in\bit^n$ to the challenger.
    \item The challenger computes $\{L_{i,\alpha}\}_{i\in[n],\alpha\in\bit}\la\Samp(1^\secp)$.
    \item If $b=0$, the challenger computes $(\widetilde{C},\vk)\la\Garble(1^\secp,C,\{L_{i,\alpha}\}_{i\in[n],\alpha\in\bit})$, and returns $(\widetilde{C},\{L_{i,x_i}\}_{i\in[n]})$ to $\cA_1$.
    If $b=1$, the challenger computes $(\widetilde{C},\vk)\la\Sim(1^\secp,1^{|C|},C(x),\{L_{i,x_i}\}_{i\in[n]})$,
    and returns $(\widetilde{C},\{L_{i,x_i}\}_{i\in[n]})$ to $\cA_1$.
    \item At some point, $\cA_1$ sends $\cert$ to the challenger, and sends the internal state to $\cA_2$.
    \item The challenger computes $\Vrfy(\vk,\cert)$.
    If the output is $\bot$, then the challenger outputs $\bot$, and sends $\bot$ to $\cA_2$.
    Otherwise, the challenger outputs $\top$, and sends $\top$ to $\cA_2$.
    \item $\cA_2$ outputs $b'\in\bit$.
    \item If the challenger outputs $\top$, then the output of the experiment is $b'$. 
    Otherwise, the output of the experiment is $\bot$.
\end{enumerate}
We say that $\Sigma$ is certified everlasting selective secure if there exists a QPT simulator $\Sim$ such that for any QPT $\cA_1$ and any unbounded $\cA_2$ it holds that
\begin{align}
    \advc{\Sigma,\cA}{cert}{ever}{selct}(\secp)\seteq
    \abs{\Pr[\expb{\cA,\Sigma}{cert}{ever}{selct}(1^\secp,0)=1]-\Pr[\expb{\cA,\Sigma}{cert}{ever}{selct}(1^\secp,1)=1]}\leq\negl(\secp).
\end{align}
\end{definition}


\subsection{Construction}\label{sec:const_garbling}
In this section, we construct a certified everlasting garbling scheme from a certified everlasting SKE scheme~(\cref{def:cert_ever_ske}).
Our construction is similar to Yao's construction of an ordinary garbling scheme \cite{FOCS:Yao86}, but there are two important differences.
First, we use a certified everlasting SKE scheme instead of an ordinary SKE scheme.
Second, we use XOR secret sharing, although \cite{FOCS:Yao86} used double encryption. The reason why we cannot use double encryption is
that our certified everlasting SKE scheme has quantum ciphertext and classical plaintext.

Before introducing our construction, let us quickly review notations for circuits.
Let $C$ be a boolean circuit. 
A boolean circuit $C$ consists of gates, $\mathsf{gate}_1,\mathsf{gate}_2,\cdots,\mathsf{gate}_q$,
where $q$ is the number of gates in the circuit.
Here, $\mathsf{gate}_i\seteq (g,w_a,w_b,w_c)$, where $g:\bit^2\ra\bit$ is a function, $w_a$, $w_b$ are the incoming wires, and $w_c$ is the outgoing wire.
(The number of outgoing wires is not necessarily one. There can be many outgoing wires, but we use the same label $w_c$ for all outgoing wires.) 
We say $C$ is leveled if each gate has an associated level and any gate at level $\ell$ has incoming wires only from gates at level $\ell-1$ and outgoing wires only to gates at level $\ell+1$. 
Let $\mathsf{out}_1,\mathsf{out}_2,\cdots,\mathsf{out}_m$ be the $m$ output wires.
For any $x\in\bit^n$, $C(x)$ is the output of the circuit $C$ on input $x$.
We consider that $\mathsf{gate_1},\mathsf{gate_2},\cdots,\mathsf{gate}_q$ are arranged in the ascending order of the level.

\paragraph{Our certified everlasting garbling scheme.}
We construct a certified everlasting garbling scheme $\Sigma_{\mathsf{cegc}}=(\Samp,\Garble,\allowbreak \Eval,\Delete,\Vrfy)$
from a certified everlasting SKE scheme
$\Sigma_{\mathsf{cesk}}=\SKE.(\keygen,\Enc,\Dec,\Delete,\Vrfy)$~(\cref{def:cert_ever_ske}).
Let $\mathcal{K}$ be the key space of $\Sigma_{\mathsf{cesk}}$.
Let $C$ be a leveled boolean circuit.
Let $n,m,q$, and $p$ be the input size,
the output size,
the number of gates,
and
the total number of wires
of $C$, 
respectively.
 

\begin{description}
\item[$\Samp(1^\secp)$:] $ $
\begin{itemize}
    \item For each $i\in[n]$ and $\sigma\in\bit$, generate $\ske.\sk_{i}^\sigma\la\SKE.\keygen(1^\secp)$.
    \item Output $\{L_{i,\sigma}\}_{i\in[n],\sigma\in\bit}\seteq\{ \ske.\sk_i^\sigma\}_{i\in [n],\sigma\in\bit}$.
\end{itemize}
\item[$\Garble(1^\secp,C,\{L_{i,\sigma}\}_{i\in[n],\sigma\in\bit})$:] $ $
\begin{itemize}
    \item For each $i\in\{n+1,\cdots, p\}$ and $\sigma\in\bit$, generate $\ske.\sk_{i}^\sigma\la\SKE.\keygen(1^\secp)$.
    \item For each $i\in [q]$, 
    compute 
    \begin{align}
    (\vk_i,\widetilde{g}_i)\la \mathsf{GateGrbl}(\mathsf{gate}_i, \{\ske.\sk_a^\sigma,\ske.\sk_b^\sigma,\ske.\sk_c^\sigma\}_{\sigma\in\bit}), 
    \end{align}
    where $\mathsf{gate}_i=(g,w_a,w_b,w_c)$ and $\mathsf{GateGrbl}$ is described in Fig~\ref{fig:Garble_circuit}.
    \item For each $i\in[m]$,
    set $\widetilde{d_i}\seteq[(\ske.\sk_{\mathsf{out}_i}^0, 0),(\ske.\sk_{\mathsf{out}_i}^1, 1)]$.
    \item Output $\widetilde{C}\seteq (\{\widetilde{g_i}\}_{i\in [q]},\{\widetilde{d}_i\}_{i\in [m]})$
    and 
    $\vk\seteq\{\vk_i\}_{i\in[q]}$.
\end{itemize}
\item[$\Eval(\widetilde{C},\{L_{i,x_i}\}_{i\in [n]})$:] $ $
\begin{itemize}
    \item Parse $\widetilde{C}= (\{\widetilde{g_i}\}_{i\in [q]},\{\widetilde{d}_i\}_{i\in [m]})$ and $\{L_{i,x_i}\}_{i\in[n]}=\{\ske.\sk'_i\}_{i\in [n]}$.
    \item For each $i\in[q]$, 
    compute $\ske.\sk'_{c}\la\mathsf{GateEval}(\widetilde{g_i}, \ske.\sk'_a,\ske.\sk'_b)$ in the ascending order of the level, where $\mathsf{GateEval}$ is described in Fig~\ref{fig:Gate_Eval}.
    If $\ske.\sk'_{c}=\bot$, output $\bot$ and abort.
    \item For each $i\in[m]$, set $y[i]=\sigma$ if $\ske.\sk'_{\mathsf{out}_i}=\ske.\sk_{\mathsf{out}_i}^\sigma$. Otherwise, set $y[i]=\bot$, and abort.
    \item Output $y\seteq y[1]||y[2]||\cdots|| y[m]$.
\end{itemize}
\item[$\Delete(\widetilde{C})$:] $ $
\begin{itemize}
    \item Parse $\widetilde{C}= (\{\widetilde{g_i}\}_{i\in [q]},\{\widetilde{d}_i\}_{i\in [m]})$.
    \item For each $i\in[q]$, compute $\cert_i\la\mathsf{GateDel}(\widetilde{g_i})$, where $\mathsf{GateDel}$ is described in Fig~\ref{fig:Gate_Del}.
    \item Output $\cert\seteq\{\cert_i\}_{i\in [q]}$.
\end{itemize}
\item[$\Vrfy(\vk,\cert)$:] $ $
\begin{itemize}
    \item Parse
    $\vk=\{\vk_i\}_{i\in[q]}$
    and $\cert= \{\cert_i\}_{i\in[q]}$.
    \item For each $i\in[q]$,
    compute $\bot/\top\la\mathsf{GateVrfy}(\vk_i,\cert_i)$, where $\mathsf{GateVrfy}$ is described in Fig~\ref{fig:Gate_Verify}.
    \item If $\top\la\mathsf{GateVrfy}(\vk_i,\cert_i)$ for all $i\in[q]$, then output $\top$.
    Otherwise, output $\bot$. 
\end{itemize}
\end{description}

\protocol{Gate Garbling Circuit $\mathsf{GateGrbl}$
}
{The description of $\mathsf{GateGrbl}$}
{fig:Garble_circuit}
{
\begin{description}
\item[Input:] $\mathsf{gate}_i,\{\ske.\sk_a^\sigma,\ske.\sk_b^\sigma,\ske.\sk_c^\sigma\}_{\sigma\in\bit}$.
\item[Output:] $\widetilde{g_i}$ and $\vk_i$.
\end{description}
\begin{enumerate}
\item Parse $\mathsf{gate}_i=(g,w_a,w_b,w_c)$.
\item Sample $\gamma_i\la\mathsf{S}_4$.\footnote{$\mathsf{S}_4$ is the symmetric group of order $4$.}
\item For each $\sigma_a,\sigma_b\in\bit$, sample $p^{\sigma_{a},\sigma_b}_{c}\la\Ks$.
\item For each $\sigma_a,\sigma_b\in\bit$, compute $(\ske.\vk_{a}^{\sigma_a,\sigma_b},\ske.\ct_{a}^{\sigma_a,\sigma_b})\la\SKE.\Enc(\ske.\sk_a^{\sigma_a},p^{\sigma_a,\sigma_b}_{c})$
and 
$(\ske.\vk_{b}^{\sigma_a,\sigma_b},\ske.\ct_b^{\sigma_a,\sigma_b})\la
\SKE.\Enc(\ske.\sk_b^{\sigma_b},p^{\sigma_a,\sigma_b}_{c}\oplus \ske.\sk_c^{g(\sigma_a,\sigma_b)})$.
\item Output $\widetilde{g_i}\seteq \{\ske.\ct_a^{\sigma_a,\sigma_b},\ske.\ct_b^{\sigma_a,\sigma_b}\}_{\sigma_a,\sigma_b\in\bit}$ in the permutated order of $\gamma_i$ and 
$\vk_i\seteq\{\ske.\vk_{a}^{\sigma_a,\sigma_b},\ske.\vk_{b}^{\sigma_a,\sigma_b}\}_{\sigma_a,\sigma_b\in\bit}$ in the permutated order of $\gamma_i$.
\end{enumerate}
}

\protocol{Gate Evaluating Circuit $\mathsf{GateEval}$
}
{The description of $\mathsf{GateEval}$}
{fig:Gate_Eval}
{
\begin{description}
\item[Input:] A garbled gate $\widetilde{g}_i$ and $(\ske.\sk'_a,\ske.\sk'_b)$.
\item[Output:] $\ske.\sk_c$ or $\bot$.
\end{description}
\begin{enumerate}
\item Parse $\widetilde{g_i}= \{\ske.\ct_a^{\sigma_a,\sigma_b},\ske.\ct_b^{\sigma_a,\sigma_b}\}_{\sigma_a,\sigma_b\in\bit}$.
\item For each $\sigma_a,\sigma_b\in\bit$, compute 
      $q_a^{\sigma_a,\sigma_b}\la\SKE.\Dec(\ske.\sk'_a,\ske.\ct_a^{\sigma_a,\sigma_b})$ and
      $q_b^{\sigma_a,\sigma_b}\la\SKE.\Dec(\ske.\sk'_b,\ske.\ct_b^{\sigma_a,\sigma_b})$.
\item If there exists a unique pair $(\sigma_a,\sigma_b)\in\bit^2$ such that $q_a^{\sigma_a,\sigma_b}\neq \bot$ and $q_b^{\sigma_a,\sigma_b}\neq\bot$, then
      compute $\ske.\sk_{c}^{'\sigma_a,\sigma_b}\seteq q_a^{\sigma_a,\sigma_b}\oplus q_b^{\sigma_a,\sigma_b}$ and output $\ske.\sk'_c\seteq \ske.\sk_c^{'\sigma_a,\sigma_b}$.
      Otherwise, output $\ske.\sk'_c\seteq\bot$.
\end{enumerate}
}

\protocol{Gate Deletion Circuit $\mathsf{GateDel}$
}
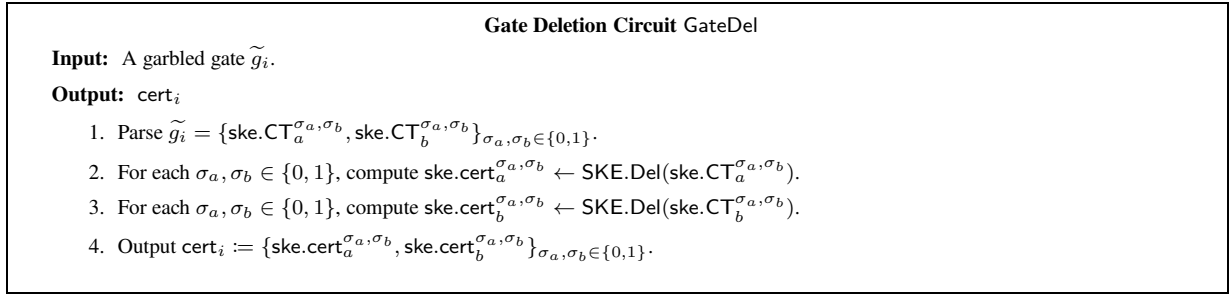
{The description of $\mathsf{GateDel}$}
{fig:Gate_Del}
{
\begin{description}
\item[Input:] A garbled gate $\widetilde{g}_i$.
\item[Output:] $\cert_i$
\end{description}
\begin{enumerate}
\item Parse $\widetilde{g_i}= \{\ske.\ct_a^{\sigma_a,\sigma_b},\ske.\ct_b^{\sigma_a,\sigma_b}\}_{\sigma_a,\sigma_b\in\bit}$.
\item For each $\sigma_a,\sigma_b\in\bit$, compute $\ske.\cert_a^{\sigma_a,\sigma_b}\la\SKE.\Delete(\ske.\ct_a^{\sigma_a,\sigma_b})$.
\item For each $\sigma_a,\sigma_b\in\bit$, compute $\ske.\cert_b^{\sigma_a,\sigma_b}\la\SKE.\Delete(\ske.\ct_b^{\sigma_a,\sigma_b})$.
\item Output $\cert_i\seteq \{\ske.\cert_a^{\sigma_a,\sigma_b},\ske.\cert_b^{\sigma_a,\sigma_b}\}_{\sigma_a,\sigma_b\in\bit}$.
\end{enumerate}
}

\protocol{Gate Verification Circuit $\mathsf{GateVrfy}$
}
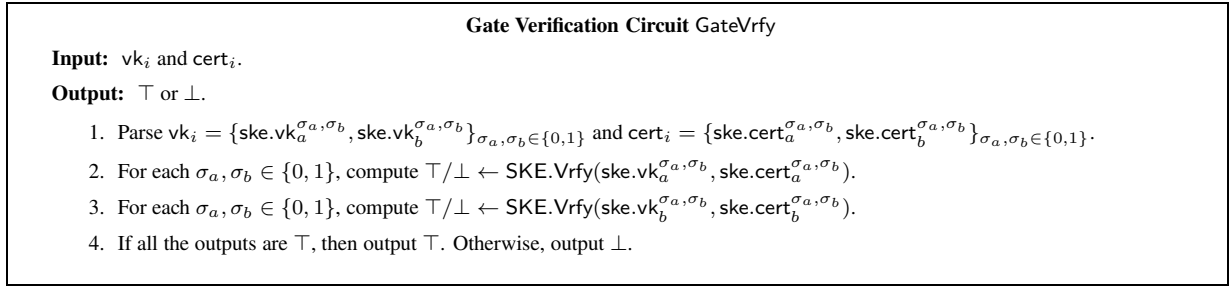
{The description of $\mathsf{GateVrfy}$}
{fig:Gate_Verify}
{
\begin{description}
\item[Input:] $\vk_i$ and $\cert_i$.
\item[Output:] $\top$ or $\bot$.
\end{description}
\begin{enumerate}
\item Parse
$\vk_i=\{\ske.\vk_{a}^{\sigma_a,\sigma_b},\ske.\vk_{b}^{\sigma_a,\sigma_b}\}_{\sigma_a,\sigma_b\in\bit}$ and
$\cert_i= \{\ske.\cert_a^{\sigma_a,\sigma_b},\ske.\cert_b^{\sigma_a,\sigma_b}\}_{\sigma_a,\sigma_b\in\bit}$.
\item For each $\sigma_a,\sigma_b\in\bit$, compute $\top/\bot\la\SKE.\Vrfy(\ske.\vk_{a}^{\sigma_a,\sigma_b},\ske.\cert_a^{\sigma_a,\sigma_b})$.
\item For each $\sigma_a,\sigma_b\in\bit$, compute $\top/\bot\la\SKE.\Vrfy(\ske.\vk_b^{\sigma_a,\sigma_b},\ske.\cert_b^{\sigma_a,\sigma_b})$.
\item If all the outputs are $\top$, then output $\top$.
Otherwise, output  $\bot$.
\end{enumerate}
}

\paragraph{Correctness:}
Correctness easily follows from that of $\Sigma_{\mathsf{cesk}}$.

\paragraph{Security:}
The following two theorems hold.
\begin{theorem}\label{thm:garble_comp_seucirty}
If $\Sigma_{\mathsf{cesk}}$ satisfies the IND-CPA security~(\cref{def:IND-CPA_security_cert_ever_ske}), 
$\Sigma_{\mathsf{cegc}}$ satisfies the selective security~(\cref{def:sel_sec_ever_garb}).
\end{theorem}
Its proof is similar to that of \cref{thm:garble_ever_security},
and therefore we omit it.

\begin{theorem}\label{thm:garble_ever_security}
If $\Sigma_{\mathsf{cesk}}$ satisfies the certified everlasting IND-CPA security~(\cref{def:cert_ever_security_cert_ever_ske}),
$\Sigma_{\mathsf{cegc}}$ satisfies the certified everlasting selective security~(\cref{def:ever_sec_ever_garb}).
\end{theorem}
Its proof is given in \cref{sec:parallel}.


\section{Certified Everlasting Functional Encryption}\label{sec:function_encryption}
In this section, we define and construct certified everlasting functional encryption (FE).
In \cref{sec:def_ever_fe}, we define certified everlasting FE.
In \cref{sec:const_func_non_adapt}, we construct a 1-bounded certified everlasting FE scheme with non-adaptive security for all $\Ppoly$ circuits
from a certified everlasting garbling scheme~(\cref{def:cert_ever_garbled}) and a certified everlasting PKE scheme~(\cref{def:cert_ever_pke}).
In \cref{sec:const_fe_adapt}, we change it to the adaptive one by using a certified everlasting RNCE scheme~(\cref{def:cert_ever_rnce_classic}).
In \cref{sec:const_multi_fe}, we further change it to a $q$-bounded certified everlasting FE for all $\NCone$ circuits 
by using a multipary computation scheme.

\subsection{Definition}\label{sec:def_ever_fe}
\begin{definition}[$q$-Bounded Certified Everlasting FE (Syntax)]\label{def:cert_ever_func}
Let $\secp$ be a security parameter and let $p$, $r$, $s$, $t$, $u$ and $v$ be polynomials.
Let $q$ be a polynomial of the security parameter $\secp$.
A $q$-bounded certified everlasting FE scheme for a class $\mathcal{F}$ of functions 
is a tuple of algorithms
$\Sigma=(\Setup,\keygen,\Enc,\Dec,\Delete,\Vrfy)$
with plaintext space $\Ms\seteq\bit^n$, ciphertext space $\Cs:= \cQ^{\otimes p(\lambda)}$,
master public key space $\mathcal{MPK}:=\{0,1\}^{r(\lambda)}$,
master secret key space $\mathcal{MSK}\seteq \bit^{s(\lambda)}$,
secret key space $\mathcal{SK}\seteq \bit^{t(\lambda)}$,
verification key space $\mathcal{VK}\seteq \bit^{u(\lambda)}$
and deletion certificate space $\mathcal{D}:= \cQ^{\otimes v(\lambda)}$.
\begin{description}
    \item[$\Setup(1^\secp)\ra(\MPK,\MSK)$:]
    The setup algorithm takes the security parameter $1^\secp$ as input, and outputs a master public key $\MPK\in\mathcal{MPK}$
    and a master secret key $\MSK\in\mathcal{MSK}$.
    \item[$\keygen(\MSK,f)\ra\sk_f$:]
    The keygeneration algorithm takes $\MSK$ and $f\in\mathcal{F}$ as input,
    and outputs a secret key $\sk_f\in\mathcal{SK}$.
    \item[$\Enc(\MPK,m)\ra (\vk,\CT)$:]
    The encryption algorithm takes $\MPK$ and $m\in\Ms$ as input,
     and outputs a verification key $\vk\in\mathcal{VK}$ and a ciphertext $\CT\in \Cs$.
    \item[$\Dec(\sk_f,\CT)\ra y$ $\mbox{\bf  or }\bot$:]
    The decryption algorithm takes $\sk_f$ and $\CT$ as input, and outputs $y$ or $\bot$.
    \item[$\Delete(\CT)\ra \cert$:]
    The deletion algorithm takes $\CT$ as input, and outputs a certificate $\cert\in\mathcal{D}$.
    \item[$\Vrfy(\vk,\cert)\ra \top \mbox{ {\bf or} } \bot$:]
    The verification algorithm takes $\vk$ and $\cert$ as input,
    and outputs $\top$ or $\bot$.
\end{description}
\end{definition}

We require that a certified everlasting FE scheme satisfies correctness defined below.
\begin{definition}[Correctness for Certified Everlasting FE]\label{def:correctness_cert_ever_func}
There are three types of correctness, namely,
evaluation correctness,
verification correctness,
and modification correctness.

\paragraph{Evaluation Correctness:}
There exists a negligible function $\negl$ such that for any $\secp\in \N$, $m\in\Ms$ and $f\in\mathcal{F}$,
\begin{align}
\Pr\left[
y\ne f(m)
\ \middle |
\begin{array}{ll}
(\MPK,\MSK)\la\Setup(1^\secp)\\
\sk_f\la\keygen(\MSK,f)\\
(\vk,\CT)\lrun \Enc(\MPK,m)\\
y\la\Dec(\sk_f,\ct)
\end{array}
\right] 
\le\negl(\secp).
\end{align}

\paragraph{Verification Correctness:}
There exists a negligible function $\negl$ such that for any $\secp\in\N$ and $m\in \Ms$,
\begin{align}
\Pr\left[
\Vrfy(\vk,\cert)=\bot
\ \middle |
\begin{array}{ll}
(\MPK,\MSK)\la\Setup(1^\secp)\\
(\vk,\CT)\lrun \Enc(\MPK,m)\\
\cert\la\Delete(\CT)
\end{array}
\right] 
\leq
\negl(\secp).
\end{align}

\paragraph{Modification Correctness:}
There exists a negligible function $\negl$ and a QPT algorithm $\Modify$ such that for any $\secp\in \N$ and $m\in\Ms$,
\begin{align}
\Pr\left[
\Vrfy(\vk,\cert^*)=\bot
\ \middle |
\begin{array}{ll}
(\MPK,\MSK)\la\Setup(1^\secp)\\
\sk_f\lrun \keygen(\MSK,f)\\
(\vk,\ct) \lrun \Enc(\MPK,m)\\
a,b\la\bit^{p(\lambda)}\\
\cert \lrun \Delete(Z^bX^a\ct X^aZ^b)\\
\cert^*\lrun \Modify(a,b,\cert)  
\end{array}
\right] 
\leq
\negl(\secp).
\end{align}

\end{definition}
\begin{remark}
Minimum requirements for correctness are evaluation correctness and verification correctness.
In this paper, however, we also require modification correctness,
because we need modification correctness for the construction of certified everlasting FE in \cref{sec:const_fe_adapt}.
\end{remark}

\begin{remark}
In FE, we usually want to run $\Dec$ algorithm for many different functions $f$ on the same ciphertext $\ct$.
One might think that the quantum $\ct$ is destroyed by $\Dec$ algorithm, and therefore it can be used only once.
However, it is easy to see that $\Dec$ algorithm can be always modified so that it does not disturb the quantum state $\ct$
by using the gentle measurement lemma~\cite{Mark2011}.  
\end{remark}

In this paper, we introduce four types of definitions of security, \cref{def:non_ada_security_cert_ever_func_sim,def:ada_security_cert_ever_func_sim,def:ever_non_ada_security_cert_ever_func_sim,def:ever_ada_security_cert_ever_func_sim}.
(We note that these securities are simulation based ones defined in \cite{C:GorVaiWee12}.)
The first two definitions (\cref{def:non_ada_security_cert_ever_func_sim,def:ada_security_cert_ever_func_sim}) are ordinal security definitions of FE. 
(\cref{def:non_ada_security_cert_ever_func_sim} is non-adaptive one and \cref{def:ada_security_cert_ever_func_sim} is adaptive one.)
The third and fourth definitions (\cref{def:ever_non_ada_security_cert_ever_func_sim,def:ever_ada_security_cert_ever_func_sim}) are security definitions with certified everlasting security that we
newly introduce in this paper.
(\cref{def:ever_non_ada_security_cert_ever_func_sim} is non-adaptive one and \cref{def:ever_ada_security_cert_ever_func_sim} is adaptive one.)

\begin{definition}[$q$-Bounded Non-Adaptive Security for Certified Everlasting FE (Simulation Base)\cite{C:GorVaiWee12}]\label{def:non_ada_security_cert_ever_func_sim}
Let $q$ be a polynomial of $\lambda$.
Let $\Sigma=(\Setup,\keygen,\Enc,\Dec,\Delete,\Vrfy)$ be a $q$-bounded certified everlasting FE scheme.
We consider the following security experiment $\expa{\Sigma,\cA}{non}{\mathsf{adapt}}(\secp,b)$ against a QPT adversary $\cA$.
Let $\Sim$ be a QPT algorithm.
\begin{enumerate}
    \item The challenger runs $(\MPK,\MSK)\la \Setup(1^\secp)$ and sends $\MPK$ to $\cA$.
    \item $\cA$ is allowed to make arbitrary key queries at most $q$ times.
    For the $\ell$-th key query, the challenger receives $f_\ell\in\mathcal{F}$, 
    computes $\sk_{f_\ell}\la\keygen(\MSK,f_\ell)$,
    and sends $\sk_{f_\ell}$ to $\cA$.
    Let $q^*$ be the number of times that $\cA$ makes key queries.
    Let $\mathcal{V}\seteq \{y_i\seteq f_i(m),f_i,\sk_{f_i}\}_{i\in[q^*]}$.
    \item $\cA$ chooses $m\in\Ms$ and sends $m$ to the challenger.
    \item The experiment works as follows:
    \begin{itemize}
        \item If $b=0$, the challenger computes $(\vk,\CT)\la \Enc(\MPK,m)$, and sends $\CT$ to $\cA$.
        \item If $b=1$, the challenger computes $\ct\la \Sim(\MPK,\mathcal{V},1^{|m|})$, and sends $\ct$ to $\cA$.
    \end{itemize}
    \item $\cA$ outputs $b'\in\bit$. The output of the experiment is $b'$.
\end{enumerate}
We say that $\Sigma$ is $q$-bounded non-adaptive secure if there exists a $QPT$ simulator $\Sim$ such that for any $QPT$ adversary $\cA$ it holds that
\begin{align}
\advb{\Sigma,\cA}{non}{adapt}(\secp) \seteq \abs{\Pr[\expa{\Sigma,\cA}{non}{\mathsf{adapt}}(\secp,0)=1]  - \Pr[\expa{\Sigma,\cA}{non}{\mathsf{adapt}}(\secp,1)=1]} \leq \negl(\secp).
\end{align}
\end{definition}

\begin{definition}[$q$-Bounded Adaptive Security for Certified Everlasting FE (Simulation Base)\cite{C:GorVaiWee12}]\label{def:ada_security_cert_ever_func_sim}
Let $q$ be a polynomial of $\lambda$.
Let $\Sigma=(\Setup,\keygen,\Enc,\Dec,\Delete,\Vrfy)$ be a $q$-bounded certified everlasting FE scheme.
We consider the following security experiment $\Expt{\Sigma,\cA}{\mathsf{adapt}}(\secp,b)$ against a QPT adversary $\cA$.
Let $\Sim_1$ and $\Sim_2$ be a QPT algorithm.
\begin{enumerate}
    \item The challenger runs $(\MPK,\MSK)\la \Setup(1^\secp)$ and sends $\MPK$ to $\cA$.
    \item $\cA$ is allowed to make arbitrary key queries at most $q$ times.
    For the $\ell$-th key query, the challenger receives $f_\ell\in\mathcal{F}$, 
    computes $\sk_{f_\ell}\la\keygen(\MSK,f_\ell)$,
    and sends $\sk_{f_\ell}$ to $\cA$.
    Let $q^*$ be the number of times that $\cA$ makes key queries.
    Let $\mathcal{V}\seteq \{y_i\seteq f_i(m),f_i,\sk_{f_i}\}_{i\in[q^*]}$.
    \item $\cA$ chooses $m\in\Ms$ and sends $m$ to the challenger.
    \item The experiment works as follows:
    \begin{itemize}
        \item If $b=0$, the challenger computes $(\vk,\CT)\la \Enc(\MPK,m)$, and sends $\CT$ to $\cA$.
        \item If $b=1$, the challenger computes $(\ct,\state_{q^*})\la \Sim_1(\MPK,\mathcal{V},1^{|m|})$, and sends $\ct$ to $\cA$, where $\state_{q^*}$ is a quantum state.
    \end{itemize}
    \item $\cA$ is allowed to make arbitrary key queries at most $q-q^*$ times.
    For the $\ell$-th key query, the challenger works as follows:
    \begin{itemize}
        \item If $b=0$, the challenger receives $f_\ell\in\mathcal{F}$, 
    computes $\sk_{f_\ell}\la\keygen(\MSK,f_\ell)$,
    and sends $\sk_{f_\ell}$ to $\cA$.
    \item If $b=1$, the challenger receives $f_\ell\in\mathcal{F}$, computes $(\sk_{f_{\ell}},\state_{\ell})\la\Sim_2(\MSK,f_{\ell},f_\ell(m),\state_{\ell-1})$, and sends $\sk_{f_\ell}$ to $\cA$.
    \end{itemize}
    \item $\cA$ outputs $b'\in\bit$. The output of the experiment is $b'$.
\end{enumerate}
We say that $\Sigma$ is $q$-bounded adaptive secure if there exists a $QPT$ simulator $\Sim=(\Sim_1,\Sim_2)$ such that for any $QPT$ adversary $\cA$ it holds that
\begin{align}
\adva{\Sigma,\cA}{adapt}(\secp) \seteq \abs{\Pr[\Expt{\Sigma,\cA}{\mathsf{adapt}}(\secp,0)=1]  - \Pr[\Expt{\Sigma,\cA}{\mathsf{adapt}}(\secp,1)=1]} \leq \negl(\secp).
\end{align}
\end{definition}

\begin{definition}[$q$-Bounded Certified Everlasting Non-Adaptive Security for Certified Everlasting FE (Simulation Base)]\label{def:ever_non_ada_security_cert_ever_func_sim}
Let $q$ be a polynomial of $\lambda$.
Let $\Sigma=(\Setup,\keygen,\Enc,\Dec,\Delete,\Vrfy)$ be a $q$-bounded certified everlasting FE scheme.
We consider the following security experiment $\expd{\Sigma,\cA}{cert}{ever}{non}{\mathsf{adapt}}(\secp,b)$ against a QPT adversary $\cA_1$ and an unbounded adversary $\cA_2$. Let $\Sim$ be a QPT algorithm.
\begin{enumerate}
    \item The challenger runs $(\MPK,\MSK)\la \Setup(1^\secp)$ and sends $\MPK$ to $\cA_1$.
    \item $\cA_1$ is allowed to make arbitrary key queries at most $q$ times.
    For the $\ell$-th key query, the challenger receives $f_\ell\in\mathcal{F}$, 
    computes $\sk_{f_\ell}\la\keygen(\MSK,f_\ell)$
    and sends $\sk_{f_\ell}$ to $\cA_1$.
    Let $q^*$ be the number of times that $\cA_1$ makes key queries.
    Let $\mathcal{V}\seteq \{y_i\seteq f_i(m),f_i,\sk_{f_i}\}_{i\in[q^*]}$.
    \item $\cA_1$ chooses $m\in\Ms$ and sends $m$ to the challenger.
    \item The experiment works as follows:
    \begin{itemize}
        \item If $b=0$, the challenger computes $(\vk,\CT)\la \Enc(\MPK,m)$, and sends $\CT$ to $\cA_1$.
        \item If $b=1$, the challenger computes $(\vk,\ct)\la \Sim(\MPK,\mathcal{V},1^{|m|})$, and sends $\ct$ to $\cA_1$.
    \end{itemize}
    \item At some point, $\cA_1$ sends $\cert$ to the challenger and its internal state to $\cA_2$.
    \item The challenger computes $\Vrfy(\vk,\cert)$.
    If the output is $\top$, then the challenger outputs $\top$, and sends $\MSK$ to $\cA_2$.
    Otherwise, the challenger outputs $\bot$, and sends $\bot$ to $\cA_2$.
    \item $\cA_2$ outputs $b'\in\bit$. If the challenger outputs $\top$, the output of the experiment is $b'$. Otherwise, the output of the experiment is $\bot$.
\end{enumerate}
We say that $\Sigma$ is $q$-bounded certified everlasting non-adaptive secure if there exists a $QPT$ simulator $\Sim$ such that for any $QPT$ adversary $\cA_1$ and any unbounded adversary $\cA_2$ it holds that
\begin{align}
\advd{\Sigma,\cA}{cert}{ever}{non}{adapt}(\secp) \seteq \abs{\Pr[\expd{\Sigma,\cA}{cert}{ever}{non}{\mathsf{adapt}}(\secp,0)=1]  - \Pr[\expd{\Sigma,\cA}{cert}{ever}{non}{\mathsf{adapt}}(\secp,1)=1]} \leq \negl(\secp).
\end{align}
\end{definition}

\begin{definition}[$q$-Bounded Certified Everlasting Adaptive Security for Certified Everlasting FE (Simulation Base)]\label{def:ever_ada_security_cert_ever_func_sim}
Let $q$ be a polynomial of $\lambda$.
Let $\Sigma=(\Setup,\keygen,\Enc,\Dec,\Delete,\Vrfy)$ be a $q$-bounded certified everlasting FE scheme.
We consider the following security experiment $\expc{\Sigma,\cA}{cert}{ever}{\mathsf{adapt}}(\secp,b)$ against a QPT adversary $\cA_1$ and an unbounded adversary $\cA_2$.
Let $\Sim_1$, $\Sim_2$, and $\Sim_3$ be a QPT algorithm.
\begin{enumerate}
    \item The challenger runs $(\MPK,\MSK)\la \Setup(1^\secp)$ and sends $\MPK$ to $\cA_1$.
    \item $\cA_1$ is allowed to make arbitrary key queries at most $q$ times.
    For the $\ell$-th key query, the challenger receives $f_\ell\in\mathcal{F}$, 
    computes $\sk_{f_\ell}\la\keygen(\MSK,f_\ell)$
    and sends $\sk_{f_\ell}$ to $\cA_1$.
    Let $q^*$ be the number of times that $\cA_1$ makes key queries.
    Let $\mathcal{V}\seteq \{y_i\seteq f_i(m),f_i,\sk_{f_i}\}_{i\in[q^*]}$.
    \item $\cA_1$ chooses $m\in\Ms$ and sends $m$ to the challenger.
    \item The experiment works as follows:
    \begin{itemize}
        \item If $b=0$, the challenger computes $(\vk,\CT)\la \Enc(\MPK,m)$, and sends $\CT$ to $\cA_1$.
        \item If $b=1$, the challenger computes $(\ct,\state_{q^*})\la \Sim_1(\MPK,\mathcal{V},1^{|m|})$, and sends $\ct$ to $\cA_1$, where $\state_{q^*}$ is a quantum state.
    \end{itemize}
    \item $\cA_1$ is allowed to make arbitrary key queries at most $q-q^*$ times.
    For the $\ell$-th key query, the challenger works as follows.
    \begin{itemize}
    \item If $b=0$, the challenger receives $f_\ell\in\mathcal{F}$, 
    computes $\sk_{f_\ell}\la\keygen(\MSK,f_\ell)$,
    and sends $\sk_{f_\ell}$ to $\cA_1$.
    \item If $b=1$, the challenger receives $f_\ell\in\mathcal{F}$, computes $(\sk_{f_{\ell}},\state_{\ell})\la\Sim_2(\MSK,f_\ell,f_\ell(m),\state_{\ell-1})$, and sends $\sk_{f_\ell}$ to $\cA_1$,
    where $\state_{\ell}$ is a quantum state.
    \end{itemize}
    \item If $b=1$, the challenger runs $\vk\la\Sim_3(\state_{q'})$.
    Here, $q'$ is the number of times that $\cA_1$ makes key queries in total.
    \item At some point, $\cA_1$ sends $\cert$ to the challenger and its internal state to $\cA_2$.
    \item The challenger computes $\Vrfy(\vk,\cert)$. If the output is $\top$, then the challenger outputs $\top$, and sends $\MSK$ to $\cA_2$.
    Otherwise, the challenger outputs $\bot$, and sends $\bot$ to $\cA_2$. 
    \item $\cA_2$ outputs $b'\in\bit$. If the challenger outputs $\top$, the output of the experiment is $b'$. Otherwise, the output of the experiment is $\bot$.
\end{enumerate}
We say that $\Sigma$ is $q$-bounded certified everlasting adaptive secure if there exists a $QPT$ simulator $\Sim=(\Sim_1,\Sim_2,\Sim_3)$ such that for any $QPT$ adversary $\cA_1$ and any unbounded adversary $\cA_2$ it holds that
\begin{align}
\advc{\Sigma,\cA}{cert}{ever}{adapt}(\secp) \seteq \abs{\Pr[\expc{\Sigma,\cA}{cert}{ever}{\mathsf{adapt}}(\secp,0)=1]  - \Pr[\expc{\Sigma,\cA}{cert}{ever}{\mathsf{adapt}}(\secp,1)=1]} \leq \negl(\secp).
\end{align}
\end{definition}

\subsection{Construction of 1-Bounded Certified Everlasting Functional Encryption with Non-Adaptive Security}\label{sec:const_func_non_adapt}
In this section, we construct a 1-bounded certified everlasting FE scheme with non-adaptive security from a certified everlasting garbling scheme~(\cref{def:cert_ever_garbled}) and a certified everlasting PKE scheme~(\cref{def:cert_ever_pke}).

\paragraph{Our 1-bounded certified everlasting FE scheme with non-adaptive security.}
We use a universal circuit $U(\cdot,x)$ in which a plaintext $x$ is hard-wired. The universal circuit takes a function $f$ as input and outputs $f(x)$.
Let $s\seteq |f|$.
We construct a 1-bounded certified everlasting FE scheme with non-adaptive security $\Sigma_{\mathsf{cefe}}=(\Setup,\keygen,\Enc,\Dec,\Delete,\Vrfy)$
from a certified everlasting garbling scheme $\Sigma_{\mathsf{cegc}}=\mathsf{GC}.(\Samp,\Garble,\Eval,\Delete,\Vrfy)$~(\cref{def:cert_ever_garbled})
and a certified everlasting PKE scheme $\Sigma_{\mathsf{cepk}}=\PKE.(\keygen,\Enc,\Dec,\Delete,\Vrfy)$~(\cref{def:cert_ever_pke}).

\begin{description}
\item[$\Setup(1^\secp)$:] $ $
\begin{itemize}
    \item Generate $(\pke.\pk_{i,\alpha},\pke.\sk_{i,\alpha})\la\PKE.\keygen(1^\secp)$ for every $i\in [s]$ and $\alpha\in\bit$.
    \item Output $\MPK\seteq \{\pke.\pk_{i,\alpha}\}_{i\in[s],\alpha\in\{0,1\}}$ and $\MSK\seteq \{\pke.\sk_{i,\alpha}\}_{i\in[s],\alpha\in\bit}$.
\end{itemize}
\item[$\keygen(\MSK,f)$:] $ $
\begin{itemize}
    \item Parse $\MSK=\{\pke.\sk_{i,\alpha}\}_{i\in[s],\alpha\in\bit}$ and $f=(f_1,\cdots,f_s)$.
    \item Output $\sk_f\seteq (f,\{\pke.\sk_{i,f[i]}\}_{i\in[s]})$.
\end{itemize}
\item[$\Enc(\MPK,m)$:]$ $
\begin{itemize}
    \item Parse $\MPK=\{\pke.\pk_{i,\alpha}\}_{i\in[s],\alpha\in\bit}$.
    \item Compute $\{L_{i,\alpha}\}_{i\in[s],\alpha\in\bit}\la\mathsf{GC}.\Samp(1^\secp)$.
    \item Compute $(\widetilde{U},\mathsf{gc}.\vk)\la\mathsf{GC}.\Garble(1^\secp,U(\cdot,m),\{L_{i,\alpha}\}_{i\in[s],\alpha\in\bit})$.
    \item For every $i\in[s]$ and $\alpha\in\bit$, compute $(\pke.\vk_{i,\alpha},\pke.\ct_{i,\alpha})\la\PKE.\Enc(\pke.\pk_{i,\alpha},L_{i,\alpha})$.
    \item Output $\vk\seteq (\mathsf{gc}.\vk,\{\pke.\vk_{i,\alpha}\}_{i\in[s],\alpha\in\bit})$
    and $\CT\seteq (\widetilde{U},\{\pke.\ct_{i,\alpha}\}_{i\in[s],\alpha\in\bit})$.
\end{itemize}
\item[$\Dec(\sk_f,\ct)$:] $ $
\begin{itemize}
    \item Parse $\sk_f=(f,\{\pke.\sk_i\}_{i\in[s]})$ and $\ct=(\widetilde{U},\{\pke.\ct_{i,\alpha}\}_{i\in[s],\alpha\in\bit})$.
    \item For every $i\in [s]$, compute $L_i,\la\PKE.\Dec(\pke.\sk_{i},\pke.\ct_{i,f[i]})$.
    \item Compute $y\la \mathsf{GC}.\Eval(\widetilde{U},\{L_i\}_{i\in[s]})$.
    \item Output $y$.
\end{itemize}
\item[$\Delete(\CT)$:] $ $
\begin{itemize}
    \item Parse $\ct=(\widetilde{U},\{\pke.\ct_{i,\alpha}\}_{i\in[s],\alpha\in\bit})$.
    \item Compute $\mathsf{gc}.\cert\la\mathsf{GC}.\Delete(\widetilde{U})$.
    \item For every $i\in[s]$ and $\alpha\in\bit$, compute $\pke.\cert_{i,\alpha}\la\PKE.\Delete(\pke.\ct_{i,\alpha})$.
    \item Output $\cert\seteq (\mathsf{gc}.\cert,\{\pke.\cert_{i,\alpha}\}_{i\in[s],\alpha\in\bit})$.
\end{itemize}
\item[$\Vrfy(\vk,\cert)$:] $ $
\begin{itemize}
    \item Parse $\vk=(\mathsf{gc}.\vk,\{\pke.\vk_{i,\alpha}\}_{i\in[s],\alpha\in\bit})$ and $\cert=(\mathsf{gc}.\cert,\{\pke.\cert_{i,\alpha}\}_{i\in[s],\alpha\in\bit})$.
    \item Output $\top$ if $\top\la\mathsf{GC}.\Vrfy(\mathsf{gc}.\vk,\mathsf{gc}.\cert)$ and $\top\la\PKE.\Vrfy(\pke.\vk_{i,\alpha},\pke.\cert_{i,\alpha})$ for every $i\in[s]$ and $\alpha\in\bit$.
    Otherwise, output $\bot$.
\end{itemize}
\end{description}

\paragraph{Correctness:}
Correctness easily follows from that of $\Sigma_{\mathsf{cegc}}$ and $\Sigma_{\mathsf{cepk}}$.

\paragraph{Security:}
The following two theorems hold.
\begin{theorem}\label{thm:func_non_ad_seucirty}
If $\Sigma_{\mathsf{cegc}}$ satisfies the selective security~(\cref{def:sel_sec_ever_garb}) and $\Sigma_{\mathsf{cepk}}$ satisfies the IND-CPA security~(\cref{def:IND-CPA_security_cert_ever_pke}),
$\Sigma_{\mathsf{cefe}}$ satisfies the 1-bounded non-adaptive security~(\cref{def:non_ada_security_cert_ever_func_sim}).
\end{theorem}
Its proof is similar to that of \cref{thm:func_ever_non_ad_security}, and therefore we omit it.
\begin{theorem}\label{thm:func_ever_non_ad_security}
If $\Sigma_{\mathsf{cegc}}$ satisfies the certified everlasting selective security~(\cref{def:ever_sec_ever_garb}) and $\Sigma_{\mathsf{cepk}}$ satisfies the certified everlasting IND-CPA security~(\cref{def:cert_ever_security_cert_ever_pke}),
$\Sigma_{\mathsf{cefe}}$ satisfies the 1-bounded certified everlasting non-adaptive security~(\cref{def:ever_non_ada_security_cert_ever_func_sim}).
\end{theorem}
Its proof is given in \cref{sec:proof_non_ad_fe}.

\subsection{Construction of 1-Bounded Certified Everlasting Functional Encryption with Adaptive Security}\label{sec:const_fe_adapt}
In this section, we change the non-adaptive scheme constructed in the previous subsection to the adaptive one by using
a certified everlasting RNC scheme~(\cref{def:cert_ever_rnce_classic}).

\paragraph{Our 1-bounded certified everlasting FE scheme with adaptive security.}
We construct a 1-bounded certified everlasting FE scheme with adaptive security $\Sigma_{\mathsf{cefe}}=(\Setup,\keygen,\Enc,\Dec,\Delete,\Vrfy)$ 
from a 1-bounded certified everlasting FE scheme with non-adaptive security $\Sigma_{\mathsf{nad}}=\mathsf{NAD}.(\Setup,\keygen,\Enc,\Dec,\Delete,\Vrfy)$,
where the ciphertext space is $\Cs\seteq\cQ^{\otimes n}$,
and a certified everlasting RNCE scheme 
\begin{align}
\Sigma_{\mathsf{cence}}=\NCE.(\Setup,\keygen,\Enc,\Dec,\Fake,\Reveal,\Delete,\Vrfy)
\end{align}
(\cref{def:cert_ever_rnce_classic}).
Let $\NAD.\Modify$ be a QPT algorithm such that
\begin{align}
\Pr\left[
\NAD.\Vrfy(\nad.\vk,\nad.\cert^*)\neq\top
\ \middle |
\begin{array}{ll}
(\nad.\MPK,\nad.\MSK)\lrun \NAD.\Setup(1^\secp)\\
(\nad.\vk,\nad.\ct) \lrun \NAD.\Enc(\nad.\MPK,m)\\
a,c\la\bit^{n}\\
\nad.\cert \lrun \NAD.\Delete(Z^cX^a\nad.\ct X^aZ^c)\\
\nad.\cert^*\lrun \NAD.\Modify(a,c,\nad.\cert) 
\end{array}
\right] 
\leq
\negl(\secp).
\end{align}
for any $m$.

Our construction is as follows.
\begin{description}
\item[$\Setup(1^\secp)$:]$ $
\begin{itemize}
    \item Run $(\mathsf{nad}.\MPK,\mathsf{nad}.\MSK)\la\mathsf{NAD}.\Setup(1^\secp)$.
    \item Run $(\nce.\pk,\nce.\MSK)\la\NCE.\Setup(1^{\secp})$.
    \item Output $\MPK\seteq(\mathsf{nad}.\MPK,\nce.\pk)$ and $\MSK\seteq (\nad.\MSK,\nce.\MSK)$.
\end{itemize}
\item[$\keygen(\MSK,f)$:]$ $
\begin{itemize}
    \item Parse $\MSK=(\nad.\MSK,\nce.\MSK)$.
    \item Compute $\mathsf{nad}.\sk_f\la\mathsf{NAD}.\keygen(\mathsf{nad}.\MSK,f)$.
    \item Compute $\nce.\sk\la\NCE.\keygen(\nce.\MSK)$.
    \item Output $\sk_f\seteq (\nad.\sk_f,\nce.\sk)$.
\end{itemize}
\item[$\Enc(\MPK,m)$:]$ $
\begin{itemize}
    \item Parse $\MPK=(\nad.\MPK,\nce.\pk)$.
    \item Compute $(\mathsf{nad}.\vk,\mathsf{nad}.\ct)\la\mathsf{NAD}.\Enc(\mathsf{nad}.\MPK,m)$.
    \item Generate $a,c\la\bit^n$. Let $\Psi\seteq Z^cX^a\mathsf{nad}.\ct X^aZ^c$.
    \item Compute $(\nce.\vk,\nce.\ct)\la\NCE.\Enc(\nce.\pk,(a,c))$.
    \item Output $\vk\seteq(\mathsf{nad}.\vk,\nce.\vk,a,c)$ and $\ct\seteq (\Psi,\nce.\ct)$.
\end{itemize}
\item[$\Dec(\sk_f,\ct)$:]$ $
\begin{itemize}
    \item Parse $\sk_f=(\nad.\sk_f,\nce.\sk)$ and $\ct=(\Psi,\nce.\ct)$.
    \item Compute $(a',c')\la\NCE.\Dec(\nce.\sk,\nce.\ct)$.
    \item Compute $\mathsf{nad}.\ct'\seteq X^{a'}Z^{c'}\Psi Z^{c'}X^{a'}$.
    \item Compute $y\la\mathsf{NAD}.\Dec(\mathsf{nad}.\sk_f,\mathsf{nad}.\ct')$.
    \item Output $y$.
\end{itemize}
\item[$\Delete(\ct)$:]$ $
\begin{itemize}
    \item Parse $\ct=(\Psi,\nce.\ct)$.
    \item Compute $\nad.\cert\la \NAD.\Delete(\Psi)$.
    \item Compute $\nce.\cert\la\NCE.\Delete(\nce.\ct)$.
    \item Output $\cert\seteq(\nad.\cert,\nce.\cert)$.
\end{itemize}
\item[$\Vrfy(\vk,\cert)$:]$ $
\begin{itemize}
    \item Parse $\vk=(\mathsf{nad}.\vk,\nce.\vk,a,c)$ and $\cert=(\nad.\cert,\nce.\cert)$.
    \item Compute $\mathsf{nad}.\cert^*\la\NAD.\Modify (a,c,\nad.\cert)$.
    \item Output $\top$ if $\top\la\NCE.\Vrfy(\nce.\vk,\nce.\cert)$ and $\top\la\mathsf{NAD}.\Vrfy(\mathsf{nad}.\vk,\mathsf{nad}.\cert^*)$.
    Otherwise, output $\bot$.
\end{itemize}
\end{description}
\paragraph{Correctness:}
Correctness easily follows from that of $\Sigma_{\mathsf{nad}}$ and $\Sigma_{\mathsf{cence}}$.

\paragraph{Security:}
The following two theorems hold.
\begin{theorem}\label{thm:comp_security_single_ad}
If $\Sigma_{\mathsf{nad}}$ satisfies the 1-bounded non-adaptive security~(\cref{def:non_ada_security_cert_ever_func_sim}) and $\Sigma_{\mathsf{cence}}$ satisfies the RNC security~(\cref{def:rec_nc_security_classic}), 
$\Sigma_{\mathsf{cefe}}$ satisfies the 1-bounded adaptive security~(\cref{def:ada_security_cert_ever_func_sim}).
\end{theorem}
Its proof is similar to that of \cref{thm:ever_security_single_ad}, and therefore we omit it.

\begin{theorem}\label{thm:ever_security_single_ad}
If $\Sigma_{\mathsf{nad}}$ satisfies the 1-bounded certified everlasting non-adaptive security(~\cref{def:ever_non_ada_security_cert_ever_func_sim}) and $\Sigma_{\mathsf{cence}}$ satisfies the certified everlasting RNC security~(\cref{def:cert_ever_rec_nc_security_classic}), 
$\Sigma_{\mathsf{cefe}}$ satisfies the 1-bounded certified everlasting adaptive security~(\cref{def:ever_ada_security_cert_ever_func_sim}).
\end{theorem}

Its proof is given in \cref{sec:proof_ada_fe}.

\subsection{Construction of $q$-Bounded Certified Everlasting Functional Encryption for $\NCone$ circuits}\label{sec:const_multi_fe}
In this section, we construct a $q$-bounded certified everlasting FE for all $\NCone$ circuits from 1-bounded certified everlasting FE constructed in the previous subsection
and Shamir's secret sharing~(\cite{Shamir79}).
Our construction is similar to that of ordinary FE for all $\NCone$ circuits in \cite{C:GorVaiWee12}
except that we use a 1-bounded certified everlasting FE instead of an ordinary 1-bounded FE.


\paragraph{Our $q$-bounded certified everlasting FE scheme for $\NCone$ circuits.}
We consider the polynomial representation of circuits $C$ in $\NCone$.
The input message space is $\Ms\seteq\mathbb{F}^{\ell}$, and for each $\NCone$ circuit $C$, $C(\cdot)$ is an $\ell$-variate polynomial over $\mathbb{F}$ of total degree at most $D$.
Let $q=q(\lambda)$ be a polynomial of $\lambda$. Our scheme is associated with additional parameters $S=S(\lambda)$, $N=N(\lambda)$, $t=t(\lambda)$ and $v=v(\lambda)$
that satisfy
\begin{align}
    t(\lambda)=\Theta(q^2\lambda), N(\lambda)=\Theta(D^2q^2t), v(\lambda)=\Theta(\lambda), S(\lambda)=\Theta(vq^2).
\end{align}

Let us define a family $\mathcal{G}\seteq\{G_{C,\Delta}\}_{C\in\NCone,\Delta\subseteq[S]}$, where
\begin{align}
G_{C,\Delta}(x,Z_1,Z_2,\cdots,Z_S)\seteq C(x)+\sum_{i\in\Delta}Z_i    
\end{align}
is a function and $Z_1,\cdots,Z_S\in\mathbb{F}$.

We construct a $q$-bounded certified everlasting FE scheme for all $\NCone$ circuits  $\Sigma_{\mathsf{cefe}}=(\Setup,\keygen,\Enc,\Dec,\Delete,\allowbreak \Vrfy)$
from a 1-bounded certified everlasting FE scheme $\Sigma_{\mathsf{one}}=\mathsf{ONE}.(\Setup,\keygen,\Enc,\Dec,\Delete,\Vrfy)$.

\begin{description}
\item[$\Setup(1^\lambda)$:]$ $
\begin{itemize}
    \item For $i\in[N]$, generate $(\mathsf{one}.\MPK_i,\mathsf{one}.\MSK_i)\la\mathsf{ONE}.\Setup(1^\secp)$.
    \item Output $\MPK\seteq\{\mathsf{one}.\MPK_i\}_{i\in[N]}$ and $\MSK\seteq \{\mathsf{one}.\MSK_i\}_{i\in[N]}$.
\end{itemize}
\item[$\keygen(\MSK,C)$:]$ $
\begin{itemize}
    \item Parse $\MSK= \{\mathsf{one}.\MSK_i\}_{i\in[N]}$.
    \item Chooses a uniformly random set $\Gamma\subseteq[N]$ of size $tD+1$.
    \item Chooses a uniformly random set $\Delta\subseteq[S]$ of size $v$.
    \item For $i\in\Gamma$, compute $\mathsf{one}.\sk_{C,\Delta,i}\la\mathsf{ONE}.\keygen(\one.\MSK_i,G_{C,\Delta})$.
    \item Output $\sk_{C}\seteq (\Gamma,\Delta,\{\one.\sk_{C,\Delta,i}\}_{i\in\Gamma})$.
\end{itemize}
\item[$\Enc(\MPK,x)$:]$ $
\begin{itemize}
    \item Parse $\MPK=\{\mathsf{one}.\MPK_i\}_{i\in[N]}$.
    \item For $i\in[\ell]$, pick a random degree $t$ polynomial $\mu_i(\cdot)$ whose constant term is $x[i]$.
    \item For $i\in[S]$, pick a random degree $Dt$ polynomial $\xi_i(\cdot)$ whose constant term is $0$.
    \item For $i\in[N]$, compute $(\one.\vk_i,\one.\ct_i)\la\ONE.\Enc(\one.\MPK_i,(\mu_1(i),\cdots,\mu_{\ell}(i),\xi_1(i),\cdots,\xi_S(i) ))$.
    \item Output $\vk=\{\one.\vk_i\}_{i\in[N]}$ and $\ct\seteq\{\one.\ct_i\}_{i\in[N]}$.
    \end{itemize}
\item[$\Dec(\sk_C,\ct)$:]$ $
\begin{itemize}
    \item Parse $\sk_{C}= (\Gamma,\Delta,\{\one.\sk_{C,\Delta,i}\}_{i\in\Gamma})$ and $\ct=\{\one.\ct_i\}_{i\in[N]}$.
    \item For $i\in \Gamma$, compute $\eta(i)\la\ONE.\Dec(\one.\sk_{C,\Delta,i},\one.\ct_i)$.
    \item Output $\eta(0)$.
\end{itemize}
\item[$\Delete(\ct)$:]$ $
\begin{itemize}
    \item Parse $\ct=\{\one.\ct_i\}_{i\in[N]} $.
    \item For $i\in[N]$, compute $\one.\cert_i\la\ONE.\Delete(\one.\ct_i)$.
    \item Output $\cert\seteq \{\one.\cert_i\}_{i\in[N]}$.
\end{itemize}
\item[$\Vrfy(\vk,\cert)$:]$ $
\begin{itemize}
    \item Parse $\vk=\{\one.\vk_i\}_{i\in[N]}$ and $\cert= \{\one.\cert_i\}_{i\in[N]}$.
    \item For $i\in[N]$, compute $\top/\bot\la\ONE.\Vrfy(\one.\vk_i,\one.\cert_i)$.
    If all results are $\top$, output $\top$. Otherwise, output $\bot$.
\end{itemize}
\end{description}

\paragraph{Correctness:}
Verification correctness easily follows from verification correctness of $\Sigma_{\one}$.
Let us show evaluation correctness.
By decryption correctness of $\Sigma_{\one}$, for all $i\in\Gamma$ we have 
\begin{align}
\eta(i)&=G_{C,\Delta} (\mu_1(i),\cdots,\mu_\ell(i),\xi_1(i),\cdots,\xi_S(i))\\
&=C(\mu_1(i),\cdots,\mu_\ell(i))+\Sigma_{a\in\Delta}\xi_a(i).
\end{align}
Since $|\Gamma|\geq Dt+1$, this means that $\eta$ is equal to the degree $Dt$ polynomial 
\begin{align}
    \eta(\cdot)=C(\mu_1(\cdot),\cdots,\mu_\ell(\cdot) )+\Sigma_{a\in\Delta}\xi_a(\cdot)
\end{align}
Hence $\eta(0)=C(x_1,\cdots,x_\ell)=C(x)$, which means that our construction satisfies evaluation correctness.

\paragraph{Security:}
The following two theorems hold.
\begin{theorem}\label{thm:compt_bounded_fe}
If $\Sigma_{\one}$ satisfies the 1-bounded adaptive security, $\Sigma_{\mathsf{cefe}}$ satisfies the $q$-bounded adaptive security.
\end{theorem}
Its proof is similar to that of \cref{thm:ever_bounded_fe}, and therefore we omit it.

\begin{theorem}\label{thm:ever_bounded_fe}
If $\Sigma_{\one}$ satisfies the 1-bounded certified everlasting adaptive security,
$\Sigma_{\mathsf{cefe}}$ the $q$-bounded certified everlasting adaptive security.
\end{theorem}

Its proof is given in \cref{Sec:Proof_bounded_fe}.


\ifnum\submission=1
\else
\section*{Acknowledgement}
TM is supported by JST Moonshot JPMJMS2061-5-1-1, JST FOREST, MEXT QLEAP, the Grant-in-Aid for Scientific Research
(B) No.JP19H04066, the Grant-in Aid for Transformative Research Areas (A)
21H05183, and the Grant-in-Aid for Scientific Research (A) No.22H00522
\fi

\ifnum\llncs=1
\bibliographystyle{myalpha} 
\bibliography{abbrev3,crypto,reference}
\else
\ifnum\arxiv=1
\newcommand{\etalchar}[1]{$^{#1}$}

\else
\bibliographystyle{alpha} 
\bibliography{abbrev3,crypto,reference}
\fi
\fi

\ifnum\cameraready=1
\else
\appendix

	\ifnum\submission=1
	\newpage
	 	\setcounter{page}{1}
 	{
	\noindent
 	\begin{center}
	{\Large SUPPLEMENTAL MATERIALS}
	\end{center}
 	}
	\setcounter{tocdepth}{2}

	\else
	
	\fi
\appendix

\section{Proof of \cref{thm:ever_security_rnce_classic}}\label{sec:proof_rnce}
\begin{proof}[Proof of \cref{thm:ever_security_rnce_classic}]
To prove the theorem, let us introduce the sequence of hybrids.
\begin{description}
    \item[$\sfhyb{0}{}$:] This is identical to $\expd{\Sigma_{\mathsf{cence}},\cA}{cert}{ever}{rec}{nc}(\secp,0)$.
    For clarity, we describe the experiment against any adversary $\cA=(\cA_1,\cA_2)$, where $\cA_1$ is any QPT adversary and $\cA_2$ is any unbounded adversary. 
\begin{enumerate}
    \item The challenger generates $(\pke.\pk_{i,\alpha},\pke.\sk_{i,\alpha})\la\PKE.\keygen(1^\lambda)$ for all $i\in[n]$ and $\alpha\in\bit$.
    \item The challenger sends $\{\pke.\pk_{i,\alpha}\}_{i\in[n],\alpha\in\bit}$ to $\cA_1$.
    \item $\cA_1$ sends $m\in\Ms$ to the challenger.
    \item\label{step:enc_rnce_classical} The challenger generates $x\la\bit^n$, computes $(\pke.\vk_{i,\alpha},\pke.\ct_{i,\alpha})\la\PKE.\Enc(\pke.\pk_{i,\alpha},m[i])$ for all $i\in[n]$ and $\alpha\in\bit$, and sends $(\{\pke.\ct_{i,\alpha}\}_{i\in[n],\alpha\in\bit},(x,\{\pke.\sk_{i,x[i]}\}_{i\in[n]}))$ to $\cA_1$.
    \item $\cA_1$ sends $\{\pke.\cert_{i,\alpha}\}_{i\in[n],\alpha\in\bit}$ to the challenger and its internal state to $\cA_2$.
    \item The challenger computes $\PKE.\Vrfy(\pke.\vk_{i,\alpha},\pke.\cert_{i,\alpha})$ for all $i\in[n]$ and $\alpha\in\bit$.
    If all results are $\top$, the challenger outputs $\top$ and sends $\{\pke.\sk_{i,\alpha}\}_{i\in[n],\alpha\in\bit}$ to $\cA_2$. Otherwise, the challenger outputs $\bot$ and sends $\bot$ to $\cA_2$.
    \item $\cA_2$ outputs $b'\in\bit$. 
    \item If the challenger outputs $\top$, then the output of the experiment is $b'$. Otherwise, the output of the experiment is $\bot$.
\end{enumerate}
\item[$\sfhyb{1}{}$:] 
This is identical to $\sfhyb{0}{}$ except that
the challenger generates $(\pke.\vk_{i,x[i]\oplus 1},\pke.\ct_{i,x[i]\oplus 1})\la\PKE.\Enc(\pke.\pk_{i,x[i]\oplus 1},\allowbreak m[i]\oplus 1)$ for all $i\in[n]$ instead of computing $(\pke.\vk_{i,x[i]\oplus 1},\pke.\ct_{i,x[i]\oplus 1})\la\PKE.\Enc(\pke.\pk_{i,x[i]\oplus 1},m[i])$ for all $i\in[n]$.
\item[$\sfhyb{2}{}$:]
This is identical to $\sfhyb{1}{}$ except for the following three points.
\begin{enumerate}
\item The challenger generates $x^*\la\bit^n$ instead of generating $x\la\bit^n$.
\item For all $i\in[n]$, the challenger generates $(\pke.\vk_{i,x^*[i]},\allowbreak \pke.\ct_{i,x^*[i]})\la \PKE.\Enc(\pke.\pk_{i,x^*[i]},0)$
and $(\pke.\vk_{i,x^*[i]\oplus 1},\pke.\ct_{i,x^*[i]\oplus 1})\la \PKE.\Enc(\pke.\pk_{i,x^*[i]\oplus 1},1)$
instead of computing $(\pke.\vk_{i,x[i]},\pke.\ct_{i,x[i]})\allowbreak \la\PKE.\Enc(\pke.\pk_{i,x[i]},m[i])$
and $(\pke.\vk_{i,x[i]\oplus 1},\pke.\ct_{i,x[i]\oplus 1})\la\PKE.\Enc(\pke.\pk_{i,x[i]\oplus 1},m[i]\oplus 1)$.
\item The challenger sends $(\{\pke.\ct_{i,\alpha}\}_{i\in[n],\alpha\in\bit},(x^*\oplus m,\{\pke.\sk_{i,x^*[i]\oplus m[i]}\}_{i\in[n]}))$ to $\cA_1$
instead of sending $(\{\pke.\ct_{i,\alpha}\}_{i\in[n],\alpha\in\bit},(x,\{\pke.\sk_{i,x[i]}\}_{i\in[n]}))$ to $\cA_1$.
\end{enumerate}
\end{description}
It is clear that $\sfhyb{0}{}$ is identical to $\expd{\Sigma,\cA}{cert}{ever}{rec}{nc}(\secp,0)$ and $\sfhyb{2}{}$ is identical to $\expd{\Sigma,\cA}{cert}{ever}{rec}{nc}(\secp,1)$.
Hence, \cref{thm:ever_security_rnce_classic} easily follows from the following \cref{prop:hyb_0_hyb_1_rnce_classic,prop:hyb_1_hyb_2_rnce_classic}
(whose proof is given later.).
\end{proof}

\begin{proposition}\label{prop:hyb_0_hyb_1_rnce_classic}
If $\Sigma_{\mathsf{cepk}}$ is certified everlasting IND-CPA secure,
it holds that
$\abs{\Pr[\sfhyb{0}{}=1]-\Pr[\sfhyb{1}{}=1]}\leq\negl(\lambda)$.
\end{proposition}
\begin{proposition}\label{prop:hyb_1_hyb_2_rnce_classic}
$\abs{\Pr[\sfhyb{1}{}=1]-\Pr[\sfhyb{2}{}=1]}\leq\negl(\lambda)$.    
\end{proposition}

\begin{proof}[Proof of \cref{prop:hyb_0_hyb_1_rnce_classic}]

For the proof, we use \cref{lem:cut_and_choose_pke}.
We assume that $\abs{\Pr[\sfhyb{0}{}=1]-\Pr[\sfhyb{1}{}=1]}$ is non-negligible,
and construct an adversary $\cB$ that breaks the security experiment $\expc{\Sigma_{\mathsf{cepk}},\cB}{multi}{cert}{ever}(\secp,b)$ defined in~\cref{lem:cut_and_choose_pke}.
This contradicts the certified everlasting IND-CPA security of $\Sigma_{\mathsf{cepk}}$ from~\cref{lem:cut_and_choose_pke}.
Let us describe how $\cB$ works below.
\begin{enumerate}
    \item $\cB$ receives $\{\pke.\pk_{i,\alpha}\}_{i\in[n],\alpha\in\bit}$ from the challenger of
    $\expc{\Sigma_{\mathsf{cepk}},\cB}{multi}{cert}{ever}(\secp,b)$.
    \item $\cB$ sends $\{\pke.\pk_{i,\alpha}\}_{i\in[n],\alpha\in\bit}$ to $\cA_1$.
    \item $\cA_1$ chooses $m\in\Ms$ and sends $m$ to $\cB$.
    \item $\cB$ generates $x\la\bit^n$ and sends $(x,m[1],\cdots, m[n], m[1]\oplus 1,\cdots,m[n]\oplus 1)$ to the challenger of $\expc{\Sigma_{\mathsf{cepk}},\cB}{multi}{cert}{ever}(\secp,b)$.
    \item $\cB$ receives $(\{\pke.\sk_{i,x[i]}\}_{i\in[n]},\{\pke.\ct_{i,x[i]\oplus 1}\}_{i\in[n]})$ from the challenger of $\expc{\Sigma_{\mathsf{cepk}},\cB}{multi}{cert}{ever}(\secp,b)$.
    \item $\cB$ computes $(\{\pke.\vk_{i,x[i]}\}_{i\in[n]},\{\pke.\ct_{i,x[i]}\}_{i\in[n]})\la\PKE.\Enc(\pke.\pk_{i,x[i]},m[i])$ for $i\in[n]$. 
    \item $\cB$ sends $(\{\pke.\ct_{i,\alpha}\}_{i\in[n],\alpha\in\bit},(x,\{\pke.\sk_{i,x[i]}\}_{i\in[n]}))$ to $\cA_1$.
    \item $\cA_1$ sends $\{\pke.\cert_{i,\alpha}\}_{i\in[n],\alpha\in\bit}$ to $\cB$ and the internal state to $\cA_2$.
    \item $\cB$ sends $\{\pke.\cert_{i,x[i]\oplus 1}\}_{i\in[n]}$ to the challenger, and receives $\{\pke.\sk_{i,x[i]\oplus 1}\}_{i\in[n]}$ or $\bot$.
    If $\cB$ receives $\bot$, it outputs $\bot$ and aborts.
    \item $\cB$ sends $\{\pke.\sk_{i,\alpha}\}_{i\in[n],\alpha\in\bit}$ to $\cA_2$.
    \item $\cA_2$ outputs $b'$.
    \item $\cB$ computes $\PKE.\Vrfy(\pke.\vk_{i,x[i]},\pke.\cert_{i,x[i]})$ for all $i\in[n]$.
    If all results are $\top$, $\cB$ outputs $b'$.
    Otherwise, $\cB$ outputs $\bot$.
\end{enumerate}
It is clear that $\Pr[1\la\cB|b=0]=\Pr[\sfhyb{0}{}=1]$ and $\Pr[1\la\cB|b=1]=\Pr[\sfhyb{1}{}=1]$.
By assumption,\\ $\abs{\Pr[\sfhyb{0}{}=1]-\Pr[\sfhyb{1}{}=1]}$ is non-negligible.
Therefore, $\abs{\Pr[1\la\cB|b=0]-\Pr[1\la\cB|b=1]}$ is also non-negligible,
which contradicts the certified everlasting IND-CPA security of $\Sigma_{\mathsf{cepk}}$ from~\cref{lem:cut_and_choose_pke}.
\end{proof}

\begin{proof}[Proof of \cref{prop:hyb_1_hyb_2_rnce_classic}]
It is obvious that the joint probability distribution that $\cA_1$ receives $(\{\pke,\ct_{i,\alpha}\}_{i\in[n],\alpha\in\bit},\allowbreak(x,\{\pke.\sk_{i,x[i]}\}_{i\in[n]}))$ in $\sfhyb{1}{}$ is identical to the joint probability distribution that $\cA_1$ receives $(\{\pke,\ct_{i,\alpha}\}_{i\in[n],\alpha\in\bit},\allowbreak (x^*\oplus m,\{\pke.\sk_{i,x^*[i]\oplus m[i]}\}_{i\in[n]}))$ in $\sfhyb{2}{}$.
Hence, \cref{prop:hyb_1_hyb_2_rnce_classic} follows.
\end{proof}

We use the following lemma for the proof of \cref{thm:ever_security_rnce_classic} and \cref{thm:garble_ever_security}.
The proof is shown with the standard hybrid argument.
It is also easy to see that a similar lemma holds for IND-CPA security.
\begin{lemma}\label{lem:cut_and_choose_pke}
Let $s$ be some polynomial of the security parameter $\lambda$.
Let $\Sigma\seteq(\keygen,\Enc,\Dec,\Delete,\Vrfy)$ be a certified everlasting PKE scheme.
Let us consider the following security experiment 
$\expc{\Sigma,\cA}{multi}{cert}{ever}(\secp,b)$ against $\cA$ consisting of any QPT adversary $\cA_1$ and any unbounded adversary $\cA_2$.
\begin{enumerate}
    \item The challenge generates $(\pk_{i,\alpha},\sk_{i,\alpha})\la\keygen(1^\secp)$ for all $i\in[s]$ and $\alpha\in\bit$,
    and sends $\{\pk_{i,\alpha}\}_{i\in[s],\alpha\in\bit}$ to $\cA_1$.
    \item $\cA_1$ chooses $f\in\bit^s$ and $(m_0[1],m_0[2],\cdots, m_0[s],m_1[1],m_1[2], \cdots, m_1[s]) \in\Ms^{2s}$,
    and sends $(f,m_0[1],m_0[2],\\ \cdots ,m_0[s],m_1[1],m_1[2],\cdots, m_1[s] )$ to the challenger.
    \item The challenger computes 
    $
    (\vk_{i,f[i]\oplus 1},\ct_{i,f[i]\oplus 1})\la \Enc(\pk_{i,f[i]\oplus 1},m_b[i])$ for all $i\in[s]$,
    and sends $(\{\sk_{i,f[i]}\}_{i\in [s]},\\ \{\ct_{i,f[i]\oplus 1}\}_{i\in[s]})$ to $\cA_1$.
    \item At some point, $\cA_1$ sends $\{\cert_{i,f[i]\oplus 1}\}_{i\in[s]}$ to the challenger,
    and sends its internal state to $\cA_2$.
    \item The challenger computes $\Vrfy(\vk_{i,f[i]\oplus 1},\cert_{i,f[i]\oplus 1})$ for every $i\in[s]$.
    If all results are $\top$, the challenger outputs $\top$, and sends $\{\sk_{i,f[i]\oplus 1}\}_{i\in [s]}$ to $\cA_2$.
    Otherwise, the challenger outputs $\bot$, and sends $\bot$ to $\cA_2$.
    \item $\cA_2$ outputs $b'$.
    \item If the challenger outputs $\top$, then the output of the experiment is $b'$. Otherwise, the output of the experiment is $\bot$.
\end{enumerate}
If the $\Sigma$ satisfies the certified everlasting IND-CPA security,
\begin{align}
\advc{\Sigma,\cA}{multi}{cert}{ever}(\secp)
\seteq \abs{\Pr[ \expb{\Sigma,\cA}{multi}{cert}{ever}(\secp, 0)=1] - \Pr[ \expb{\Sigma,\cA}{multi}{cert}{ever}(\secp, 1)=1] }\leq \negl(\secp)
\end{align}
for any QPT adversary $\cA_1$ and any unbounded adversary $\cA_2$.
\end{lemma}

\begin{proof}[Proof of \cref{lem:cut_and_choose_pke}]
Let us consider the following hybrids for $j\in\{0,1,\cdots ,s\}$.
\begin{description}
\item[$\sfhyb{j}{}$:] $ $
\begin{enumerate}
    \item The challenger generates $(\pk_{i,\alpha},\sk_{i,\alpha})\la\keygen(1^\secp)$ for every $i\in[s]$ and $\alpha\in\bit$, and sends $\{\pk_{i,\alpha}\}_{i\in[s],\alpha\in\bit}$ to $\cA_1$.
    \item $\cA_1$ chooses $f\in\bit^s$ and $(m_0[1],m_0[2],\cdots, m_0[s],m_1[1],m_1[2],\cdots, m_1[s]) \in\Ms^{2s}$,
    and sends $(f,m_0[1],m_0[2],\cdots, m_0[s],m_1[1],m_1[2],\cdots, m_1[s] )$ to the challenger.
    \item The challenger computes 
    \begin{align}
    (\vk_{i,f[i]\oplus 1},\ct_{i,f[i]\oplus 1})\la \Enc(\pk_{i,f[i]\oplus 1},m_1[i])
    \end{align}
    for $i\in[j]$ and 
    \begin{align}
    (\vk_{i,f[i]\oplus 1},\ct_{i,f[i]\oplus 1})\la \Enc(\pk_{i,f[i]\oplus 1},m_0[i])
    \end{align}
    for $i\in\{j+1,j+2,\cdots,s\}$,
    and sends $(\{\sk_{i,f[i]}\}_{i\in [s]}, \{\ct_{i,f[i]\oplus 1}\}_{i\in[s]})$ to $\cA_1$.
    \item At some point, $\cA_1$ sends $\{\cert_{i,f[i]\oplus 1}\}_{i\in[s]}$ to the challenger,
    and sends its internal state to $\cA_2$.
    \item The challenger computes $\Vrfy(\vk_{i,f[i]\oplus 1},\cert_{i,f[i]\oplus 1})$ for every $i\in[s]$.
    If all results are $\top$, the challenger outputs $\top$, and sends $\{\sk_{i,f[i]\oplus 1}\}_{i\in [s]}$ to $\cA_2$.
    Otherwise, the challenger outputs $\bot$, and sends $\bot$ to $\cA_2$.
    \item $\cA_2$ outputs $b'$.
    \item If the challenger outputs $\top$, then the output of the experiment is $b'$. Otherwise, the output of the experiment is $\bot$.
\end{enumerate}
\end{description}
It is clear that $\Pr[\sfhyb{0}{}=1]=\Pr[\expc{\Sigma,\cA}{multi}{cert}{ever}(\secp,0)=1]$ and $\Pr[\sfhyb{s}{}=1]=\Pr[\expc{\Sigma,\cA}{multi}{cert}{ever}(\secp,1)=1]$.
Furthermore, we can show 
\begin{align}
    \abs{\Pr[\sfhyb{j}{}=1]-\Pr[\sfhyb{j+1}{}=1]}\leq\negl(\lambda)
\end{align}
for each $j\in\{0,1,\cdots,s-1\}$.
(Its proof is given below.) From these facts, we obtain \cref{lem:cut_and_choose_pke}.

Let us show the remaining one. To show it, let us assume that $\abs{\Pr[\sfhyb{j}{}=1]-\Pr[\sfhyb{j+1}{}=1]}$ is non-negligible.
Then, we can construct an adversary $\cB$ that can break the certified everlasting IND-CPA security of $\Sigma$ as follows. 

\begin{enumerate}
    \item $\cB$ receives $\pk$ from the challenger of $\expd{\Sigma,\cB}{cert}{ever}{ind}{cpa}(\secp,b)$.
    \item $\cB$ generates $\beta\la\bit$ and sets $\pk_{j+1,\beta}\seteq \pk$. 
    \item $\cB$ generates $(\pk_{i,\alpha},\sk_{i,\alpha})\la\keygen(1^\secp)$ for $i\in\{1,\cdots,j,j+2,\cdots, s\}$ and $\alpha\in\bit$, and $(\pk_{j+1,\beta\oplus 1},\sk_{j+1,\beta\oplus 1})\la\keygen(1^\secp)$.
    \item $\cB$ sends $\{\pk_{i,\alpha}\}_{i\in[s],\alpha\in\bit}$ to $\cA_1$.
    \item $\cA_1$ chooses $f\in\bit^s$ and $(m_0[1],m_0[2],\cdots, m_0[s],m_1[1],m_1[2],\cdots, m_1[s]) \in\Ms^{2s}$,
    and sends $(f,m_0[1],m_0[2],\\ \cdots, m_0[s],m_1[1],m_1[2],\cdots, m_1[s] )$ to the challenger.
    \item If $f[j+1]=\beta$, $\cB$ aborts the experiment, and outputs $\bot$.
    \item $\cB$ computes
    \begin{align}
        (\vk_{i,f[i]\oplus 1},\ct_{i,f[i]\oplus 1})\la \Enc(\pk_{i,f[i]\oplus 1},m_1[i])
    \end{align}
    for $i\in[j]$ and
    \begin{align}
        (\vk_{i,f[i]\oplus 1},\ct_{i,f[i]\oplus 1})\la \Enc(\pk_{i,f[i]\oplus 1},m_0[i])
    \end{align}
    for $i\in\{j+2,\cdots, s\}$.
    \item $\cB$ sends $(m_0[j+1],m_1[j+1])$ to the challenger of $\expd{\Sigma,\cB}{cert}{ever}{ind}{cpa}(\secp,b)$.
    The challenger computes $(\vk_{j+1,f[j+1]\oplus 1},\ct_{j+1,f[j+1]\oplus 1})\la \Enc(\pk_{j+1,f[j+1]\oplus 1},m_b[j+1])$ and sends $\ct_{j+1,f[j+1]\oplus 1}$ to $\cB$.
    \item $\cB$ sends $(\{\sk_{i,f[i]}\}_{i\in [s]}, \{\ct_{i,f[i]\oplus 1}\}_{i\in[s]})$ to $\cA_1$.
    \item $\cA_1$ sends $\{\cert_i\}_{i\in[s]}$ to $\cB$, and sends its internal state to $\cA_2$.
    \item $\cB$ sends $\cert_{j+1}$ to the challenger, and receives $\sk_{j+1,f[j+1]\oplus1}$ or $\bot$ from the challenger.
    If $\cB$ receives $\bot$ from the challenger, it outputs $\bot$ and aborts.
    \item $\cB$ sends $\{\sk_{i,f[i]\oplus 1}\}_{i\in[s]}$ to $\cA_2$.
    \item $\cA_2$ outputs $b'$.
    \item $\cB$ computes $\Vrfy$ for all $\cert_i$, and outputs $b'$ if all results are $\top$. Otherwise, $\cB$ outputs $\bot$.
\end{enumerate}
Since $\pk_{j+1,\beta}$ and $\pk_{j+1,\beta\oplus 1}$ are identically distributed, 
it holds that
$ \Pr[f[j+1]=\beta]=\Pr[f[j+1]=\beta\oplus 1]=\frac{1}{2}$.
If $b=0$ and $f[j+1]=\beta\oplus 1$, $\cB$ simulates $\sfhyb{j}{}$.
Therefore, we have
\begin{align}
\Pr[1\la\cB|b=0]&=\Pr[1\la\cB \wedge f[j+1]=\beta\oplus 1|b=0]\\
                &=\Pr[1\la\cB|b=0, f[j+1]=\beta \oplus 1]\cdot\Pr[f[j+1]=\beta \oplus 1]\\
                &=\frac{1}{2}\Pr[\sfhyb{j}{}=1].
\end{align}
If $b=1$ and $f[j+1]=\beta\oplus 1$, $\cB$ simulates $\sfhyb{j+1}{}$.
Similarly, we have 
$
\Pr[1\la\cB|b=1]=\frac{1}{2}\Pr[\sfhyb{j+1}{}=1]
$.
By assumption, $\abs{\Pr[\sfhyb{j}{}=1]-\Pr[\sfhyb{j+1}{}=1]}$ is non-negligible, and therefore
$\abs{\Pr[1\la\cB|b=0]-\Pr[1\la\cB|b=1]}$ is non-negligible, which contradicts the certified everlasting IND-CPA security of $\Sigma$.
\end{proof}


\section{Proof of \cref{thm:garble_ever_security}}\label{sec:parallel}
Let $\widehat{\mathsf{gate}}_1,\widehat{\mathsf{gate}}_2,\cdots ,\widehat{\mathsf{gate}}_q$ be the topology of the gates $\mathsf{gate}_1,\mathsf{gate}_2,\cdots ,\mathsf{gate}_q$ which indicates how gates are connected.
In other words, if $\mathsf{gate}_i=(g,w_a,w_b,w_c)$, then $\widehat{\mathsf{gate}}_i=(\bot,w_a,w_b,w_c)$.
\begin{proof}[Proof of \cref{thm:garble_ever_security}]
First, let us define a simulator $\Sim$ as follows.
\begin{description}
\item[The simulator $\Sim(1^\secp,1^{|C|},C(x),\{L_{i,x_i}\}_{i\in[n]})$:] $ $
\begin{enumerate}
    \item Parse $\{L_{i,x_i}\}_{i\in[n]}\seteq\{\ske.\sk_{i}^{x_i}\}_{i\in[n]}$.
    \item For $i\in[n]$, generate $\ske.\sk_{i}^{x_i\oplus 1}\la \SKE.\keygen(1^{\secp})$. 
    \item For $i\in\{n+1,n+2,\cdots,p\}$ and $\sigma\in\bit$, generate $\ske.\sk_i^{\sigma}\la\SKE.\keygen(1^\secp)$.
    \item For each $i\in[q]$, compute $(\vk_i,\widetilde{g}_i)\la \Sim.\mathsf{GateGrbl}(\widehat{\mathsf{gate}_i}, \{\ske.\sk_a^\sigma,\ske.\sk_b^\sigma,\ske.\sk_c^\sigma\}_{\sigma\in\bit})$, where $\Sim.\mathsf{GateGrbl}$ is described in Fig~\ref{fig:Sim_Garble_circuit} and $\widehat{\mathsf{gate}_i}=(\bot,w_a,w_b,w_c)$.
    \item For each $i\in[m]$, generate $\widetilde{d_i}\seteq \left[\left(\ske.\sk^{0}_{\mathsf{out}_i}, C(x)_i\right),\left(\ske.\sk^{1}_{\mathsf{out}_i}, C(x)_i\oplus 1 \right)\right]$.
    \item Output $\widetilde{C}\seteq (\{\widetilde{g_i}\}_{i\in[q]},\{\widetilde{d_i}\}_{i\in[m]})$ and $\vk\seteq\{\vk_i\}_{i\in[q]}$.
\end{enumerate}
\end{description}

\protocol{Simulation Gate Garbling Circuit $\Sim.\mathsf{GateGrbl}$
}
{The description of $\Sim.\mathsf{GateGrbl}$}
{fig:Sim_Garble_circuit}
{
\begin{description}
\item[Input:] $(\widehat{\mathsf{gate}_i},\{\ske.\sk_a^\sigma,\ske.\sk_b^\sigma,\ske.\sk_c^\sigma\}_{\sigma\in\bit})$.
\item[Output:] $\widetilde{g_i}$ and $\vk_i$.
\end{description}
\begin{enumerate}
\item For each $\sigma_a,\sigma_b\in\bit$, sample $p^{\sigma_{a},\sigma_b}_{a,b}\la\Ks$.
\item Sample $\gamma_i\la\mathsf{S}_4$.
\item For each $\sigma_a,\sigma_b\in\bit$, compute $(\ske.\vk_{a}^{\sigma_a,\sigma_b},\ske.\ct_{a}^{\sigma_a,\sigma_b})\la\SKE.\Enc(\ske.\sk_a^{\sigma_a},p^{\sigma_a,\sigma_b}_{c})$
and 
$(\ske.\vk_{b}^{\sigma_a,\sigma_b},\ske.\ct_b^{\sigma_a,\sigma_b})\la
\SKE.\Enc(\ske.\sk_b^{\sigma_b},p^{\sigma_a,\sigma_b}_{c}\oplus \ske.\sk_c^{0})$.
\item Output $\widetilde{g_i}\seteq \{\ske.\ct_a^{\sigma_a,\sigma_b},\ske.\ct_b^{\sigma_a,\sigma_b}\}_{\sigma_a,\sigma_b\in\bit}$ in permutated order of $\gamma_i$ and 
$\vk_i\seteq\{\ske.\vk_{a}^{\sigma_a,\sigma_b},\ske.\vk_{b}^{\sigma_a,\sigma_b}\}_{\sigma_a,\sigma_b\in\bit}$ 
in permutated order of $\gamma_i$.
\end{enumerate}
}

For each $j\in[q]$, we define a QPT algorithm (a simulator) $\mathsf{InputDep}\Sim_j$ as follows.
\begin{description}
\item[The simulator $\mathsf{InputDep}\Sim_j(1^\secp,C,x,\{L_{i,x_i}\}_{i\in[n]})$:] $ $
\begin{enumerate}
    \item Parse $\{L_{i,x_i}\}_{i\in[n]}=\{\ske.\sk_{i}^{x_i}\}_{i\in[n]}$. 
    \item For $i\in[n]$, generate $\ske.\sk_{i}^{x_i\oplus 1}\la \SKE.\keygen(1^{\secp})$.
    \item For $i\in\{n+1,n+2,\cdots,p\}$ and $\sigma\in\bit$, generate $\ske.\sk_i^{\sigma}\la\SKE.\keygen(1^\secp)$.
    \item
    For $i\in[j]$, compute $(\vk_i,\widetilde{g}_i)\la \mathsf{InputDep.GateGrbl}(\mathsf{gate}_i, \{\ske.\sk_a^\sigma,\ske.\sk_b^\sigma,\ske.\sk_c^\sigma\}_{\sigma\in\bit})$, where $\mathsf{InputDep.GateGrbl}$ is described in Fig.~\ref{fig:Dep_Garble_circuit} and $\mathsf{gate}_i=(g,w_a,w_b,w_c)$
    \item For each $i\in \{j+1,j+2,\cdots,q\}$, compute $(\vk_i,\widetilde{g}_i)\la \mathsf{GateGrbl}(\mathsf{gate}_i, \{\ske.\sk_a^\sigma,\ske.\sk_b^\sigma,\ske.\sk_c^\sigma\}_{\sigma\in\bit})$, where $\mathsf{GateGrbl}$ is described in Fig~\ref{fig:Garble_circuit} and $\mathsf{gate}_i=(g,w_a,w_b,w_c)$.
    \item For each $i\in[m]$, generate $\widetilde{d_i}\seteq \left[\left(\ske.\sk^{0}_{\mathsf{out}_i}, 0\right),\left(\ske.\sk^{1}_{\mathsf{out}_i}, 1 \right)\right]$.
    \item Output $\widetilde{C}\seteq(\{\widetilde{g_i}\}_{i\in[q]},\{\widetilde{d_i}\}_{i\in[m]})$
    and $\vk\seteq\{\vk_i\}_{i\in[q]}$.
\end{enumerate}
\end{description}

\protocol{Input Dependent Gate Garbling Circuit $\mathsf{InputDep.GateGrbl}$
}
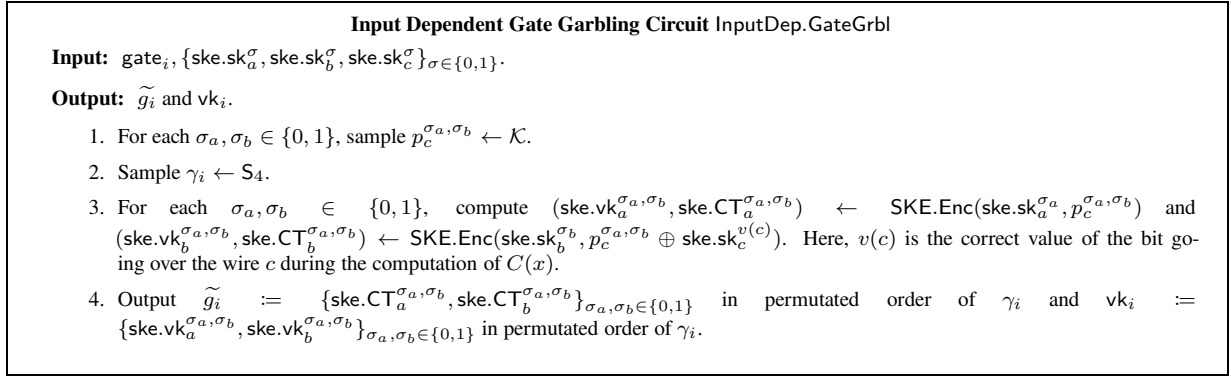
{The description of $\mathsf{InputDep.GateGrbl}$}
{fig:Dep_Garble_circuit}
{
\begin{description}
\item[Input:] $\mathsf{gate}_i,\{\ske.\sk_a^\sigma,\ske.\sk_b^\sigma,\ske.\sk_c^\sigma\}_{\sigma\in\bit}$.
\item[Output:] $\widetilde{g_i}$ and $\vk_i$.
\end{description}
\begin{enumerate}
\item For each $\sigma_a,\sigma_b\in\bit$, sample $p^{\sigma_{a},\sigma_b}_{c}\la\Ks$.
\item Sample $\gamma_i\leftarrow \mathsf{S}_4$.
\item For each $\sigma_a,\sigma_b\in\bit$, compute $(\ske.\vk_{a}^{\sigma_a,\sigma_b},\ske.\ct_{a}^{\sigma_a,\sigma_b})\la\SKE.\Enc(\ske.\sk_a^{\sigma_a},p^{\sigma_a,\sigma_b}_{c})$
and 
$(\ske.\vk_{b}^{\sigma_a,\sigma_b},\ske.\ct_b^{\sigma_a,\sigma_b})\la
\SKE.\Enc(\ske.\sk_b^{\sigma_b},p^{\sigma_a,\sigma_b}_{c}\oplus \ske.\sk_c^{v(c)})$.
Here, $v(c)$ is the correct value of the bit going over the wire $c$ during the computation of $C(x)$.
\item Output $\widetilde{g_i}\seteq \{\ske.\ct_a^{\sigma_a,\sigma_b},\ske.\ct_b^{\sigma_a,\sigma_b}\}_{\sigma_a,\sigma_b\in\bit}$ in permutated order of $\gamma_i$ and 
$\vk_i\seteq\{\ske.\vk_{a}^{\sigma_a,\sigma_b},\ske.\vk_{b}^{\sigma_a,\sigma_b}\}_{\sigma_a,\sigma_b\in\bit}$ in permutated order of $\gamma_i$.
\end{enumerate}
}

For each $j\in\{0,1,\cdots,q\}$, let us define a sequence of hybrid games $\{\mathsf{Hyb}_j\}_{j\in\{0,1,\cdots,q\}}$ against any adversary $\cA\seteq(\cA_1,\cA_2)$, where $\cA_1$ is any QPT adversary and $\cA_2$ is any unbounded adversary.
Note that  
\begin{align}
\mathsf{InputDep}\Sim_0(1^\secp,C,x,\{L_{i,x_i}\}_{i\in[n]})= \Garble(1^\secp,C,\{L_{i,\alpha}\}_{i\in[n],\alpha\in\bit}).  
\end{align}

\begin{description}
\item[The hybrid game $\mathsf{Hyb}_j$:] $ $
\begin{enumerate}
    \item $\cA_1$ sends a circuit $C\in\Cs_n$ and an input $x\in\bit^n$ to the challenger.
    \item The challenger computes $\{L_{i,\alpha}\}_{i\in[n],\alpha\in\bit}\la\Samp(1^\secp)$.
    \item  The challenger computes $(\widetilde{C},\vk)\la\GC.\mathsf{InputDep}\Sim_j(1^\secp,C,x,\{L_{i,x_i}\}_{i\in[n]})$,
    and sends $(\widetilde{C},\{L_{i,x_i}\}_{i\in[n]})$ to $\cA_1$.
    \item At some point, $\cA_1$ sends $\cert$ to the challenger and the internal state to $\cA_2$.
    \item The challenger computes $\Vrfy(\vk,\cert)\ra\top/\bot$. If the output is $\bot$, then the challenger outputs $\bot$ and sends $\bot$ to $\cA_2$.
    Else, the challenger outputs $\top$ and sends $\top$ to $\cA_2$.
    \item $\cA_2$ outputs $b'\in\bit$.
    \item If the challenger outputs $\top$, then the output of the experiment is $b'$. 
    Otherwise, the output of the experiment is $\bot$.
\end{enumerate}
\end{description}

Note that $\mathsf{Hyb}_0$ is identical to $\expb{\Sigma_{\mathsf{cegc}},\cA}{cert}{ever}{select}(1^\secp,0)$ by definition.
Therefore, \cref{thm:garble_ever_security} easily follows from the following \cref{prop:garble_hyb_j,prop:garble_hyb_q} (whose proofs are given later).
\end{proof}

\begin{proposition}\label{prop:garble_hyb_j}
If $\Sigma_{\mathsf{cesk}}$ satisfies the certified everlasting IND-CPA security, 
it holds that
$
\abs{\Pr[\mathsf{Hyb}_{j-1}=1]-\Pr[\mathsf{Hyb}_{j}=1]} \leq \negl(\secp)
$
for all $j\in[q]$.
\end{proposition}

\begin{proposition}\label{prop:garble_hyb_q}
$\abs{\Pr[\mathsf{Hyb}_q=1]-\Pr[\expb{\Sigma_{\mathsf{cegc}},\cA}{cert}{ever}{select}(1^\secp,1)=1]} \leq \negl(\secp)$.
\end{proposition}

\begin{proof}[Proof of \cref{prop:garble_hyb_j}]
For the proof, we use \cref{lem:parallel} whose statement and proof are given in \cref{sec:parallel}.
We construct an adversary $\cB$ that breaks the security experiment of $\expb{\Sigma_{\mathsf{cesk}},\cB}{parallel}{cert}{ever}(\secp,b)$,
which is described in \cref{lem:parallel}, assuming that 
$\abs{\Pr[\mathsf{Hyb}_{j-1}=1]-\Pr[\mathsf{Hyb}_{j}=1]}$ is non-negligible.
This contradicts the certified everlasting IND-CPA security of $\Sigma_{\mathsf{cesk}}$
from \cref{lem:parallel}.
Let us describe how $\cB$ works below.
\begin{enumerate}
    \item $\cB$ receives $C\in \Cs_n$ and $x\in\bit^n$ from $\cA_1$. Let $\mathsf{gate_{j}}=(g,w_{\alpha},w_{\beta},w_{\gamma})$.
    \item 
    The challenger of $\expb{\Sigma_{\mathsf{cesk}},\cB}{parallel}{cert}{ever}(\secp,b)$ generates $\ske.\sk_\alpha^{v(\alpha)\oplus 1}\la\SKE.\keygen(1^\secp)$ and $\ske.\sk_\beta^{v(\beta)\oplus 1}\la\SKE.\keygen(1^\secp)$\footnote{Recall that $v(\alpha)$ is the correct value of the bit going over the wire $\alpha$ during the computation of $C(x)$.}.
    \item 
    For each $i\in[p]\setminus\{\alpha,\beta\}$ and $\sigma\in\bit$, $\cB$ generates $\ske.\sk_i^\sigma\la\SKE.\keygen(1^\secp)$. 
    $\cB$ generates $\ske.\sk_\alpha^{v(\alpha)}\la\SKE.\keygen(1^\secp)$ and $\ske.\sk_{\beta}^{v(\beta)}\la\SKE.\keygen(1^\secp)$.
    $\cB$ sets $\{L_{i,x_i}\}_{i\in[n]}\seteq \{\ske.\sk_{i}^{x_i}\}_{i\in[n]}$.
    \item For each $i\in[j-1]$, $\cB$ computes $(\vk_i,\widetilde{g}_i)\la \mathsf{InputDep.GateGrbl}(\mathsf{gate}_i, \{\ske.\sk_a^\sigma,\ske.\sk_b^\sigma,\ske.\sk_c^\sigma\}_{\sigma\in\bit})$, where $\mathsf{InputDep.GateGrbl}$ is described in Fig~\ref{fig:Dep_Garble_circuit} and $\mathsf{gate}_i=(g,w_a,w_b,w_c)$.
    $\cB$ calls the encryption query of $\expb{\Sigma_{\mathsf{cesk}},\cB}{parallel}{cert}{ever}(\secp,b)$ if it needs to use $\ske.\sk_\alpha^{v(\alpha)\oplus 1}$ or $\ske.\sk_\beta^{v(\beta)\oplus 1}$ to run $(\vk_i,\widetilde{g}_i)\la \mathsf{InputDep.GateGrbl}(\mathsf{gate}_i,\\ \{\ske.\sk_a^\sigma,\ske.\sk_b^\sigma,\ske.\sk_c^\sigma\}_{\sigma\in\bit})$.
    \item 
    $\cB$ samples $p_\gamma^{v(\alpha),v(\beta)}\la \Ks$.
    $\cB$ computes 
    \begin{align}
    &(\ske.\vk_{\alpha}^{v(\alpha),v(\beta)},\ske.\ct_{\alpha}^{v(\alpha),v(\beta)})\la\SKE.\Enc(\ske.\sk_{\alpha}^{v(\alpha)},p^{v(\alpha),v(\beta)}_{\gamma}),\\
    &(\ske.\vk_{\beta}^{v(\alpha),v(\beta)},\ske.\ct_\beta^{v(\alpha),v(\beta)})\la\SKE.\Enc(\ske.\sk_\beta^{v(\beta)},p^{v(\alpha),v(\beta)}_{\gamma}\oplus \ske.\sk_\gamma^{v(\gamma)}).
    \end{align}
    \item $\cB$ sets 
    \begin{align}
    &(x_0,y_0,z_0)\seteq (\ske.\sk_{\gamma}^{g(v(\alpha),v(\beta)\oplus 1)},\ske.\sk_{\gamma}^{g(v(\alpha)\oplus 1,v(\beta))},\ske.\sk_{\gamma}^{g(v(\alpha)\oplus 1,v(\beta)\oplus 1)}),\\
    &(x_1,y_1,z_1)\seteq(\ske.\sk_\gamma^{v(\gamma)},\ske.\sk_\gamma^{v(\gamma)},\ske.\sk_\gamma^{v(\gamma)}),
    \end{align}
    and sends $(\ske.\sk_{\alpha}^{v(\alpha)},\ske.\sk_{\beta}^{v(\beta)},\{x_\sigma,y_\sigma,z_\sigma\}_{\sigma\in\bit})$ to the challenger of $\expb{\Sigma_{\mathsf{cesk}},\cB}{parallel}{cert}{ever}(\secp,b)$.
    \item The challenger samples $(x,y,z)\la \Ks^3$ and $(\ske.\sk_\alpha^{v(\alpha)\oplus 1},\ske.\sk_\beta^{v(\beta)\oplus 1})\la\keygen(1^\secp)$.
    The challenger computes 
    \begin{align}
    &(\ske.\vk_{\alpha}^{v(\alpha),v(\beta)\oplus 1},\ske.\ct_{\alpha}^{v(\alpha),v(\beta)\oplus 1})\la\Enc(\ske.\sk_\alpha^{v(\alpha)},x),\\
    &(\ske.\vk_{\beta}^{v(\alpha),v(\beta)\oplus 1},\ske.\ct_{\beta}^{v(\alpha),v(\beta)\oplus 1})\la\Enc(\ske.\sk_\beta^{v(\beta)\oplus 1},x\oplus x_b),\\
    &(\ske.\vk_{\alpha}^{v(\alpha)\oplus 1,v(\beta)},\ske.\ct_{\alpha}^{v(\alpha)\oplus 1,v(\beta)})\la\Enc(\ske.\sk_{\alpha}^{v(\alpha)\oplus 1},y),\\
    &(\ske.\vk_{\beta}^{v(\alpha)\oplus 1,v(\beta)},\ske.\ct_\beta^{v(\alpha)\oplus 1,v(\beta)})\la\Enc(\ske.\sk_{\beta}^{v(\beta)},y\oplus y_b),\\
    &(\ske.\vk_\alpha^{v(\alpha)\oplus 1,v(\beta)\oplus 1},\ske.\ct_\alpha^{v(\alpha)\oplus 1,v(\beta)\oplus 1})\la \Enc(\ske.\sk_\alpha^{v(\alpha)\oplus 1},z),\\
    &(\ske.\vk_\beta^{v(\alpha)\oplus 1,v(\beta)\oplus 1},\ske.\ct_\beta^{v(\alpha)\oplus 1,v(\beta)\oplus 1})\la \Enc(\ske.\sk_\beta^{v(\beta)\oplus 1},z\oplus z_b),
    \end{align}
    and sends
    \begin{align}
    (\ske.\ct_\alpha^{v(\alpha),v(\beta)\oplus 1},\ske.\ct_{\beta}^{v(\alpha),v(\beta)\oplus 1},\ske.\ct_\alpha^{v(\alpha)\oplus 1,v(\beta)},\ske.\ct_{\beta}^{v(\alpha)\oplus 1,v(\beta)},
    \ske.\ct_{\alpha}^{v(\alpha)\oplus 1,v(\beta)\oplus 1},
    \ske.\ct_{\beta}^{v(\alpha)\oplus 1,v(\beta)\oplus 1})
    \end{align}
    to $\cB$.
    \item $\cB$ samples $\gamma_j\la \mathsf{S}_4$.
    $\cB$ sets $\widetilde{g_{j}}\seteq \{\ske.\ct_\alpha^{\sigma_{\alpha},\sigma_\beta},\ske.\ct_\beta^{\sigma_\alpha,\sigma_\beta}\}_{\sigma_\alpha,\sigma_\beta\in\bit}$ in the permutated order of $\gamma_j$.
    \item For each $i\in\{j+1,j+2,\cdots, q\}$, $\cB$ computes $(\vk_i,\widetilde{g}_i)\la \mathsf{GateGrbl}(\mathsf{gate}_i, \{\ske.\sk_a^\sigma,\ske.\sk_b^\sigma,\ske.\sk_c^\sigma\}_{\sigma\in\bit})$, where $\cB$ calls the encryption query of $\expb{\Sigma_{\mathsf{cesk}},\cB}{parallel}{cert}{ever}(\secp,b)$ if $\cB$ needs to use $\ske.\sk_\alpha^{v(\alpha)\oplus 1}$ or $\ske.\sk_\beta^{v(\beta)\oplus 1}$ to run $(\vk_i,\widetilde{g}_i)\la \mathsf{GateGrbl}(\mathsf{gate}_i, \{\ske.\sk_a^\sigma,\ske.\sk_b^\sigma,\ske.\sk_c^\sigma\}_{\sigma\in\bit})$.
    \item $\cB$ computes $\widetilde{d_i}\seteq [(\ske.\sk^0_{\mathsf{out}_i},0),(\ske.\sk^1_{\mathsf{out}_i},1)]$ for $i\in[m]$,
    sets $\widetilde{C}\seteq (\{\widetilde{g_{i}}\}_{i\in[q]},\{\widetilde{d_i}\}_{i\in[m]})$,
    and sends $(\widetilde{C},\{L_{i,x_i}\}_{i\in[n]})$ to $\cA_1$.
    \item At some point, $\cA_1$ sends $\cert\seteq\{\cert_i\}_{i\in [q]}$ to $\cB$ and the internal state to $\cA_2$, respectively.
    \item 
    $\cB$ re-sorts $\cert_{j}=\{\ske.\cert_{\alpha}^{\sigma_{\alpha},\sigma_{\beta}},\ske.\cert_{\beta}^{\sigma_{\alpha},\sigma_{\beta}}\}_{\sigma_{\alpha},\sigma_{\beta}\in\bit}$ according to $\gamma_j$.
    $\cB$ sends 
    \begin{align}
    (\ske.\cert_{\alpha}^{v(\alpha),v(\beta)\oplus 1},
    \ske.\cert_{\beta}^{v(\alpha),v(\beta)\oplus 1},
    \ske.\cert_\alpha^{v(\alpha)\oplus 1,v(\beta)},
    \ske.\cert_\beta^{v(\alpha)\oplus 1,v(\beta)},
    \ske.\cert_\alpha^{v(\alpha)\oplus 1,v(\beta)\oplus 1},
    \ske.\cert_\beta^{v(\alpha)\oplus 1,v(\beta)\oplus 1})
    \end{align}
    to the challenger of $\expb{\Sigma_{\mathsf{cesk}},\cB}{parallel}{cert}{ever}(\secp,b)$ and receives $\bot$ or $(\ske.\sk_\alpha^{'v(\alpha)\oplus 1},\ske.\sk_\beta^{'v(\beta)\oplus 1})$ from the challenger.
    $\cB$ computes $\SKE.\Vrfy(\ske.\vk_{\alpha}^{v(\alpha),v(\beta)},\ske.\cert_{\alpha}^{v(\alpha),v(\beta)})$
    and $\SKE.\Vrfy(\ske.\vk_{\beta}^{v(\alpha),v(\beta)},\ske.\cert_{\beta}^{v(\alpha),v(\beta)})$.
    $\cB$ computes $\mathsf{GateVrfy}(\vk_i,\cert_i)$ for each $i\in\{1,2,\cdots,j-1,j+1,j+2,\cdots,q\}$,
    where $\mathsf{GateVrfy}$ is described in Fig\ref{fig:Gate_Verify}.
    If $\cB$ receives $(\ske.\sk_\alpha^{'v(\alpha)\oplus 1},\ske.\sk_\beta^{'v(\beta)\oplus 1})$ from the challenger,
    $\top\la\SKE.\Vrfy(\ske.\vk_{\alpha}^{v(\alpha),v(\beta)},\ske.\cert_{\alpha}^{v(\alpha),v(\beta)})$,
    $\top\la\SKE.\Vrfy(\ske.\vk_{\beta}^{v(\alpha),v(\beta)},\ske.\cert_{\beta}^{v(\alpha),v(\beta)})$,
    and $\top\la\mathsf{GateVrfy}(\cert_i,\vk_i)$ for all 
    $i\in\{1,2,\cdots,j-1,j+1,j+2,\cdots,q\}$, then $\cB$ sends $\top$ to $\cA_2$.
    Otherwise, $\cB$ sends $\bot$ to $\cA_2$, and aborts.
    \item $\cB$ outputs the output of $\cA_2$.
\end{enumerate}
It is clear that $\Pr[1\la\cB|b=0]=\Pr[\sfhyb{j-1}{}=1]$ and $\Pr[1\la\cB|b=1]=\Pr[\sfhyb{j}{}=1]$.
Therefore, if for an adversary $\cA$,
$
\abs{\Pr[\mathsf{Hyb}_{j-1}=1]-\Pr[\mathsf{Hyb}_{j}=1]}
$
is non-negligible, then
$
\allowbreak\left|\Pr[\expb{\Sigma_{\mathsf{cesk}},\cB}{parallel}{cert}{ever}(\secp,0)=1]-\Pr[\expb{\Sigma_{\mathsf{cesk}},\cB}{parallel}{cert}{ever}\allowbreak(\secp,1)=1]\right|
$
is non-negligible.
From \cref{lem:parallel}, it contradicts the certified everlasting IND-CPA security of $\Sigma_{\mathsf{cesk}}$ , which completes the proof.
\end{proof}

\begin{proof}[Proof of \cref{prop:garble_hyb_q}]
To show \cref{prop:garble_hyb_q}, it is sufficient to prove that the probability distribution of $\widetilde{C}$ in $\expb{\Sigma_{\mathsf{cegc}},\cA}{cert}{ever}{select}(1^\secp,1)$ is statistically identical to
that of $\widetilde{C}$ in $\mathsf{Hyb}_{q}$.

First, let us remind the difference between $\mathsf{Hyb}_{q}$ and $\expb{\Sigma_{\mathsf{cegc}},\cA}{cert}{ever}{select}(1^\secp,1)$. 
In both experiments , $\widetilde{C}$ consists of $\{\widetilde{g_i}\}_{i\in[q]}$ and $\{\widetilde{d_i}\}_{i\in[m]}$.
On the other hand the contents of $\{\widetilde{g_i}\}_{i\in[q]}$ and $\{\widetilde{d_i}\}_{i\in[m]}$ are different in each experiments.
In $\sfhyb{q}{}$,
$\widetilde{g_i}$ consists of $(\ske.\ct_a^{\sigma_a,\sigma_b},\ske.\ct_{b}^{\sigma_a,\sigma_b})$
where 
\begin{align}
&(\ske.\vk_{a}^{\sigma_a,\sigma_b},\ske.\ct_{a}^{\sigma_a,\sigma_b})\la\SKE.\Enc(\ske.\sk_a^{\sigma_a},p^{\sigma_a,\sigma_b}_{c}),\\
&(\ske.\vk_{b}^{\sigma_a,\sigma_b},\ske.\ct_b^{\sigma_a,\sigma_b})\la\SKE.\Enc(\ske.\sk_b^{\sigma_b},p^{\sigma_a,\sigma_b}_{c}\oplus \ske.\sk_c^{v(c)}),
\end{align}
and
$\widetilde{d_i}$ is
\begin{align}
    [(\ske.\sk_{\mathsf{out_i}}^0, 0),(\ske.\sk_{\mathsf{out_i}}^1,1)].
\end{align}

In $\expb{\Sigma_{\mathsf{cegc}},\cA}{cert}{ever}{select}(1^\secp,1)$, 
$\widetilde{g_i}$ consists of $(\ske.\ct_a^{\sigma_a,\sigma_b},\ske.\ct_{b}^{\sigma_a,\sigma_b})$
where
\begin{align}
&(\ske.\vk_{a}^{\sigma_a,\sigma_b},\ske.\ct_{a}^{\sigma_a,\sigma_b})\la\SKE.\Enc(\ske.\sk_a^{\sigma_a},p^{\sigma_a,\sigma_b}_{c}),\\
&(\ske.\vk_{b}^{\sigma_a,\sigma_b},\ske.\ct_b^{\sigma_a,\sigma_b})\la\SKE.\Enc(\ske.\sk_b^{\sigma_b},p^{\sigma_a,\sigma_b}_{c}\oplus \ske.\sk_c^{0}),
\end{align}
and
$\widetilde{d_i}$ is 
\begin{align}
[(\ske.\sk_{\mathsf{out_i}}^0, C(x)_i),(\ske.\sk_{\mathsf{out_i}}^1, C(x)_i\oplus 1)].
\end{align}

The resulting distribution of $(\{\widetilde{g_i}\}_{i\in[q]},\{\widetilde{d_i}\}_{i\in[m]})$ in $\mathsf{Hyb}_{q}$ is statistically identical to
the resulting distribution of $(\{\widetilde{g_i}\}_{i\in[q]},\{\widetilde{d_i}\}_{i\in[m]})$ in $\expb{\Sigma_{\mathsf{cegc}},\cA}{cert}{ever}{select}(1^\secp,1)$.
This is because, at any level that is not output, the keys $\ske.\sk_c^0,\ske.\sk_c^1$ are used completely identically in the subsequent level so there is no difference
between always encrypting $\ske.\sk_c^{v(c)}$ and $\ske.\sk_c^0$.
At the output level, there is no difference between encrypting $\ske.\sk_c^{v(c)}$ and giving the real mapping $\ske.\sk_c^{v(c)}\ra v(c)$ 
or encrypting $\ske.\sk_c^0$ and giving the programming mapping $\ske.\sk_c^0\ra C(x)_i$, which completes the proof.
\end{proof}

We use the following lemma for the proof of \cref{prop:garble_hyb_j}.
The proof is shown with the standard hybrid argument. It is also easy to see that a similar lemma holds for IND-CPA security.
\begin{lemma}\label{lem:parallel}
Let $\Sigma\seteq(\keygen,\Enc,\Dec,\Delete,\Vrfy)$ be a certified everlasting SKE scheme.
Let us consider the following security experiment 
$\expb{\Sigma,\cA}{parallel}{cert}{ever}(\secp,b)$ against $\cA$ consisting of any QPT adversary $\cA_1$ and any unbounded adversary $\cA_2$.
\begin{enumerate}
    \item The challenger generates $(\sk'^{0},\sk'^1)\la\keygen(1^\secp)$.
    \item $\cA_1$ can call encryption queries. More formally, it can do the followings: $\cA_1$ chooses $\beta\in\bit$, $\sk\in\mathcal{SK}$ and $m\in\Ms$.
    $\cA_1$ sends $(\beta,\sk,m)$ to the challenger. 
    \begin{itemize}
    \item 
     If $\beta=0$, the challenger generates $m^*\la\Ms$, computes $(\vk^{0}_{m},\ct^{0}_{m})\la\Enc(\sk'^{0},m^*)$ and $(\vk^{1}_m,\ct^{1}_m)\la\Enc(\sk,m\oplus m^*)$, and sends $\{\vk^\sigma_m,\ct^\sigma_m\}_{\sigma\in\bit}$ to $\cA_1$.
     \item If $\beta=1$, the challenger generates $m^*\la\Ms$, computes $(\vk^{1}_{m},\ct^{1}_{m})\la\Enc(\sk'^{1},m\oplus m^*)$ and
     $(\vk^{0}_{m},\ct^{0}_{m})\la\Enc(\sk,m^*)$, and sends $\{\vk^\sigma_m,\ct^\sigma_m\}_{\sigma\in\bit}$ to $\cA_1$.
    \end{itemize}
    $\cA_1$ can repeat this process polynomially many times.
    \item $\cA_1$ generates $(\sk^0,\sk^1)\la\keygen(1^\secp)$ and chooses two triples of messages $(x_0,y_0,z_0)\in\Ms^3$ and $(x_1,y_1,z_1)\in\Ms^3$, and sends $(\sk^0,\sk^1,\{x_\sigma,y_\sigma,z_\sigma\}_{\sigma\in\bit})$ to the challenger.
    \item The challenger generates $(x,y,z)\la\Ms^3$.
    The challenger computes
    \begin{align}
    &(\vk_{x}^0,\ct_{x}^0)\la\Enc(\sk^0,x),\,\,\, (\vk_{x}^1,\ct_{x}^1)\la\Enc(\sk'^1,x\oplus x_b)\\
    &(\vk_{y}^0,\ct_{y}^0)\la\Enc(\sk'^0,y),\,\,\, (\vk_y^1,\ct_y^1)\la\Enc(\sk^1,y\oplus y_b)\\
    &(\vk_z^0,\ct_z^0)\la \Enc(\sk'^0,z), \,\,\, (\vk_z^1,\ct_z^1)\la \Enc(\sk'^1,z\oplus z_b)
    \end{align}
    and sends $\{\ct_{x}^\sigma,\ct_y^\sigma,\ct_z^\sigma\}_{\sigma\in\bit}$ to $\cA_1$.
    \item $\cA_1$ can call encryption queries. More formally, it can do the followings:
    $\cA_1$ chooses $\beta\in\bit$, $\sk\in\mathcal{SK}$ and $m\in\Ms$.
    $\cA_1$ sends $(\beta,\sk,m)$ to the challenger. 
    \begin{itemize}
    \item 
     If $\beta=0$, the challenger generates $m^*\la\Ms$, computes $(\vk^{0}_{m},\ct^{0}_{m})\la\Enc(\sk'^{0},m^*)$ and $(\vk^{1}_m,\ct^{1}_m)\la\Enc(\sk,m\oplus m^*)$, and sends $\{\vk^\sigma_m,\ct^\sigma_m\}_{\sigma\in\bit}$ to $\cA_1$.
     \item If $\beta=1$, the challenger generates $m^*\la\Ms$, computes $(\vk^{1}_{m},\ct^{1}_{m})\la\Enc(\sk'^{1},m\oplus m^*)$ and
     $(\vk^{0}_{m},\ct^{0}_{m})\la\Enc(\sk,m^*)$, and sends $\{\vk^\sigma_m,\ct^\sigma_m\}_{\sigma\in\bit}$ to $\cA_1$.
    \end{itemize}
    $\cA_1$ can repeat this process polynomially many times.
    \item $\cA_1$ sends $\{\cert_x^\sigma,\cert_y^\sigma,\cert_z^\sigma\}_{\sigma\in\bit}$ to the challenger, and sends the internal state to $\cA_2$.
    \item The challenger computes $\Vrfy(\vk_x^\sigma,\cert_x^\sigma)$, $\Vrfy(\vk_y^\sigma,\cert_y^\sigma)$ and $\Vrfy(\vk_z^\sigma,\cert_z^\sigma)$ for each $\sigma\in\bit$. If all results are $\top$, then the challenger outputs $\top$, and sends $\{\sk'^\sigma\}_{\sigma\in\{0,1\}}$ to $\cA_2$.
    Otherwise, the challenger outputs $\bot$, and sends $\bot$ to $\cA_2$.
    \item $\cA_2$ outputs $b'\in\bit$.
    \item If the challenger outputs $\top$, then the output of the experiment is $b'$. Otherwise, the output of the experiment is $\bot$.
\end{enumerate}
If the $\Sigma$ satisfies the certified everlasting IND-CPA security,
\begin{align}
\advc{\Sigma,\cA}{parallel}{cert}{ever}(\secp)
\seteq \abs{\Pr[ \expb{\Sigma,\cA}{parallel}{cert}{ever}(\secp, 0)=1] - \Pr[ \expb{\Sigma,\cA}{parallel}{cert}{ever}(\secp, 1)=1] }\leq \negl(\secp)
\end{align}
for any QPT adversary $\cA_1$ and any unbounded adversary $\cA_2$.
\end{lemma}

\begin{proof}[Proof of \cref{lem:parallel}]
We define the following hybrid experiment.
\begin{description}
\item[$\sfhyb{1}{}$:] This is identical to $\expc{\Sigma,\cA}{parallel}{cert}{ever}(\secp, 0)$ except that the challenger encrypts $(x_0,y_0,z_1)$ instead of encrypting $(x_0,y_0,z_0)$.
\item[$\sfhyb{2}{}$:] This is identical to $\sfhyb{1}{}$ except that the challenger encrypts $(x_0,y_1,z_1)$ instead of encrypting $(x_0,y_0,z_1)$.
\end{description}

\cref{lem:parallel} easily follows from the following \cref{prop:Exp_Hyb_1_parallel,prop:Hyb_1_Hyb_2_parallel,prop:Hyb_2_Hyb_3_parallel} (whose proof is given later.).
\end{proof}

\begin{proposition}\label{prop:Exp_Hyb_1_parallel}
If $\Sigma$ is certified everlasting IND-CPA secure, it holds that
$
\abs{\Pr[\expc{\Sigma,\cA}{parallel}{cert}{ever}(\secp, 0)=1]-\Pr[\sfhyb{1}{}=1]}\leq \negl(\lambda).
$
\end{proposition}

\begin{proposition}\label{prop:Hyb_1_Hyb_2_parallel}
If $\Sigma$ is certified everlasting IND-CPA secure, it holds that
$
\abs{\Pr[\sfhyb{1}{}=1]-\Pr[\sfhyb{2}{}=1]}\leq \negl(\lambda).
$
\end{proposition}

\begin{proposition}\label{prop:Hyb_2_Hyb_3_parallel}
If $\Sigma$ is certified everlasting IND-CPA secure, it holds that
$
\abs{\Pr[\sfhyb{2}{}=1]-\Pr[\expc{\Sigma,\cA}{parallel}{cert}{ever}(\secp, 1)=1]}\leq \negl(\lambda).
$
\end{proposition}

\begin{proof}[Proof of \cref{prop:Exp_Hyb_1_parallel}]
We assume that
$
\abs{\Pr[\expc{\Sigma,\cA}{parallel}{cert}{ever}(\secp, 0)=1]-\Pr[\sfhyb{1}{}(1)=1]}
$
is non-negligible, and construct an adversary $\cB$ that breaks the security experiment of $\expd{\Sigma,\cB}{cert}{ever}{ind}{cpa}(\secp, b)$. 
This contradicts the certified everlasting IND-CPA security of $\Sigma$.
Let us describe how $\cB$ works.
\begin{enumerate}
    \item The challenger of $\expd{\Sigma,\cB}{cert}{ever}{ind}{cpa}(\secp,b)$ generates $\sk'^{0}\la\keygen(1^\secp)$, and $\cB$ generates $\sk'^{1}\la\keygen(1^\secp)$.
    \item $\cA_1$ chooses $\beta\in\bit$, $\sk\in\Ks$ and $m\in\Ms$.
    $\cA_1$ sends $(\beta,\sk,m)$ to $\cB$.
    \begin{itemize} 
    \item If $\beta=0$, $\cB$ generates $m^*\la\Ms$, sends $m^*$ to the challenger, receives $(\vk_m^0,\ct_m^0)$ from the challenger, computes $(\vk^{1}_m,\ct^{1}_m)\la\Enc(\sk,m\oplus m^*)$, and sends $\{\vk^\sigma_m,\ct^\sigma_m\}_{\sigma\in\bit}$ to $\cA_1$.
    \item If $\beta=1$, $\cB$ generates $m^*\la\Ms$, computes $(\vk^{1}_{m},\ct^{1}_{m})\la\Enc(\sk'^{1},m\oplus m^*)$ and $(\vk^{0}_{m},\ct^{0}_{m})\la\Enc(\sk,m^*)$ and sends $\{\vk^\sigma_m,\ct^\sigma_m\}_{\sigma\in\bit}$ to $\cA_1$.
    \end{itemize}
    $\cB$ repeats this process when $(\beta,\sk,m)$ is sent from $\cA_1$.
    \item $\cB$ receives $(\sk^0,\sk^1,\{x_\sigma,y_\sigma,z_\sigma\}_{\sigma\in\bit})$ from $\cA_1$.
    \item $\cB$ generates $(x,y,z)\la\Ms^3$.
          $\cB$ computes 
          \begin{align}
          &(\vk_{x}^0,\ct_{x}^0)\la\Enc(\sk^0,x), (\vk_{x}^1,\ct_{x}^1)\la\Enc(\sk'^1,x\oplus x_0),\\
          &\hspace{2cm}(\vk_y^1,\ct_y^1)\la\Enc(\sk^1,y\oplus y_0),\\
          &\hspace{2cm}(\vk_z^1,\ct_z^1)\la \Enc(\sk'^1,z\oplus z_0).
          \end{align}
    \item $\cB$ sets $m_0\seteq z$ and $m_1\seteq z\oplus z_0\oplus z_1$. 
          $\cB$ sends $(m_0,m_1)$ to the challenger.
    \item The challenger computes $(\vk_z^0,\ct_z^0)\la \Enc(\sk'^0,m_{b})$, and sends $\ct_z^0$ to $\cB$.
    \item $\cB$ sends an encryption query $y$ to the challenger, and receives $(\vk_y^0,\ct_y^0)$.
    \item $\cB$ sends $\{\ct_{x}^\sigma,\ct_y^\sigma,\ct_z^\sigma\}_{\sigma\in\bit}$ to $\cA_1$.
    \item $\cA_1$ chooses $\beta\in\bit$, $\sk\in\Ks$ and $m\in\Ms$.
    $\cA_1$ sends $(\beta,\sk,m)$ to $\cB$.
    \begin{itemize} 
    \item If $\beta=0$, $\cB$ generates $m^*\la\Ms$, sends $m^*$ to the challenger, receives $(\vk_m^0,\ct_m^0)$ from the challenger, computes $(\vk^{1}_m,\ct^{1}_m)\la\Enc(\sk,m\oplus m^*)$, and sends $\{\vk^\sigma_m,\ct^\sigma_m\}_{\sigma\in\bit}$ to $\cA_1$.
    \item If $\beta=1$, $\cB$ generates $m^*\la\Ms$, computes $(\vk^{1}_{m},\ct^{1}_{m})\la\Enc(\sk'^{1},m\oplus m^*)$ and $(\vk^{0}_{m},\ct^{0}_{m})\la\Enc(\sk,m^*)$ and sends $\{\vk^\sigma_m,\ct^\sigma_m\}_{\sigma\in\bit}$ to $\cA_1$.
    \end{itemize}
    $\cB$ repeats this process when $(\beta,\sk,m)$ is sent from $\cA_1$.
    \item $\cA_1$ sends $\{\cert^{\sigma}_x,\cert^{\sigma}_y,\cert^{\sigma}_z\}_{\sigma\in\bit}$ to $\cB$, and sends the internal state to $\cA_2$.
    \item $\cB$ sends $\cert_z^0$ to the challenger, and receives $\sk'^0$ or $\bot$ from the challenger.
    If $\cB$ receives $\bot$, it outputs $\bot$ and aborts.
    \item $\cB$ sends $\{\sk'^{\sigma}\}_{\sigma\in\bit}$ to $\cA_2$.
    \item $\cA_2$ outputs $b'$.
    \item $\cB$ computes $\Vrfy(\vk^\sigma_x,\cert_x^{\sigma})$ and $\Vrfy(\vk^\sigma_y,\cert_y^{\sigma})$ for each $\sigma\in\bit$, and $\Vrfy(\vk_z^{1},\cert_z^1)$.
    If all results are $\top$, $\cB$ outputs $b'$.
    Otherwise, $\cB$ outputs $\bot$.
\end{enumerate}
It is clear that $\Pr[1\la\cB|b=0]=\Pr[\expc{\Sigma,\cA}{parallel}{cert}{ever}(\secp, 0)=1]$.
Since $z$ is uniformly distributed, $(z,z\oplus z_1)$ and $(z\oplus z_0\oplus z_1,z\oplus z_0)$ are identically distributed. 
Therefore, it holds that $\Pr[1\la\cB|b=1]=\Pr[\sfhyb{1}{}=1]$.
By assumption, $\abs{\Pr[\expc{\Sigma,\cA}{parallel}{cert}{ever}(\secp, 0)=1]-\Pr[\sfhyb{1}{}=1]}$ is non-negligible, and therefore $\abs{\Pr[1\la\cB|b=0]-\Pr[1\la\cB|b=1]}$ is non-negligible, which contradicts the certified everlasting IND-CPA security of $\Sigma_{\mathsf{cesk}}$.
\end{proof}

\begin{proof}[Proof of \cref{prop:Hyb_1_Hyb_2_parallel}]
The proof is very similar to that of \cref{prop:Exp_Hyb_1_parallel}.
Therefore we skip the proof.
\end{proof}

\begin{proof}[Proof of \cref{prop:Hyb_2_Hyb_3_parallel}]
The proof is very similar to that of \cref{prop:Exp_Hyb_1_parallel}.
Therefore, we skip the proof.
\end{proof}

\section{Proof of \cref{thm:func_ever_non_ad_security}}\label{sec:proof_non_ad_fe}
\begin{proof}[Proof of \cref{thm:func_ever_non_ad_security}]
Let us describe how the simulator $\Sim$ works.
\begin{description}
\item[$\Sim(\MPK,\mathcal{V},1^{|m|})$:]$ $
\begin{enumerate}
    \item Parse $\MPK=\{\pke.\pk_{i,\alpha}\}_{i\in[s],\alpha\in\bit}$ and $\mathcal{V}=\{f(m),f,(f,\{\pke.\sk_{i,f[i]}\}_{i\in[s]})\} \,\, or \,\,\emptyset$.
    \item If $\mathcal{V}=\emptyset$, generate $f\la\bit^s$.
    \item Generate $\{L_{i,\alpha}\}_{i\in[s],\alpha\in\bit}\la\mathsf{GC}.\Samp(1^\secp)$ and $L_{i,f[i]\oplus1}^*\la\mathcal{L}$ for every $i\in[s]$.
    \item Compute $(\widetilde{U},\mathsf{gc}.\vk)\la\mathsf{GC}.\Sim(1^\secp,1^{|f|},U(f,m),\{L_{i,f[i]}\}_{i\in[s]})$.
    \item Compute $(\pke.\vk_{i,f[i]},\pke.\ct_{i,f[i]})\la \PKE.\Enc(\pke.\pk_{i,f[i]},L_{i,f[i]})$ and $(\pke.\vk_{i,f[i]\oplus 1},\pke.\ct_{i,f[i]\oplus 1})\la \PKE.\Enc(\pke.\pk_{i,f[i]\oplus 1},L^*_{i,f[i]\oplus 1})$ for every $i\in[s]$.
    \item Output $\vk\seteq (\mathsf{gc}.\vk,\{\pke.\vk_{i,\alpha}\}_{i\in[s],\alpha\in\bit})$ and $\ct\seteq (\widetilde{U},\{\pke.\ct_{i,\alpha}\}_{i\in[s],\alpha\in\bit})$.
\end{enumerate}
\end{description}
Let us define the sequence of hybrids as follows.

\begin{description}
\item [$\sfhyb{0}{}$:] This is identical to $\expd{\Sigma_{\mathsf{cefe}},\cA}{cert}{ever}{non}{adapt}(\secp,0)$.
\begin{enumerate}
    \item The challenger generates $(\pke.\pk_{i,\alpha},\pke.\sk_{i,\alpha})\la\PKE.\keygen(1^\secp)$ for every $i\in [s]$ and $\alpha\in\bit$, and sends
    $\{\pke.\pk_{i,\alpha}\}_{i\in[s],\alpha\in\bit}$ to $\cA_1$.
    \item\label{step:function_query_single_fe}
    $\cA_1$ is allowed to call a key query at most one time.
    If a key query is called,
    the challenger receives an function $f$ from $\cA_1$, and sends $(f,\{\pke.\sk_{i,f[i]}\}_{i\in[s]})$ to $\cA_1$.
    \item $\cA_1$ chooses $m\in\Ms$, and sends $m$ to the challenger.
    \item \label{step:encryption_single_non_ad_fe}
    The challenger computes $\{L_{i,\alpha}\}_{i\in[s],\alpha\in\bit}\la\mathsf{GC}.\Samp(1^\secp)$, $(\widetilde{U},\mathsf{gc}.\vk)\la\mathsf{GC}.\Garble(1^\secp,U(\cdot,m),\{L_{i,\alpha}\}_{i\in[s],\alpha\in\bit})$, and
    $(\pke.\vk_{i,\alpha},\pke.\ct_{i,\alpha})\la\PKE.\Enc(\pke.\pk_{i,\alpha},L_{i,\alpha})$ for every $i\in[s]$ and $\alpha\in\bit$,
    and sends $(\widetilde{U},\{\pke.\ct_{i,\alpha}\}_{i\in[s],\alpha\in\bit})$ to $\cA_1$.
    \item $\cA_1$ sends $(\mathsf{gc}.\cert,\{\pke.\cert_{i,\alpha}\}_{i\in[s],\alpha\in\bit})$ to the challenger,
    and sends its internal state to $\cA_2$.
    \item If $\top\la \mathsf{GC}.\Vrfy(\mathsf{gc}.\vk,\mathsf{gc}.\cert)$, and $\top\la\PKE.\Vrfy(\pke.\vk_{i,\alpha},\pke.\cert_{i,\alpha})$ for every $i\in[s]$ and $\alpha\in\bit$, the challenger outputs $\top$, and  sends $\{\pke.\sk_{i,\alpha}\}_{i\in[s],\alpha\in\bit}$ to $\cA_2$.
    Otherwise, the challenger outputs $\bot$, and sends $\bot$ to $\cA_2$.
    \item $\cA_2$ outputs $b'$. If the challenger outputs $\top$, the output of the experiment is $b'$.
    Otherwise, the output of the experiment is $\bot$.
\end{enumerate}
\item[$\sfhyb{1}{}$:]This is identical to $\sfhyb{0}{}$ except for the following four points.
First, the challenger generates $f\in\bit^s$ if a key query is not called in step~\ref{step:function_query_single_fe}.
Second, the challenger randomly generates $L^*_{i,f[i]\oplus 1}\la \mathcal{L}$ for every $i\in[s]$
and $\{L_{i,\alpha}\}_{i\in[s],\alpha\in\bit}\la\mathsf{GC}.\Samp(1^\secp)$ in step~\ref{step:function_query_single_fe} regardless of whether a key query is called or not.
Third, the challenger does not compute $\{L_{i,\alpha}\}_{i\in[s],\alpha\in\bit}\la\mathsf{GC}.\Samp(1^\secp)$ in step~\ref{step:encryption_single_non_ad_fe}.
Fourth, 
the challenger computes 
$(\pke.\vk_{i,f[i]\oplus 1},\pke.\ct_{i,f[i]\oplus 1})\la\PKE.\Enc(\pke.\pk_{i,f[i]\oplus 1},L^*_{i,f[i]\oplus 1})$ for every $i\in[s]$ instead of computing 
$(\pke.\vk_{i,f[i]\oplus 1},\pke.\ct_{i,f[i]\oplus 1})\la\PKE.\Enc(\pke.\pk_{i,f[i]\oplus 1},L_{i,f[i]\oplus 1})$ for every $i\in[s]$.

\item[$\sfhyb{2}{}$:]This is identical to $\sfhyb{1}{}$ except for the following point.
The challenger computes $(\widetilde{U},\mathsf{gc}.\vk)\la\GC.\Sim(1^\secp,1^{|f|},\allowbreak U(f,m), \{L_{i,f[i]}\}_{i\in[s]})$
instead of computing $(\widetilde{U},\mathsf{gc}.\vk)\la\mathsf{GC}.\Garble(1^\secp,U(\cdot,m),\{L_{i,\alpha}\}_{i\in[s],\alpha\in\bit})$.
\end{description}
From the definition of $\expd{\Sigma_{\mathsf{cefe}},\cA}{cert}{ever}{non}{adapt}(\secp, b)$ and $\Sim$,
it is clear that $\Pr[\sfhyb{0}{}=1]=\Pr[\expd{\Sigma_{\mathsf{cefe}},\cA}{cert}{ever}{non}{adapt}(\secp,0)=1]$ and
$\Pr[\sfhyb{2}{}=1]=\Pr[\expd{\Sigma_{\mathsf{cefe}},\cA}{cert}{ever}{non}{adapt}(\secp,1)=1]$.
Therefore, \cref{thm:func_ever_non_ad_security} easily follows from the following \cref{prop:Exp_Hyb_1_non_ad_fe,prop:Hyb_1_Hyb_2_non_ad_fe}. (whose proof is given later.)
\end{proof}

\begin{proposition}\label{prop:Exp_Hyb_1_non_ad_fe}
If $\Sigma_{\mathsf{cepk}}$ satisfies the certified everlasting IND-CPA security,
\begin{align}
    \abs{\Pr[\sfhyb{0}{}=1]-\Pr[\sfhyb{1}{}=1]}\leq\negl(\lambda).
\end{align}
\end{proposition}
\begin{proposition}\label{prop:Hyb_1_Hyb_2_non_ad_fe}
If $\Sigma_{\mathsf{cegc}}$ satisfies the certified everlasting selective security,
\begin{align}
    \abs{\Pr[\sfhyb{1}{}=1]-\Pr[\sfhyb{2}{}=1]}\leq\negl(\lambda).
\end{align}
\end{proposition}

\begin{proof}[Proof of \cref{prop:Exp_Hyb_1_non_ad_fe}]

For the proof, we use \cref{lem:cut_and_choose_pke} whose statement and proof is given in \cref{sec:proof_rnce}.
We assume that $\abs{\Pr[\sfhyb{0}{}=1]-\Pr[\sfhyb{1}{}=1]}$ is non-negligible, and construct an adversary $\cB$ that breaks the security experiment of  $\expc{\Sigma_{\mathsf{cepk}},\cB}{multi}{cert}{ever}(\secp,b)$ defined in \cref{lem:cut_and_choose_pke}.
This contradicts the certified everlasting IND-CPA of $\Sigma_{\mathsf{cepk}}$ from \cref{lem:cut_and_choose_pke}.
Let us describe how $\cB$ works below.
\begin{enumerate}
    \item $\cB$ receives $\{\pke.\pk_{i,\alpha}\}_{i\in[s],\alpha\in\bit}$ from the challenger of $\expc{\Sigma_{\mathsf{cepk}},\cB}{multi}{cert}{ever}(\secp,b)$,
    and sends $\{\pke.\pk_{i,\alpha}\}_{i\in[s],\alpha\in\bit}$ to $\cA_1$.
    \item $\cA_1$ is allowed to call a key query at most one time.
    If a key query is called,
    $\cB$ receives an function $f$ from $\cA_1$, generates $L^*_{i,f[i]\oplus 1}\la \mathcal{L}$ for every $i\in[s]$ and  $\{L_{i,\alpha}\}_{i\in[s],\alpha\in\bit}\la\mathsf{GC}.\Samp(1^\secp)$.
    If a key query is not called, $\cB$ generates $f\la\bit^s$, $L^*_{i,f[i]\oplus 1}\la \mathcal{L}$ for every $i\in[s]$ and  $\{L_{i,\alpha}\}_{i\in[s],\alpha\in\bit}\la\mathsf{GC}.\Samp(1^\secp)$. 
    \item $\cB$ sends $(f,L_{1,f[1]\oplus 1},L_{2,f[2]\oplus 1},\cdots ,L_{s,f[s]\oplus 1},L^*_{1,f[1]\oplus 1},L^*_{2,f[2]\oplus 1},\cdots, L^{*}_{s,f[s]\oplus 1})$ to the challenger of $\expc{\Sigma_{\mathsf{cepk}},\cB}{multi}{cert}{ever}(\secp,b)$.
    \item $\cB$ receives $(\{\pke.\sk_{i,f[i]}\}_{i\in[s]},\{\pke.\ct_{i,f[i]\oplus 1}\}_{i\in[s]})$ from the challenger.
    If a key query is called, $\cB$ sends $(f,\{\pke.\sk_{i,f[i]}\}_{i\in[s]})$ to $\cA_1$.
    \item $\cA_1$ chooses $m\in\Ms$, and sends $m$ to $\cB$.
    \item $\cB$ computes
    $(\widetilde{U},\mathsf{gc}.\vk)\la\mathsf{GC}.\Garble(1^\secp,U(\cdot,m),\{L_{i,\alpha}\}_{i\in[s],\alpha\in\bit})$ and
    $(\pke.\vk_{i,f[i]},\pke.\ct_{i,f[i]})\la\PKE.\Enc(\allowbreak \pke.\pk_{i,f[i]},\allowbreak L_{i,f[i]})$ for every $i\in[s]$, and sends $(\widetilde{U},\{\pke.\ct_{i,\alpha}\}_{i\in[s],\alpha\in\bit})$ to $\cA_1$.
    \item $\cA_1$ sends $(\mathsf{gc}.\cert,\{\pke.\cert_{i,\alpha}\}_{i\in[s],\alpha\in\bit})$ to $\cB$, and sends its internal state to $\cA_2$.
    \item $\cB$ sends $\{\pke.\cert_{i,f[i]\oplus1}\}_{i\in[s]}$ to the challenger, and receives $\{\pke.\sk_{i,f[i]\oplus 1}\}_{i\in[s]}$ or $\bot$ from the challenger. If $\cB$ receives $\bot$ from the challenger, it outputs $\bot$ and aborts. 
    \item $\cB$ sends $\{\pke.\sk_{i,\alpha}\}_{i\in[s],\alpha\in\bit}$ to $\cA_2$.
    \item $\cA_2$ outputs $b'$.
    \item $\cB$ computes $\mathsf{GC}.\Vrfy$ for $\gc.\cert$ and $\PKE.\Vrfy$ for all $\{\pke.\cert_{i,f[i]}\}_{i\in[s]}$, and outputs $b'$ if all results are $\top$.
    Otherwise, $\cB$ outputs $\bot$. 
\end{enumerate}
It is clear that $\Pr[1\la\cB|b=0]=\Pr[\sfhyb{0}{}=1]$ and 
$\Pr[1\la\cB|b=1]=\Pr[\sfhyb{1}{}=1]$.
By assumption, $\allowbreak \abs{\Pr[\sfhyb{0}{}=1]- \Pr[\sfhyb{1}{}=1]}$ is non-negligible, and therefore $\abs{\Pr[1\la\cB|b=0]-\Pr[1\la\cB|b=1]}$ is non-negligible, which contradicts the certified everlasting IND-CPA security of $\Sigma_{\mathsf{cepk}}$ from \cref{lem:cut_and_choose_pke}.
\end{proof}

\begin{proof}[Proof of \cref{prop:Hyb_1_Hyb_2_non_ad_fe}]

We assume that $\abs{\Pr[\sfhyb{1}{}=1]-\Pr[\sfhyb{2}{}=1]}$ is non-negligible, and construct an adversary $\cB$ that breaks the certified everlasting selective security of $\Sigma_{\mathsf{cegc}}$.
Let us describe how $\cB$ works below.
\begin{enumerate}
    \item $\cB$ generates $(\pke.\pk_{i,\alpha},\pke.\sk_{i,\alpha})\la\PKE.\keygen(1^\secp)$ for every $i\in [s]$ and $\alpha\in\bit$, and sends
    $\{\pke.\pk_{i,\alpha}\}_{i\in[s],\alpha\in\bit}$ to $\cA_1$.
    \item $\cA_1$ is allowed to call a key query at most one time.
    If a key query is called, $\cB$ receives an function $f$ from $\cA_1$, generates $L^*_{i,f[i]\oplus 1}\la \mathcal{L}$ for every $i\in[s]$, and sends $(f,\{\pke.\sk_{i,f[i]}\}_{i\in[s]})$ to $\cA_1$.
    If a key query is not called, $\cB$ generates $f\la\bit^s$ and $L^*_{i,f[i]\oplus 1}\la \mathcal{L}$ for every $i\in[s]$.
    \item $\cA_1$ chooses $m\in\Ms$, and sends $m$ to $\cB$.
    \item $\cB$ sends a circuit $U(\cdot,m)$ and an input $f\in\bit^s$ to the challenger of $\expb{\cB,\Sigma_{\mathsf{cegc}}}{cert}{ever}{selct}(1^\secp,b)$.
    \item The challenger computes $\{L_{i,\alpha}\}_{i\in[s],\alpha\in\bit}\la\mathsf{GC}.\Samp(1^\secp)$ and does the following:
    \begin{itemize}
        \item If $b=0$, the challenger computes $(\widetilde{U},\mathsf{gc}.\vk)\la\mathsf{GC}.\Garble(1^\secp,U(\cdot,m),\{L_{i,\alpha}\}_{i\in[s],\alpha\in\bit})$,
        and sends $(\widetilde{U},\{L_{i,f[i]}\}_{i\in[s]})$ to $\cB$.
        \item If $b=1$, the challenger computes $(\widetilde{U},\mathsf{gc}.\vk)\la\GC.\Sim(1^\secp,1^{|f|},U(f,m),\{L_{i,f[i]}\}_{i\in[s]})$,
        and sends $(\widetilde{U},\{L_{i,f[i]}\}_{i\in[s]})$ to $\cB$.
    \end{itemize}
    \item $\cB$ computes $(\pke.\vk_{i,f[i]},\pke.\ct_{i,f[i]})\la\PKE.\Enc(\pke.\pk_{i,f[i]},L_{i,f[i]})$ and $(\pke.\vk_{i,f[i]\oplus 1},\pke.\ct_{i,f[i]\oplus 1})\la\PKE.\Enc(\pke.\pk_{i,f[i]\oplus 1},L^*_{i,f[i]\oplus 1}) $ for every $i\in[s]$.
    \item $\cB$ sends $(\widetilde{U},\{\pke.\ct_{i,\alpha}\}_{i\in[s],\alpha\in\bit})$ to $\cA_1$.
    \item $\cA_1$ sends $(\mathsf{gc}.\cert,\{\pke.\cert_{i,\alpha}\}_{i\in[s],\alpha\in\bit})$ to the challenger, and sends its internal state to $\cA_2$.
    \item $\cB$ sends $\mathsf{gc}.\cert$ to the challenger, and receives $\top$ or $\bot$ from the challenger.
    If $\cB$ receives $\bot$ from the challenger, it outputs $\bot$ and aborts.
    \item $\cB$ sends $\{\pke.\sk_{i,\alpha}\}_{i\in[s],\alpha\in\bit}$ to $\cA_2$.
    \item $\cA_2$ outputs $b'$.
    \item $\cB$ computes $\PKE.\Vrfy$ for all $\pke.\cert_{i,\alpha}$, and outputs $b'$ if all results are $\top$.
    Otherwise, $\cB$ outputs $\bot$.
\end{enumerate}
It is clear that $\Pr[1\la\cB|b=0]=\Pr[\sfhyb{1}{}=1]$ and $\Pr[1\la\cB|b=1]=\Pr[\sfhyb{2}{}=1]$.
By assumption, $\abs{\Pr[\sfhyb{1}{}=1]-\Pr[\sfhyb{2}{}=1]}$ is non-negligible, and therefore $\abs{\Pr[1\la\cB|b=0]-\Pr[1\la\cB|b=1]}$ is non-negligible,
which contradicts the certified everlasting selective security of $\Sigma_{\mathsf{cegc}}$. 
\end{proof}


\section{Proof of \cref{thm:ever_security_single_ad}}\label{sec:proof_ada_fe}
\begin{proof}[Proof of \cref{thm:ever_security_single_ad}]

For a given $2n$-qubit, let $A$ be the $n$-qubit of the first half of the $2n$-qubit, and let $B$ be the $n$-qubit of the second half of the $2n$-qubit.
Let $\NAD.\Sim$ be the simulating algorithm of the ciphertext $\nad.\ct$ .
Let us describe how the simulator $\Sim=(\Sim_1,\Sim_2,\Sim_3)$ works below.
\begin{description}
\item[$\Sim_1(\MPK,\mathcal{V},1^{|m|})$:]$ $
\begin{enumerate}
\item Parse $\MPK=(\nad.\MPK,\nce.\pk)$ and $\mathcal{V}=(f,f(m),(\nad.\sk_f,\nce.\sk))\,\, or\,\, \emptyset$.
\footnote{
If an adversary calls a key query before the adversary receives a challenge ciphertext,
then $\mathcal{V}=(f,f(m),(\nad.\sk_f,\nce.\sk))$.
Otherwise, $\mathcal{V}=\emptyset$.
}
\item $\Sim_1$ does the following:
\begin{itemize}
\item If $\mathcal{V}=\emptyset$, generate $\ket{\widetilde{0^n0^n}}$ and $(\nce.\vk,\widetilde{\nce.\ct},\nce.\aux)\la\NCE.\Fake(\nce.\pk)$.
Let $\Psi_A\seteq\Tr_B(\ket{\widetilde{0^n0^n}}\bra{\widetilde{0^n0^n}})$ and
$\Psi_B\seteq\Tr_A(\ket{\widetilde{0^n0^n}}\bra{\widetilde{0^n0^n}})$.
Output $\ct\seteq (\Psi_A,\widetilde{\nce.\ct})$ and $\state\seteq (\nce.\aux,\nce.\pk,\nad.\MPK,\Psi_B,1^{|m|},\nce.\vk,0)$.
\item If $\mathcal{V}=(f,f(m),(\nad.\sk_f,\nce.\sk))$, generate $a,c\la\bit^n$, $(\nce.\vk,\nce.\ct)\la\NCE.\Enc(\nce.\pk,(a,c))$, $(\nad.\vk,\nad.\ct)\la\NAD.\Sim(\nad.\MPK,(f,f(m),\nad.\sk_f),1^{|m|})$ and $\Psi\seteq Z^cX^a\nad.\ct X^a Z^c$.
Output $\ct\seteq (\Psi,\nce.\ct)$ and $\state\seteq(\nad.\vk,\nce.\vk,a,c,1)$.
\end{itemize}

\end{enumerate}
\item[$\Sim_2(\MSK,f,f(m),\state)$:]$ $
\begin{enumerate}
\item Parse $\MSK\seteq (\nad.\MSK,\nce.\MSK)$ and $\state =(\nce.\aux,\nce.\pk,\nad.\MPK,\Psi_B,1^{|m|},\nce.\vk,0)$.
\item Compute $\nad.\sk_f\la\NAD.\keygen(\nad.\MSK,f)$.
\item Compute $(\nad.\vk,\nad.\ct)\la\NAD.\Sim(\nad.\MPK,(f,f(m),\nad.\sk_f),1^{|m|})$. 
Measure the $i$-th qubit of $\mathsf{nad}.\ct$ and $\Psi_B$ in the Bell basis and let $(x_i,z_i)$ be the measurement outcome for all $i\in[N]$.
\item Compute $\widetilde{\nce.\sk}\la\NCE.\Reveal(\nce.\pk,\nce.\MSK,\nce.\aux,(x,z))$.
\item Output $\sk_f\seteq (\nad.\sk_f,\widetilde{\nce.\sk})$ and $\state'\seteq(\nad.\vk,\nce.\vk,x,z,1)$.
\end{enumerate}
\item[$\Sim_3(\state^*)$:]$ $
\begin{enumerate}
\item Parse  $\state^* = (\nad.\vk,\nce.\vk,x^*,z^*,1)$ or $\state^* =(\nce.\aux,\nce.\pk,\nad.\MPK,\Psi_B,1^{|m|},\nce.\vk,0)$.
\item $\Sim_3$ does the following:
\begin{itemize}
    \item If the final bit of $\state^*$ is $0$,
    compute $(\nad.\vk,\nad.\ct)\la\NAD.\Sim(\nad.\MPK,\emptyset,1^{|m|})$. 
    Measure the $i$-th qubit of $\mathsf{nad}.\ct$ and $\Psi_B$ in the Bell basis and let $(x_i,z_i)$ be the measurement outcome for all $i\in[N]$.
     Output $\vk\seteq (\nad.\vk,\nce.\vk,x,z)$.
     \item If the final bit of $\state^*$ is $1$, output $\vk\seteq(\nad.\vk,\nce.\vk,x^*,z^*)$.
\end{itemize}
\end{enumerate}
Let us define the sequence of hybrids as follows.
\item[$\sfhyb{0}{}$:]This is identical to $\expc{\Sigma_{\mathsf{cefe},\cA}}{cert}{ever}{adapt}(0)$.
\begin{enumerate}
    \item The challenger generates $(\mathsf{nad}.\MPK,\mathsf{nad}.\MSK)\la\mathsf{NAD}.\Setup(1^\secp)$ and $(\nce.\pk,\nce.\MSK)\la\NCE.\Setup(1^{\secp})$, and sends $(\nad.\MPK,\nce.\pk)$ to $\cA_1$.
    \item\label{step:key_query_adapt_fe_1} $\cA_1$ is allowed to make an arbitrary key query at most one time.
    For a key query, the challenger receives $f\in\mathcal{F}$,  
    computes $\mathsf{nad}.\sk_{f}\la\mathsf{NAD}.\keygen(\mathsf{nad}.\MSK,f)$
    and $\nce.\sk\la\NCE.\keygen(\nce.\MSK)$,
    and sends $(\nad.\sk_f,\nce.\sk)$ to $\cA_1$.
    \item $\cA_1$ chooses $m\in\Ms$, and sends $m$ to the challenger.
    \item\label{step:encryption_single_fe} The challenger generates $a,c\la\bit^n$,
    computes $(\mathsf{nad}.\vk,\mathsf{nad}.\CT)\la \mathsf{NAD}.\Enc(\mathsf{nad}.\MPK,m)$,
    $\Psi\seteq Z^{c}X^{a}\mathsf{nad}.\ct X^{a}Z^{c}$
    and $(\nce.\vk,\nce.\ct)\la\NCE.\Enc(\nce.\pk,(a,c))$,
    and sends $(\Psi,\nce.\CT)$ to $\cA_1$.
    \item\label{step:key_query_adapt_fe_2}
    If a key query is not called in step~\ref{step:key_query_adapt_fe_1},
    $\cA_1$ is allowed to make an arbitrary key query at most one time.
    For a key query, the challenger receives $f\in\mathcal{F}$,  
    computes $\mathsf{nad}.\sk_{f}\la\mathsf{NAD}.\keygen(\mathsf{nad}.\MSK,f)$
    and $\nce.\sk\la\NCE.\keygen(\nce.\MSK)$,
    and sends $(\nad.\sk_f,\nce.\sk)$ to $\cA_1$.
    \item $\cA_1$ sends $(\nad.\cert,\nce.\cert)$ to the challenger and its internal state to $\cA_2$.
    \item\label{step:vrfy_fe} The challenger computes $\mathsf{nad}.\cert^*\la\NAD.\Modify(a,c,\nad.\cert)$.
        The challenger computes $\NCE.\Vrfy(\allowbreak \nce.\vk,\nce.\cert)$ and $\mathsf{NAD}.\Vrfy(\mathsf{nad}.\vk,\mathsf{nad}.\cert^*)$.
        If the results are $\top$, the challenger outputs $\top$ and sends $(\nad.\MSK,\nce.\MSK)$ to $\cA_2$.
        Otherwise, the challenger outputs $\bot$ and sends $\bot$ to $\cA_2$.
    \item $\cA_2$ outputs $b'$. The output of the experiment is $b'$ if the challenger outputs $\top$.
    Otherwise, the output of the experiment is $\bot$.
\end{enumerate}
\item[$\sfhyb{1}{}$:]$ $
This is different from $\sfhyb{0}{}$ in the following second points.
First, when a key query is not called in step~\ref{step:key_query_adapt_fe_1}, the challenger computes $(\nce.\vk,\widetilde{\nce.\ct},\nce.\aux)\la\NCE.\Fake(\nce.\pk)$ and sends $(\Psi,\widetilde{\nce.\ct})$ to $\cA_1$
instead of computing $(\nce.\vk,\nce.\ct)\la\NCE.\Enc(\nce.\pk,(a,c))$
and sending $(\Psi,\nce.\ct)$ to $\cA_1$.
Second, in step~\ref{step:key_query_adapt_fe_2}, the challenger computes $\widetilde{\nce.\sk}\la\NCE.\Reveal(\nce.\pk,\nce.\MSK,\nce.\aux,(a,c))$ and sends $(\nad.\sk_f,\nce.\sk)$ to $\cA_1$ instead of computing $\nce.\sk\la\NCE.\keygen(\nce.\MSK)$ and sending $(\nad.\sk_f,\nce.\sk)$ to $\cA_1$.

\item[$\sfhyb{2}{}$:]$ $
This is different from $\sfhyb{1}{}$ in the following three points.
First, when a key query is not called in step \ref{step:key_query_adapt_fe_1}, the challenger generates $\ket{\widetilde{0^n0^n}}$ instead of generating $a,c\la\bit^n$ and $\Psi=Z^{c}X^{a}\mathsf{nad}.\ct X^{a}Z^{c}$.
Let $\Psi_A\seteq \Tr_B(\ket{\widetilde{0^n0^n}}\bra{\widetilde{0^n0^n}})$ and $\Psi_B\seteq \Tr_A(\ket{\widetilde{0^n0^n}}\bra{\widetilde{0^n0^n}})$.
Second, when a key query is not called in step \ref{step:key_query_adapt_fe_1}, the challenger sends $(\Psi_A,\widetilde{\nce.\ct})$ to $\cA_1$ instead of sending $(\Psi,\widetilde{\nce.\ct})$ to $\cA_1$ and then that measures the $i$-th qubit of $\nad.\ct$ and $\Psi_B$ in the Bell basis for all $i\in[n]$.
Let $(x_i,z_i)$ be the measurement outcome for all $i\in[n]$.
Third, the challenger computes $\widetilde{\nce.\sk}\la\NCE.\Reveal(\nce.\pk,\nce.\MSK,\nce.\aux,(x,z))$ instead of computing $\widetilde{\nce.\sk}\la\NCE.\Reveal(\nce.\pk,\nce.\MSK,\nce.\aux,(a,c))$ in step~\ref{step:key_query_adapt_fe_2} and computes $\mathsf{nad}.\cert^*\la\NAD.\Modify(x,z,\nad.\cert)$ instead of computing
$\mathsf{nad}.\cert^*\la\NAD.\Modify(a,c,\nad.\cert)$ in step~\ref{step:vrfy_fe}.

\item[$\sfhyb{3}{}$:]$ $
This is different from $\sfhyb{2}{}$ in the following three points.
First, when a key query is not called in step~\ref{step:key_query_adapt_fe_1}, the challenger does not generate $(\mathsf{nad}.\vk,\mathsf{nad}.\CT)\la \mathsf{NAD}.\Enc(\mathsf{nad}.\MPK,m)$ and measure the $i$-th qubit of $\nad.\ct$ and $\Psi_B$ in the Bell basis in step~\ref{step:encryption_single_fe}.
Second, if a key query is called in step~\ref{step:key_query_adapt_fe_2}, the challenger computes $(\mathsf{nad}.\vk,\mathsf{nad}.\CT)\la \mathsf{NAD}.\Enc(\mathsf{nad}.\MPK,m)$
and measures the $i$-th qubit of $\nad.\ct$ and $\Psi_B$ in the Bell basis for all $i\in[n]$ after it computes $\mathsf{nad}.\sk_{f}\la\mathsf{NAD}.\keygen(\mathsf{nad}.\MSK,f)$.
Third, if a key query is not called throughout the experiment,
the challenger computes $(\mathsf{nad}.\vk,\mathsf{nad}.\CT)\la \mathsf{NAD}.\Enc(\mathsf{nad}.\MPK,m)$, measures the $i$-th qubit of $\nad.\ct$ and $\Psi_B$ in the Bell basis after step~\ref{step:key_query_adapt_fe_2}.

\item[$\sfhyb{4}{}$:]$ $
This is identical to $\sfhyb{3}{}$ except that the challenger computes $(\mathsf{nad}.\vk,\mathsf{nad}.\CT)\la \mathsf{NAD}.\Sim(\mathsf{nad}.\MPK,\mathcal{V},1^{|m|})$ instead of computing $(\mathsf{nad}.\vk,\mathsf{nad}.\CT)\la \mathsf{NAD}.\Enc(\mathsf{nad}.\MPK,m)$, where $\mathcal{V}=(f,f(m),\nad.\sk_f)$ if a key query is called and $\mathcal{V}=\emptyset$ if a key query is not called.

\end{description}
From the definition of $\expc{\Sigma_{\mathsf{cefe}},\cA}{cert}{ever}{adapt}(\secp, b)$ and $\Sim=(\Sim_1,\Sim_2,\Sim_3)$,
it is clear that $\Pr[\sfhyb{0}{}=1]=\Pr[\expc{\Sigma_{\mathsf{cefe}},\cA}{cert}{ever}{adapt}(\secp,0)=1]$ and
$\Pr[\sfhyb{4}{}=1]=\Pr[\expc{\Sigma_{\mathsf{cefe}},\cA}{cert}{ever}{adapt}(\secp,1)=1]$.
Therefore, \cref{thm:ever_security_single_ad} easily follows from \cref{prop:exp_hyb_1_cefe_si_ad,prop:hyb_1_hyb_2_cefe_si_ad,prop:hyb_2_hyb_3_cefe_si_ad,prop:hyb_3_hyb_4_cefe_si_ad}.
(Whose proof is given later.)
\end{proof}

\begin{proposition}\label{prop:exp_hyb_1_cefe_si_ad}
If $\Sigma_{\mathsf{cence}}$ is certified everlasting RNC secure, 
it holds that
\begin{align}
\abs{\Pr[\sfhyb{0}{}=1]-\Pr[\sfhyb{1}{}=1]}\leq\negl(\secp).
\end{align}
\end{proposition}

\begin{proposition}\label{prop:hyb_1_hyb_2_cefe_si_ad}
\begin{align}
\Pr[\sfhyb{1}{}=1]=\Pr[\sfhyb{2}{}=1].
\end{align}
\end{proposition}

\begin{proposition}\label{prop:hyb_2_hyb_3_cefe_si_ad}
\begin{align}
    \Pr[\sfhyb{2}{}=1]=\Pr[\sfhyb{3}{}=1].
\end{align}
\end{proposition}

\begin{proposition}\label{prop:hyb_3_hyb_4_cefe_si_ad}
If $\Sigma_{\mathsf{cegc}}$ is certified everlasting selective secure,
it holds that
\begin{align}
    \abs{\Pr[\sfhyb{3}{}=1]-\Pr[\sfhyb{4}{}=1]}\leq\negl(\secp).
\end{align}
\end{proposition}

\begin{proof}[Proof of \cref{prop:exp_hyb_1_cefe_si_ad}]
When an adversary makes key queries in step \ref{step:key_query_adapt_fe_1}, it is clear that $\Pr[\sfhyb{0}{}=1]=\Pr[\sfhyb{1}{}=1]$.
Hence, we consider the case where the adversary does not make a key query in step $\ref{step:key_query_adapt_fe_1}$ below.

We assume that $\abs{\Pr[\sfhyb{0}{}=1]-\Pr[\sfhyb{1}{}=1]}$ is non-negligible, and construct an adversary $\cB$ that breaks the certified everlasting RNC security of $\Sigma_{\mathsf{cence}}$.
Let us describe how $\cB$ works below.
\begin{enumerate}
    \item $\cB$ receives $\nce.\pk$ from the challenger of $\expd{\Sigma_{\mathsf{cence}},\cB}{cert}{ever}{rec}{nc}(\secp,b)$,
    generates $(\mathsf{nad}.\MPK,\mathsf{nad}.\MSK)\la \mathsf{NAD}.\keygen(1^\secp)$,
    and sends $(\nad.\MPK,\nce.\pk)$ to $\cA_1$.
    \item $\cB$ receives a message $m\in\Ms$, computes $(\mathsf{nad}.\vk,\mathsf{nad}.\ct)\la\mathsf{NAD}.\Enc(\mathsf{nad}.\MPK,m)$, generates $a,c\la\bit^n$, computes $\Psi\seteq Z^cX^a\mathsf{nad}.\ct X^aZ^c$, sends $(a,c)$ to the challenger, receives $(\nce.\ct^*,\nce.\sk^*)$ from the challenger, and sends $(\Psi,\nce.\ct^*)$ to $\cA_1$.
    \item $\cA_1$ is allowed to send a key query at most one time.
    For a key query, $\cB$ receives an function $f$, generates $\mathsf{nad}.\sk_f\la\mathsf{NAD}.\keygen(\mathsf{nad}.\MSK,f)$,
    and sends $(\nad.\sk_f,\nce.\sk^*)$ to $\cA_1$.
    \item $\cA_1$ sends $(\nad.\cert,\nce.\cert)$ to $\cB$ and its internal state to $\cA_2$.
    \item $\cB$ sends $\nce.\cert$ to the challenger, and receives $\nce.\MSK$ or $\bot$ from the challenger.
    $\cB$ computes $\mathsf{nad}.\cert^*\la\NAD.\Modify(a,c,\nad.\cert)$ and $\mathsf{NAD}.\Vrfy(\mathsf{nad}.\vk,\mathsf{nad}.\cert^*)$.
    If the result is $\top$ and $\cB$ receives $\nce.\MSK$ from the challenger,
    $\cB$ sends $(\nad.\MSK,\nce.\MSK)$ to $\cA_2$.
    Otherwise, $\cB$ outputs $\bot$, sends $\bot$ to $\cA_2$, and aborts.
    \item $\cA_2$ outputs $b'$.
    \item $\cB$ outputs $b'$. 
\end{enumerate}
It is clear that $\Pr[1\la\cB|b=0]=\Pr[\sfhyb{0}{}=1]$ and $\Pr[1\la\cB|b=1]=\Pr[\sfhyb{1}{}=1]$.
By assumption, $\abs{\Pr[\sfhyb{0}{}=1]-\Pr[\sfhyb{1}{}=1]}$ is non-negligible, and therefore $\allowbreak\abs{\Pr[1\la\cB|b=0]-\allowbreak\Pr[1\la\cB|b=1]}$ is non-negligible,
which contradicts the certified everlasting RNC security of $\Sigma_{\mathsf{cence}}$. 
\end{proof}

\begin{proof}[Proof of \cref{prop:hyb_1_hyb_2_cefe_si_ad}]

We clarify the difference between $\sfhyb{1}{}$ and $\sfhyb{2}{}$.
First, in $\sfhyb{2}{}$, the challenger uses $(x,z)$ instead of using $(a,c)$ as in $\sfhyb{1}{}$.
Second, in $\sfhyb{2}{}$, the challenger sends $\Psi_A$ to $\cA_1$ instead of sending $Z^{c}X^{a}\mathsf{nad}.\ct X^{a}Z^{c}$ to $\cA_1$ as in $\sfhyb{1}{}$.
Hence,
it is sufficient to prove that $x$ and $z$ are uniformly randomly distributed and
$\Psi_A$ is identical to $Z^{z}X^{x}\mathsf{nad}.\ct X^{x}Z^{z}$.
These two things are obvious from \cref{lem:quantum_teleportation}.
\end{proof}

\begin{proof}[Proof of \cref{prop:hyb_2_hyb_3_cefe_si_ad}]
The difference between $\sfhyb{2}{}$ and $\sfhyb{3}{}$ is only the order of operating the algorithm $\NAD.\Enc$ and the Bell measurement on $\nad.\CT$ and $\Psi_B$.
Therefore, it is clear that the probability distribution of the ciphertext and the decryption key given to the adversary in $\sfhyb{2}{}$ is identical to that the ciphertext and the decryption key given to the adversary in $\sfhyb{3}{}$.
\end{proof}

\begin{proof}[Proof of \cref{prop:hyb_3_hyb_4_cefe_si_ad}]
We assume that $\abs{\Pr[\sfhyb{3}{}=1]-\Pr[\sfhyb{4}{}=1]}$ is non-negligible, and construct an adversary $\cB$ that breaks the 1-bounded certified everlasting non-adaptive security of $\Sigma_{\mathsf{nad}}$.
Let us describe how $\cB$ works below.
\begin{enumerate}
    \item $\cB$ receives $\mathsf{nad}.\MPK$ from the challenger of $\expd{\Sigma_{\mathsf{nad}},\cB}{cert}{ever}{non}{adapt}(\secp,b)$, 
    generates $(\nce.\pk,\nce.\MSK)\la\NCE.\Setup(1^\secp)$, and sends $(\nad.\MPK,\nce.\pk)$ to $\cA_1$.
    \item\label{step:key_query_adapt_fe_hyb_3_hyb_4}  $\cA_1$ is allowed to call a key query at most one time.
    For a key query, $\cB$ receives $f$ from $\cA_1$, sends $f$ to the challenger as a key query, receives $\mathsf{nad}.\sk_f$ from the challenger, computes $\nce.\sk\la\NCE.\keygen(\nce.\MSK)$, and sends $(\nad.\sk_f,\nce.\sk)$ to $\cA_1$.
    \item $\cA_1$ chooses $m\in\Ms$ and sends $m$ to $\cB$.
    \item $\cB$ does the following.
    \begin{itemize}
        \item If a key query is called in step~\ref{step:key_query_adapt_fe_hyb_3_hyb_4}, $\cB$ sends a challenge query $m$ to the challenger, receives $\nad.\ct$ from the challenger, generates $a,c\la\bit^n$,
        $\Psi\seteq Z^{c}X^{a}\mathsf{nad}.\ct X^{a}Z^{c}$ and $(\nce.\vk,\nce.\ct)\la\NCE.\Enc(\nce.\pk,(a,c))$,
        and sends $(\Psi,\nce.\CT)$ to $\cA_1$.
        \item If a key query is not called in step~\ref{step:key_query_adapt_fe_hyb_3_hyb_4}, $\cB$ generates $\ket{\widetilde{0^n0^n}}$.
        Let $\Psi_A\seteq \Tr_B(\ket{\widetilde{0^n0^n}}\bra{\widetilde{0^n0^n}})$ and $\Psi_B\seteq \Tr_A(\ket{\widetilde{0^n0^n}}\bra{\widetilde{0^n0^n}})$.
        $\cB$ computes $(\nce.\vk,\widetilde{\nce.\ct},\nce.\aux)\la\NCE.\Fake(\nce.\pk)$ and sends $(\Psi_A,\widetilde{\nce.\ct})$ to $\cA_1$.
    \end{itemize}
    \item
    If a key query is not called in step~\ref{step:key_query_adapt_fe_hyb_3_hyb_4}, 
    $\cA_1$ is allowed to make a key query at most one time.
    If $\cB$ receives an function $f$ as key query, $\cB$ sends $f$ to the challenger as key query, and receives $\mathsf{nad}.\sk_f$ from the challenger.
    $\cB$ sends a challenge query $m$ to the challenger, receives $\nad.\ct$, measures the $i$-th qubit of $\nad.\ct$ and $\Psi_B$ in the Bell basis, and let $(x_i,z_i)$ be the measurement outcome for all $i\in[n]$. 
    $\cB$ computes $\widetilde{\nce.\sk}\la\NCE.\Reveal(\nce.\pk,\nce.\MSK,\nce.\aux,(x,z))$ and sends $(\nad.\sk_f,\widetilde{\nce.\sk})$ to $\cA_1$.
    \item If $\cB$ does not receive a key query throughout the experiment,
    $\cB$ sends a challenge query $m$ to the challenger, receives $\mathsf{nad}.\ct$, and measures the $i$-th qubit of $\mathsf{nad}.\ct$ and $\Psi_B$ in the Bell basis and let $(x_i,z_i)$ be the measurement outcome for all $i\in[n]$.
    \item $\cA_1$ sends $(\nad.\cert,\nce.\cert)$ to $\cB$ and its internal state to $\cA_2$.
    \item $\cB$ computes $\mathsf{nad}.\cert^*\la\NAD.\Modify(x^*,z^*,\nad.\cert)$, where $(x^*,z^*)=(a,c)$ if a key query is called in step~\ref{step:key_query_adapt_fe_hyb_3_hyb_4} and $(x^*,z^*)=(x,z)$ if a key query is not called in step~\ref{step:key_query_adapt_fe_hyb_3_hyb_4}.
        $\cB$ sends $\mathsf{nad}.\cert$ to the challenger, and receives $\mathsf{nad}.\MSK$ or $\bot$ from the challenger.
        $\cB$ computes $\NCE.\Vrfy(\nce.\vk,\nce.\cert)$.
        If the result is $\top$ and $\cB$ receives $\mathsf{nad}.\MSK$ from the challenger, $\cB$ sends $(\nad.\MSK,\nce.\MSK)$ to $\cA_2$.
        Otherwise, $\cB$ outputs $\bot$, sends $\bot$ to $\cA_2$, and aborts.
    \item $\cA_2$ outputs $b'$.
    \item $\cB$ outputs $b'$. 
\end{enumerate}

It is clear that $\Pr[1\la\cB|b=0]=\Pr[\sfhyb{3}{}=1]$ and $\Pr[1\la\cB|b=1]=\Pr[\sfhyb{4}{}=1]$.
By assumption, $\abs{\Pr[\sfhyb{3}{}=1]-\Pr[\sfhyb{4}{}=1]}$ is non-negligible, and therefore $\abs{\Pr[1\la\cB|b=0]-\Pr[1\la\cB|b=1]}$ is non-negligible, which contradicts the 1-bounded certified everlasting non-adaptive security of $\Sigma_{\mathsf{nad}}$. 
\end{proof}


\section{Proof of \cref{thm:ever_bounded_fe}}\label{Sec:Proof_bounded_fe}
\begin{proof}[Proof of \cref{thm:ever_bounded_fe}]
Let us denote the simulating algorithm of $\Sigma_{\one}$ as $\ONE.\Sim=\ONE.(\Sim_1,\Sim_2,\Sim_3)$.
Let us describe how the simulator $\Sim=(\Sim_1,\Sim_2,\Sim_3)$ works below.
\begin{description}
\item[$\Sim_1(\MPK,\mathcal{V},1^{|x|})$:] Let $q^*$ be the number of times that $\cA_1$ has made key queries before it sends a challenge query. 
\begin{enumerate}
    \item Parse $\MPK\seteq \{\one.\MPK_{i}\}_{i\in[N]}$ and $\mathcal{V}\seteq \{C_j,C_j(x),(\Gamma_j,\Delta_j,\{\one.\sk_{C_j,\Delta_j,i}\}_{i\in[\Gamma_j]})\}_{j\in[q^*]}$.
    \item Generate a uniformly random set $\Gamma_i\subseteq[N]$ of size $Dt+1$ and a uniformly random set $\Delta_i\subseteq[S]$ of size $v$ for all $i\in\{q^*+1,\cdots, q\}$.
    Let $\Delta_0\seteq \emptyset$. Let $\mathcal{L}\seteq \bigcup_{i\neq i'}(\Gamma_i\cap \Gamma_{i'})$.
    $\Sim_1$ aborts if $\left|\mathcal{L}\right|> t$ or
    there exists some $i\in[q]$ such that $\Delta_i \setminus (\bigcup_{j\neq i}\Delta_j)=\emptyset$.
    \item $\Sim_1$ uniformly and independently samples $\ell$ random degree $t$ polynomials $\mu_1,\cdots, \mu_\ell$ whose constant terms are all $0$. 
    \item $\Sim_1$ samples the polynomials $\xi_1,\cdots, \xi_S$ as follows for $j\in[q]$:
    \begin{itemize}
        \item fix $a^*\in\Delta_j \setminus (\Delta_0\cup\cdots\cup\Delta_{j-1})$;
        \item for all $a\in(\Delta_j \setminus (\Delta_0\cup\cdots\cup \Delta_{j-1})) \setminus \{a^*\}$, set $\xi_a$ to be a uniformly random degree $Dt$ polynomial whose constant term is $0$;
        \item if $j\leq q^*$, pick a random degree $Dt$ polynomial $\eta_j(\cdot)$ whose constant term is $C_j(x)$; if $j>q^*$, pick random values for $\eta_j(i)$ for all $i\in\mathcal{L}$;
        \item the evaluation of $\xi_{a^*}$ on the points in $\mathcal{L}$ is defined by the relation:
        \begin{align}
            \eta_j(\cdot)=C_j(\mu_1(\cdot),\cdots, \mu_\ell(\cdot))+\sum_{a\in\Delta_j}\xi_a(\cdot).
        \end{align}
        \item Finally, for all $a\notin (\Delta_1\cup\cdots \cup \Delta_q)$, set $\xi_a$ to be a uniformly random degree $Dt$ polynomial whose constant term is $0$.
    \end{itemize}
    \item For each $i\in\mathcal{L}$, $\Sim_1$ computes 
    \begin{align}
        (\one.\vk_i,\one.\ct_i)\la\ONE.\Enc(\one.\MPK_i,(\mu_1(i),\cdots,\mu_\ell(i),\xi_1(i),\cdots, \xi_S(i) )).
    \end{align}
    \item For each $i\notin \mathcal{L}$, $\Sim_1$ does the following:
    \begin{itemize}
        \item If $i\in \Gamma_j$ for some $j\in[q^*]$
        \footnote
        { Note that $j$ is uniquely determined since $i\notin \mathcal{L}$ .
        },
        computes
        \begin{align}
        (\one.\ct_i,\one.\state_i)\la\ONE.\Sim_1(\one.\MPK_i,(G_{C_j,\Delta_j,i},\eta_j(i),\one.\sk_{C_j,\Delta_j,i}),1^{|m|}).
        \end{align}
        \item If $i\notin\Gamma_j$ for all $j\in[q^*]$, 
        computes 
        \begin{align}
            (\one.\ct_i,\one.\state_i)\la\ONE.\Sim_1(\one.\MPK_i,\emptyset,1^{|m|}).
        \end{align}
    \end{itemize}
    \item Output $\ct\seteq \{\one.\ct_i\}_{i\in[N]}$ and $\state\seteq (\{\Gamma_i\}_{i\in[q]},\{\Delta_i\}_{i\in[q]},\{\eta_j(i)\}_{j\in\{q^*+1,\cdots,q\},i\in\mathcal{L}},\{\one.\state_i\}_{i\in[N] \setminus \mathcal{L}},\allowbreak \{\one.\vk_i\}_{i\in\mathcal{L}})$.
\end{enumerate}
\item[$\Sim_2(\MSK,C_j,C_j(x),\state)$:]The simulator simulates the $j$-th key query for $j>q^*$.
\begin{enumerate}
    \item Parse $\MSK\seteq\{\one.\MSK_i\}_{i\in[N]}$ and $\state_{j-1} \seteq  (\{\Gamma_i\}_{i\in[q]},\{\Delta_i\}_{i\in[q]},\{\eta_s(i)\}_{s\in\{q^*+1,\cdots,q\},i\in\mathcal{L}},\{\one.\state_i\}_{i\in[N] \setminus \mathcal{L}},\allowbreak \{\one.\vk_i\}_{i\in\mathcal{L}})$.
    \item For each $i\in\Gamma_j\cap \mathcal{L}$, generate $\one.\sk_{C_j,\Delta_j,i}\la\ONE.\keygen(\one.\MSK_i,G_{C_j,\Delta_j})$.
    \item For each $i\in\Gamma_j \setminus \mathcal{L}$, 
    generate a random degree $Dt$ polynomial $\eta_j(\cdot)$ whose constant term is $C_j(x)$ and subject to the constraints on the values in $\mathcal{L}$ chosen earlier,
    and generate 
    \begin{align}
    (\one.\sk_{C_j,\Delta_j,i},\one.\state_i^*)\la\ONE.\Sim_2(\one.\MSK_i,\eta_j(i),G_{C_j,\Delta_j},\one.\state_i).
    \end{align}
    For simplicity, let us denote $\one.\state_i^*$ as $\one.\state_i$ for $i\in\Gamma_j \setminus \mathcal{L}$.
    \item Output $\sk_{C_j}\seteq(\Gamma_j,\Delta_j,\{\one.\sk_{C_j,\Delta_j,i}\}_{i\in\Gamma_j})$ and
    $\state_j\seteq (\{\Gamma_i\}_{i\in[q]},\{\Delta_i\}_{i\in[q]},\{\eta_j(i)\}_{j\in\{q^*+1,\cdots,q\},i\in\mathcal{L}},\allowbreak \{\one.\state_i\}_{i\in[N]\setminus \mathcal{L}},\{\one.\vk_i\}_{i\in\mathcal{L}})$.
\end{enumerate}
\item[$\Sim_3(\state^*)$:]The simulator simulates a verification key.
\begin{enumerate}
    \item Parse $\state^*\seteq (\{\Gamma_i\}_{i\in[q]},\{\Delta_i\}_{i\in[q]},\{\eta_j(i)\}_{j\in\{q^*+1,\cdots,q\},i\in\mathcal{L}}, \{\one.\state_i\}_{i\in[N] \setminus \mathcal{L}},\{\one.\vk_i\}_{i\in\mathcal{L}})$.
    \item For each $i\in [N] \setminus \mathcal{L}$, compute $\one.\vk_i\la\ONE.\Sim_3(\one.\state_i)$.
    \item Output $\vk\seteq\{\one.\vk_i\}_{i\in[N]}$.
\end{enumerate}
Let us define the sequence of hybrids as follows.

\item [$\sfhyb{0}{}$:] This is identical to $\expc{\Sigma_{\mathsf{cefe}},\cA}{cert}{ever}{adapt}(\secp, 0)$.
\begin{enumerate}
    \item\label{step:keygen_bounded_fe} The challenger generates $(\one.\MPK_i,\one.\MSK_i)\la\ONE.\Setup(1^\secp)$ for $i\in[N]$.
    \item $\cA_1$ is allowed to call key queries at most $q$ times.
    For the $j$-th key query, the challenger receives an function $C_{j}$ from $\cA_1$, generates a uniformly random set $\Gamma_{j}\in[N]$ of size $Dt+1$ and $\Delta_{j}\in [S]$ of size $v$.
    For $i\in \Gamma_{j}$, the challenger generates $\mathsf{one}.\sk_{C_{j},\Delta_{j},i}\la\mathsf{ONE}.\keygen(\one.\MSK_i,G_{C_{j},\Delta_{j}})$, and sends $(\Gamma_{j},\Delta_{j},\{\one.\sk_{C_{j},\Delta_{j},i}\}_{i\in\Gamma_{j}})$ to $\cA_1$.
    Let $q^*$ be the number of times that $\cA_1$ has called key queries in this step.
    \item $\cA_1$ chooses $x\in\Ms$ and sends $x$ to the challenger.
    \item\label{step:enc_bounded_fe} The challenger generates a random degree $t$ polynomial $\mu_i(\cdot)$ whose constant term is $x[i]$ for $i\in[\ell]$
    and a random degree $Dt$ polynomial $\xi_i(\cdot)$ whose constant term is $0$.
    For $i\in[N]$, the challenger computes $(\one.\vk_i,\one.\ct_i)\la\ONE.\Enc(\one.\MPK_i,(\mu_1(i),\cdots,\mu_{\ell}(i),\xi_1(i),\cdots,\xi_S(i)))$,
    and sends $\{\one.\ct_i\}_{i\in[N]}$ to $\cA_1$.
    \item\label{step:key_bounded_fe_2} $\cA_1$ is allowed to call a key query at most $q-q^*$ times.
    For the $j$-th key query, the challenger receives an function $C_{j}$ from $\cA_1$, generates a uniformly random set $\Gamma_{j}\in[N]$ of size $Dt+1$ and $\Delta_{j}\in [S]$ of size $v$.
    For $i\in \Gamma_{j}$, the challenger generates $\mathsf{one}.\sk_{C_{j},\Delta_{j},i}\la\mathsf{ONE}.\keygen(\one.\MSK_i,G_{C_{j},\Delta_{j}})$, and sends $(\Gamma_{j},\Delta_{j},\{\one.\sk_{C_{j},\Delta_{j},i}\}_{i\in\Gamma_{j}})$ to $\cA_1$.
    \item $\cA_1$ sends $\{\one.\cert_i\}_{i\in[N]}$ to the challenger and its internal state to $\cA_2$.
    \item If $\top\la \ONE.\Vrfy(\one.\vk_i,\one.\cert_i)$ for all $i\in[N]$, the challenger outputs $\top$ and sends $\{\one.\MSK_i\}_{i\in[N]}$ to $\cA_2$.
    Otherwise, the challenger outputs $\bot$ and sends $\bot$ to $\cA_2$.
    \item $\cA_2$ outputs $b$.
    \item The experiment outputs $b$ if the challenger outputs $\top$.
    Otherwise, the experiment outputs $\bot$.
\end{enumerate}
\item[$\sfhyb{1}{}$:] This is identical to $\sfhyb{0}{}$ except for the following three points.
First, the challenger generates uniformly random set $\Gamma_i\in[N]$ of size $Dt+1$ and $\Delta_{i}\in [S]$ of size $v$ for $i\in\{q^*+1,\cdots,q\}$ in step \ref{step:enc_bounded_fe} instead of generating them when a key query is called.
Second, if $\left|\mathcal{L}\right|> t$, the challenger aborts and the experiment outputs $\bot$.
Third, if there exists some $i\in[q]$ such that $\Delta_i \setminus (\bigcup_{j\neq i}\Delta_j)=\emptyset$, the challenger aborts and the experiment outputs $\bot$.

\item[$\sfhyb{2}{}$:]This is identical to $\sfhyb{1}{}$ except that the challenger samples $\xi_1,\cdots,\xi_S, \eta_1,\cdots,\eta_q$ as in the simulator $\Sim_1$ described above.

\item[$\sfhyb{3}{}$:]
This is identical to $\sfhyb{2}{}$ except that the challenger generates $\{\one.\ct_i\}_{i\in[N] \setminus \{\mathcal{L}\}}$,
$\{\one.\sk_{C_j,\Delta_j,i}\}_{i\in\Gamma_j}$ for $j\in\{q^*+1,\cdots ,q'\}$,
and $\vk\seteq\{\one.\vk_i\}_{i\in[N] \setminus \{\mathcal{L}\}}$ as in the simulator $\Sim=(\Sim_1,\Sim_2,\Sim_3)$ described above,
where $q'$ is the number of key queries that the adversary makes in total.

\item[$\sfhyb{4}{}$:]
This is identical to $\sfhyb{3}{}$ except that the challenger generates $\mu_1,\cdots, \mu_\ell$ as in the simulator $\Sim_1$ described above.
\end{description}

From the definition of $\expc{\Sigma_{\mathsf{cefe}},\cA}{cert}{ever}{adapt}(\secp, b)$ and $\Sim=(\Sim_1,\Sim_2,\Sim_3)$,
it is clear that $\Pr[\sfhyb{0}{}=1]=\Pr[\expc{\Sigma_{\mathsf{cefe}},\cA}{cert}{ever}{adapt}(\secp,0)=1]$ and
$\Pr[\sfhyb{4}{}=1]=\Pr[\expc{\Sigma_{\mathsf{cefe}},\cA}{cert}{ever}{adapt}(\secp,1)=1]$.
Therefore, \cref{thm:ever_bounded_fe} easily follows from \cref{prop:hyb_0_hyb_1_bounded_fe,prop:hyb_1_hyb_2_bounded_fe,prop:hyb_2_hyb_3_bounded_fe,prop:hyb_3_hyb_4_bounded_fe} (whose proofs are given later).
\end{proof}

\begin{proposition}\label{prop:hyb_0_hyb_1_bounded_fe}
\begin{align}
    \abs{\Pr[\sfhyb{0}{}=1]-\Pr[\sfhyb{1}{}=1]}\leq \negl(\lambda).
\end{align}
\end{proposition}

\begin{proposition}\label{prop:hyb_1_hyb_2_bounded_fe}
\begin{align}
    \Pr[\sfhyb{1}{}=1]=\Pr[\sfhyb{2}{}=1].
\end{align}
\end{proposition}

\begin{proposition}\label{prop:hyb_2_hyb_3_bounded_fe}
If $\Sigma_{\one}$ is 1-bounded certified everlasting adaptive secure,
\begin{align}
    \abs{\Pr[\sfhyb{2}{}=1]-\Pr[\sfhyb{3}{}=1]}\leq \negl(\lambda).
\end{align}
\end{proposition}

\begin{proposition}\label{prop:hyb_3_hyb_4_bounded_fe}
\begin{align}
    \Pr[\sfhyb{3}{}=1]=\Pr[\sfhyb{4}{}=1].
\end{align}
\end{proposition}

\begin{proof}[Proof of \cref{prop:hyb_0_hyb_1_bounded_fe}]
Let $\sfhyb{0}{'}$ be the experiment identical to $\sfhyb{0}{}$ except that the challenger generates a set $\Gamma_i\in[N]$ and $\Delta_{i}\in [S]$ for $i\in\{q^*+1,\cdots ,q\}$
in step~\ref{step:enc_bounded_fe}.
It is clear that $\Pr[\sfhyb{0}{}=1]=\Pr[\sfhyb{0}{'}=1]$. 

Let $\sfhyb{0}{*}$ be the experiment identical to $\sfhyb{0}{'}$ except that it outputs $\bot$ if $|\mathcal{L}|>t$.
It is clear that $\Pr[\sfhyb{0}{'}=1\wedge (|\mathcal{L}|\leq t)]=\Pr[\sfhyb{0}{*}=1\wedge (|\mathcal{L}|\leq t)]$.
Hence, it holds that
\begin{align}
    \abs{\Pr[\sfhyb{0}{'}=1]-\Pr[\sfhyb{0}{*}=1]}\leq \Pr[|\mathcal{L}|>t]
\end{align}
from \cref{lem:defference}.

Let $\mathsf{Collide}$ be the event that
there exists some $i\in[q]$ such that $\Delta_i \setminus (\bigcup_{j\neq i}\Delta_j)=\emptyset$.
$\sfhyb{0}{*}$ is identical to $\sfhyb{1}{}$ when $\mathsf{Collide}$ does not occur.
Hence, it is clear that $\Pr[\sfhyb{0}{*}=1\wedge\overline{\mathsf{Collide}}]=\Pr[\sfhyb{1}{}=1\wedge \overline{\mathsf{Collide}}]$.
Therefore, it holds that 
\begin{align}
    \abs{\Pr[\sfhyb{0}{*}=1]-\Pr[\sfhyb{1}{}=1]}\leq \Pr[\mathsf{Collide}]
\end{align}
from \cref{lem:defference}.

From the discussion above, we have 
\begin{align}
    \abs{\Pr[\sfhyb{0}{}=1]-\Pr[\sfhyb{1}{}=1]}\leq \Pr[|\mathcal{L}|>t]+\Pr[\mathsf{Collide}].
\end{align}
The following \cref{lem:small_pair,lem:coverfree} shows that 
$\Pr[|\mathcal{L}|>t]\leq 2^{-\Omega(\lambda)}$
and $\Pr[\mathsf{Collide}]\leq q2^{-\Omega(\lambda)}$, which completes the proof.
\end{proof}

\begin{lemma}[\cite{C:GorVaiWee12}]\label{lem:small_pair}
Let $\Gamma_1,\cdots ,\Gamma_q\subseteq [N]$ be randomly chosen subsets of size $tD+1$.
Let $t=\Theta(q^2\lambda)$ and $N=\Theta(D^2q^2 t)$.
Then, 
\begin{align}
    \Pr\left[\left|\bigcup_{i\neq i'}(\Gamma_i\cap \Gamma_i)\right|>t\right]\leq2^{-\Omega(\lambda)}
\end{align}
where the probability is over the random choice of the subsets $\Gamma_1,\cdots,\Gamma_q$.
\end{lemma}
\begin{lemma}[\cite{C:GorVaiWee12}]\label{lem:coverfree}
Let $\Delta_1,\cdots,\Delta_q\subseteq[S]$ be randomly chosen subsets of size $v$.
Let $v(\lambda)=\Theta(\lambda)$ and $S(\lambda)=\Theta(vq^2)$.
Let $\mathsf{Collide}$ be the event that there exists some $i\in[q]$ such that $\Delta_i\backslash(\bigcup_{j\neq i}\Delta_j)=\emptyset$.
Then, we have
\begin{align}
    \Pr\left[\mathsf{Collide}\right]\leq q2^{-\Omega(\lambda)}
\end{align}
where the probability is over the random choice of subsets $\Delta_1,\cdots ,\Delta_q$.
\end{lemma}

\begin{proof}[Proof of \cref{prop:hyb_1_hyb_2_bounded_fe}]
In the encryption in $\sfhyb{1}{}$, $\xi_{a^*}$ is chosen at random and $\eta_j(\cdot)$ is defined by the relation.
$\Sim$ essentially chooses $\eta_j(\cdot)$ at random which defines $\xi_{a^*}$.
It is easy to see that reversing the order of how the polynomials are chosen produces the same distribution.
\end{proof}

\begin{proof}[Proof of \cref{prop:hyb_2_hyb_3_bounded_fe}]
To prove the proposition, let us define a hybrid experiment $\sfhyb{2}{s}$ for each $s\in[N]$ as follows.
\begin{description}
\item[$\sfhyb{2}{s}$:] This is identical to $\sfhyb{2}{}$ except for the following three points.
First, the challenger generates $\{\one.\ct_i\}_{i\in[s]\setminus\mathcal{L}}$ as in the simulator $\Sim_1$.
Second, the challenger generates $\{\one.\sk_{C_j,\Delta_j,i}\}_{i\in\Gamma_j\cap[s]}$ for $j\in\{q^*+1,\cdots ,q'\}$ as in the simulator $\Sim_2$,
where $q'$ is the number of key queries that the adversary makes in total.
Third, the challenger generates $\{\one.\vk_i\}_{i\in[s]\setminus\mathcal{L}}$ as in the simulator $\Sim_3$.
\end{description}
Let us denote $\sfhyb{2}{}$ as $\sfhyb{2}{0}$.
It is clear that $\Pr[\sfhyb{2}{N}=1]=\Pr[\sfhyb{3}{}=1]$.
Furthermore, we can show that
\begin{align}
    \abs{\Pr[\sfhyb{2}{s-1}=1]-\Pr[\sfhyb{2}{s}=1]}\leq \negl(\lambda)
\end{align}
for $s\in[N]$. (Its proof is given later.)
From these facts, we obtain \cref{prop:hyb_2_hyb_3_bounded_fe}.

Let us show the remaining one.
In the case $s\in\mathcal{L}$, it is clear that $\sfhyb{2}{s-1}$ is identical to $\sfhyb{2}{s}$.
Hence, we consider the case $s\notin\mathcal{L}$.
To show the inequality above, let us assume that $\abs{\Pr[\sfhyb{2}{s-1}=1]-\Pr[\sfhyb{2}{s}=1]}$ is non-negligible.
Then, we can construct an adversary $\cB$ that can break the 1-bounded certified everlasting adaptive security of $\Sigma_{\one}$ as follows.

\begin{enumerate}
    \item $\cB$ receives $\one.\MPK$ from the challenger of $\expc{\Sigma_{\mathsf{one}},\cA}{cert}{ever}{adapt}(\secp,b)$.
    $\cB$ sets $\one.\MPK_s\seteq \one.\MPK$.
    \item $\cB$ generates $(\one.\MPK_i,\one.\MSK_i)\la\ONE.\Setup(1^\secp)$ for all $i\in[N] \setminus s$, and sends $\{\one.\MPK_i\}_{i\in[N]}$ to $\cA_1$.
    \item $\cA_1$ is allowed to call key queries at most $q$ times.
    For the $j$-th key query, $\cB$ receives an function $C_{j}$ from $\cA_1$, generates a uniformly random set $\Gamma_{j}\in[N]$ of size $Dt+1$ and $\Delta_{j}\in [S]$ of size $v$.
    For $i\in \Gamma_{j} \setminus s$, $\cB$ generates $\mathsf{one}.\sk_{C_{j},\Delta_{j},i}\la\mathsf{ONE}.\keygen(\one.\MSK_i,G_{C_{j},\Delta_{j}})$.
    If $s\in\Gamma_j$, $\cB$ sends $G_{C_j,\Delta_j}$ to the challenger,
    receives $\one.\sk_{C_j,\Delta_j,s}$ from the challenger,
    and sends $(\Gamma_{j},\Delta_{j},\{\one.\sk_{C_{j},\Delta_{j},i}\}_{i\in\Gamma_{j}})$ to $\cA_1$.
    Let $q^*$ be the number of times that $\cA_1$ has called key queries in this step.
    \item $\cA_1$ chooses $x\in\Ms$, and sends $x$ to $\cB$.
    \item\label{step:enc_bounded_fe_hyb_2} $\cB$ generates uniformly random set $\Gamma_i\in[N]$ of size $Dt+1$ and $\Delta_{i}\in [S]$ of size $v$ for $i\in\{q^*+1,\cdots,q\}$.
    $\cB$ generates a random degree $t$ polynomial $\mu_i(\cdot)$ whose constant term is $x[i]$ for $i\in[\ell]$, and $\xi_1,\cdots,\xi_S, \eta_1,\cdots,\eta_q$ as in the simulator $\Sim_1$.
    For $i\in[s-1] \setminus \mathcal{L}$, $\cB$ generates $\one.\ct_i$ as in the simulator $\Sim_1$.
    For $i\in\{s+1,\cdots N\}\cup \mathcal{L}$, $\cB$ generates $(\one.\vk_i,\one.\ct_i)\la \ONE.\Enc(\one.\MPK_i, (\mu_1(i),\cdots,\mu_\ell(i),\xi_1(i),\cdots, \xi_S(i)))$.
    $\cB$ sends $\mu_1(s),\cdots,\mu_\ell(s),\xi_1(s),\cdots ,\xi_S(s)$ to the challenger, and receives $\one.\ct_s$ from the challenger. 
    $\cB$ sends $\{\one.\ct_i\}_{i\in[N]}$ to $\cA_1$.
    \item $\cA_1$ is allowed to call key queries at most $q-q^*$ times.
    For the $j$-th key query, $\cB$ receives an function $C_j$ from $\cA_1$.
    For $i\in\Gamma_j \setminus [s]$, $\cB$ generates $\one.\sk_{C_j,\Delta_j,i}\la\ONE.\keygen(\one.\MSK_i,G_{C_j,\Delta_j})$.
    For $i\in \Gamma_j\wedge[s-1]$, $\cB$ generates $\one.\sk_{C_j,\Delta_j,i}$ as in the simulator $\Sim_2$.
    If $s\in\Gamma_j$, $\cB$ sends $G_{C_j,\Delta_j}$ to the challenger, and receives $\one.\sk_{C_j,\Delta_j,s}$ from the challenger.
    $\cB$ sends $(\Gamma_j,\Delta_j,\{\one.\sk_{C_j,\Delta_j,i}\}_{i\in \Gamma_j})$ to $\cA_1$.
    \item For $i\in[s-1] \setminus \mathcal{L}$, $\cB$ generates $\one.\vk_i$ as in the simulator $\Sim_3$
    \footnote{
    For $i\in\{s+1,\cdots N\}\cup \mathcal{L}$, $\cB$ generated $\one.\vk_i$ in step~\ref{step:enc_bounded_fe_hyb_2}.
    }.
    \item $\cA_1$ sends $\{\one.\cert_i\}_{i\in[N]}$ to $\cB$ and its internal state to $\cA_2$.
    \item $\cB$ sends $\one.\cert_s$ to the challenger, and receives $\one.\MSK_s$ or $\bot$ from the challenger.
    $\cB$ computes $\ONE.\Vrfy(\one.\vk_i,\allowbreak \one.\cert_i)$ for all $i\in[N] \setminus s$.
    If the results are $\top$ and $\cB$ receives $\one.\MSK_s$ from the challenger, $\cB$ sends $\{\one.\MSK_i\}_{i\in[N]}$ to $\cA_2$.
    Otherwise, $\cB$ aborts.
    \item $\cA_2$ outputs $b'$.
    \item $\cB$ outputs $b'$.
\end{enumerate}
It is clear that $\Pr[1\la\cB|b=0]=\Pr[\sfhyb{2}{s-1}=1]$ and $\Pr[1\la\cB|b=1]=\Pr[\sfhyb{2}{s}=1]$.
By assumption, $\abs{\Pr[\sfhyb{2}{s-1}=1]-\Pr[\sfhyb{2}{s}=1]}$ is non-negligible, and therefore $\abs{\Pr[1\la\cB|b=0]-\Pr[1\la\cB|b=1]}$ is non-negligible, which contradicts the 1-bounded certified everlasting adaptive security of $\Sigma_{\one}$.
\end{proof}

\begin{proof}[Proof of \cref{prop:hyb_3_hyb_4_bounded_fe}]
In $\sfhyb{3}{}$, the polynomials $\mu_1,\cdots ,\mu_\ell$ are chosen with constant terms $x_1,\cdots ,x_\ell$, respectively.
In $\sfhyb{4}{}$, these polynomials are now chosen with $0$ constant terms.
This only affects the distribution of $\mu_1,\cdots, \mu_\ell$ themselves and polynomials $\xi_1,\cdots,\xi_S$.
Moreover, only the evaluations of these polynomials on the points in $\mathcal{L}$ affect the outputs of the experiments.
Now observe that:
\begin{itemize}
    \item The distribution of the values $\{\mu_1(i),\cdots, \mu_\ell(i)\}_{i\in\mathcal{L}}$ are identical to both $\sfhyb{3}{}$ and $\sfhyb{4}{}$.
    This is because in both experiments, we choose these polynomials to be random degree $t$ polynomials (with different constraints in the constant term),
    so their evaluation on the points in $\mathcal{L}$ are identically distributed, since $|\mathcal{L}|\leq t$.
    \item The values $\{\xi_1(i),\cdots ,\xi_S(i)\}_{i\in\mathcal{L}}$ depend only on the values $\{\mu_1(i),\cdots,\mu_{\ell}(i)\}_{i\in\mathcal{L}}$.
\end{itemize}
\cref{prop:hyb_3_hyb_4_bounded_fe} follows from these observations.
\end{proof}

\fi

\ifnum\cameraready=1
\else
\ifnum\submission=1
\newpage
\setcounter{tocdepth}{1}
\tableofcontents
\else
\fi
\fi

\end{document}